\documentclass[a4paper]{article}
\usepackage[utf8]{inputenc}
\usepackage{amsmath}
\usepackage{amsfonts}
\usepackage{amssymb}
\usepackage{amsthm}
\usepackage{hyperref}
\usepackage{bm}
\usepackage{bbm}
\usepackage{xcolor}
\usepackage{graphicx}
\usepackage[top=1.2in,
  bottom=1.2in,
  left=1in,
  right=1in]{geometry}

\usepackage{caption}
\usepackage{subcaption}

\usepackage{multirow}
\allowdisplaybreaks

\usepackage{caption}
\numberwithin{equation}{section}
\newtheorem{theorem}{Theorem}[section]
\newtheorem{remark}[theorem]{Remark}
\newtheorem{lemma}[theorem]{Lemma}

\newtheorem{proposition}[theorem]{Proposition}
\newtheorem{corollary}[theorem]{Corollary}
\newtheorem{notation}[theorem]{Notation}

\theoremstyle{definition}
\newtheorem{example}[theorem]{Example}

\newtheorem{assumption}{Assumption}

\DeclareMathOperator{\E}{\mathbb E}
\DeclareMathOperator{\K}{\mathcal K}
\DeclareMathOperator{\Lc}{\mathcal L}
\DeclareMathOperator\supp{supp}

\title{Option pricing in Sandwiched Volterra Volatility model}

\author{ Giulia Di Nunno$^{1,2}$\\\href{mailto:giulian@math.uio.no}{giulian@math.uio.no}
   \and  Yuliya Mishura$^{3,4}$\\\href{mailto:yuliyamishura@knu.ua}{yuliyamishura@knu.ua}
   \and Anton Yurchenko-Tytarenko$^{1,}$\footnote{Corresponding author.} \\ \href{mailto:antony@math.uio.no}{antony@math.uio.no}}

\date{%
    $^1$Department of Mathematics, University of Oslo\\[2ex]%
    $^2$Department of Business and Management Science, NHH Norwegian School of Economics, Bergen\\[2ex]%
    $^3$Department of Probability Theory, Statistics and Actuarial Mathematics, Taras Shevchenko National University of Kyiv\\[2ex]
    $^4$Division of Mathematics and Physics, Mälardalen University\\[2ex]
    July 9, 2024
}

\begin{document}
\maketitle

\begin{abstract}
    We introduce a new model of financial market with stochastic volatility driven by an arbitrary Hölder continuous Gaussian Volterra process. The distinguishing feature of the model is the form of the volatility equation which ensures the solution to be ``sandwiched'' between two arbitrary Hölder continuous functions chosen in advance. We discuss the structure of local martingale measures on this market, investigate integrability and Malliavin differentiability of prices and volatilities as well as study absolute continuity of the corresponding probability laws. Additionally, we utilize Malliavin calculus to develop an algorithm of pricing options with discontinuous payoffs.   
\end{abstract}

\textbf{Funding.} The present research is carried out within the frame and support of the ToppForsk project nr. 274410 of the Research Council of Norway with the title STORM: Stochastics for Time-Space Risk Models. The second author is supported by The Swedish Foundation for Strategic Research, grant Nr. UKR22-0017, and by Japan Science and Technology Agency CREST, project reference number JPMJCR2115.\\
\textbf{Acknowledgements.} We thank Dilip Madan for his recommendations on choosing the bounds for volatility and Åsmund Hausken Sande for the nice discussion on the numerical part.

\section{Introduction}\label{sec: introduction}

\subsection{Motivation and background}

Despite its fundamental role in the development of mathematical finance, the classic Black-Scholes approach does not reflect numerous phenomena observed in real financial markets. In particular, numerous empirical arguments clearly indicate that the volatility is far from being constant: for example, it is well known that the amplitude of variation in stock prices is prone to clustering \cite{Cont2001, Cont2006} and there is a notable negative correlation between variance and returns of an asset \cite{Black_1976, Cheung_Ng_1992, Christie_1982,  Duffee_1995}. However, the most prominent argument against the constant volatility is the behavior of the empirically observed Black-Scholes implied volatility surface $(T, \kappa) \mapsto \widehat \sigma_{\text{emp}} (T, \kappa)$ with $T$ denoting the time to maturity, $\kappa := \log\frac{K}{e^{rT} S_0}$ being the \textit{log-moneyness}, $K$ the strike, $S_0$ the current price of an underlying asset and $r$ the interest rate. Contrary to the prediction of the classical Black-Scholes framework, $\widehat \sigma_{\text{emp}} (T, \kappa)$ is not flat, heavily varies in both variables and, in addition, $\kappa \mapsto \widehat  \sigma_{\text{emp}} (T, \kappa)$ produces convex ``\textit{smiley}'' patterns for any fixed $T$ \cite{ContFonseca2002-1, ContFonseca2002-2, Das_Sundaram_1999, Derman_Miller_Park_2016, Gatheral2006}. 

One of the common ways to modify the vanilla log-normal framework so to take such effects into account is to model the instantaneous volatility as a separate stochastic process itself. It is safe to say that this idea has evolved into a full-fledged area of research on its own, see e.g. books \cite{Andersen_Davis_Kreiss_Mikosch_2009, Fouque2000, Kahl2008, KnightSatchell2007} as well as review articles \cite{Di_Nunno_Kubilius_Mishura_Yurchenko-Tytarenko_2023, Duong_Swanson_2011, Shephard_Andersen_2009, Skiadopoulos_2001}. The first models in continuous time \cite{Heston1993, HullWhite1987, Wiggins_1987} based on Brownian diffusions indeed had the ability to reproduce the ``smiley'' patterns of implied volatility (see e.g. \cite{Renault_Touzi_1996} or \cite[Section 2.8.2]{Fouque2000}) but, as indicated in \cite{Gatheral2006}, the accuracy of reproduction was imperfect. For example, as noted in \cite{ComteRenault1998}, the smile generated by most classical Brownian models levels out too quickly in comparison to actual implied volatilities. To address this issue, Comte and Renault \cite{ComteRenault1998} proposed to model volatility via fractional Brownian motion (fBm) with Hurst index $H>1/2$. This process allows for mimicking the behavior of the volatility smile amplitude for longer maturities (see also a simulation study \cite{Funahashi_Kijima_2017}) and, furthermore, it is consistent with the stylized fact stating the presence of long memory in volatility \cite{AndersenBollerslev1997, AndersenBollerslevDieboldLabys2001, BollerslevMikkelsen1996, DingGrangerEngle1993}. We also mention \cite{ChronopoulouViens2012, ComteCoutinRenault2012, Rosenbaum_2008} among other prominent works utilizing fBm with $H>1/2$.

In addition to inconsistencies for large $T$, Brownian diffusion models also have problems with shorter maturities. As reported in \cite{Delemotte_Marco_Segonne_2023, Fouque_Papanicolaou_Sircar_Solna_2004}, the empirically observed smile at-the-money becomes steeper as the time to maturity $T\to 0$ with the rule-of-thumb behavior 
\begin{equation}\label{intro: emp power law}
    \left|\frac{\widehat \sigma_{\text{emp}} (T, \kappa) - \widehat \sigma_{\text{emp}} (T, \kappa')}{\kappa - \kappa'}\right| \propto T^{-\frac{1}{2} + H}, \quad \kappa,\kappa' \approx 0, \quad H \in \left(0, \frac 1 2\right).
\end{equation}
Actually, it turns out that the \textit{power law} \eqref{intro: emp power law} is hard to reproduce by either Brownian diffusions or fractional models with $H>1/2$ (see e.g. the analysis in \cite[Section 7]{Alos_Leon_Vives_2007}). In order to replicate this behavior of implied volatility, \cite[Section 7.2.2]{Alos_Leon_Vives_2007} suggested using fBm with $H<1/2$. This methodology, further popularized by \cite{GatheralJaissonRosenbaum2018}, is known now as \textit{rough volatility}, and its efficiency in reproducing \eqref{intro: emp power law} can be explained from the two perspectives:
\begin{itemize}
    \item on the one hand, the theoretical result of \cite{Fukasawa_2021} shows that continuous semimartingale price on a market with power law \eqref{intro: emp power law} and no arbitrage implies that the paths of the market volatility \textit{must} have low H\"older regularity -- just like fBms with $H<1/2$;

    \item on the other hand, \cite{Alos_Leon_Vives_2007} proves that, in Malliavin differentiable models, the explosive behavior \eqref{intro: emp power law} in fact comes from the explosion in the Malliavin derivative of the volatility process. Conveniently, fBm with $H<1/2$ exhibits such a property.
\end{itemize}
The field of rough volatility has since evolved into a vast area of research, with a plethora of papers having been published over the years. A comprehensive and regularly updated list of literature on the subject can be found in \cite{Rough_volatility_literature} but we mention separately rough Bergomi \cite{Bayer_Friz_Gatheral_2016}, SABR \cite{Fukasawa_Gatheral_2022}, Stein-Stein \cite{Abi_Jaber_2022, Harms_Stefanovits_2019} and Heston \cite{EuchGatheralRosenbaum2018, Euch_Rosenbaum_2018, El_Euch_Rosenbaum_2019} models. However, rough volatility models are also not perfect and tend to have the following issues.
\begin{itemize}
    \item[(P1)] ``Roughness'' is usually introduced by incorporating fBm-like structures with $H<1/2$. However, as noted in \cite{Funahashi_Kijima_2017, Funahashi_Kijima_2017-1}, long memory as well as the shape of $\widehat \sigma(T,\kappa)$ for large $T$ hint that the preferable choice should be $H>1/2$. While there are stochastic models that can simultaneously be rough and very persistent (see e.g. \cite{Bennedsen_Lunde_Pakkanen_2022}), in the specific context of modeling volatility with fBm, there seems to be a peculiar paradox with the choice of $H$.

    \item[(P2)] In the stochastic volatility framework, the transition between the physical measure $\mathbb P$ and the pricing measure $\mathbb Q$ usually involves the term of the type $\int_0^T \frac{1}{\sigma^2(s)}ds$, where $\sigma = \{\sigma(t),~t\in[0,T]\}$ is the corresponding volatility process \cite{BGP2000}. However, in several important rough volatility models (e.g. rough Stein-Stein \cite{Abi_Jaber_2022, Harms_Stefanovits_2019} or rough Heston \cite{EuchGatheralRosenbaum2018, Euch_Rosenbaum_2018, El_Euch_Rosenbaum_2019}), it is not clear whether such integral is well-defined. Therefore, in such models, there is no transparent procedure for the change of measure.
    
    \item[(P3)] As shown in e.g. \cite{Gerhold_Gerstenecker_Pinter_2019}, rough volatility may produce moment explosions in price $S = \{S(t),~t\ge 0\}$. In fact, this phenomenon is not exclusive to the rough models only and is a prevalent feature within the stochastic volatility framework in general. In particular, as demonstrated in \cite[Section 3]{Andersen_Piterbarg_2005}, an ``unlucky'' choice of coefficients may cause $\mathbb E_{\mathbb P}[S^2(T)] = \infty$ and $\mathbb E_{\mathbb Q}[S^2(T)] = \infty$ for all big enough $T$ even in the standard Heston model. For more details, we refer the reader to \cite{Andersen_Piterbarg_2005}, \cite{Keller-Ressel_2011} or the article ``\textit{Moment Explosions}'' in \cite{Cont_2010}. 
    
    Models with moment explosions tend to have a number of issues of both technical and modeling nature\footnote{As a side note, we acknowledge that moment explosions may have a meaningful interpretation from the modeling point of view. More precisely, the celebrated moment formula of Lee \cite{Lee_2004} (see also its refinement \cite{Benaim_Friz_2009}) connects the asymptotic behavior of implied volatility surface when $\kappa \to \pm \infty$ with the values $$\widetilde p:= \sup\{p>0,~\mathbb E_{\mathbb Q}[S^{1+p}(T)] < \infty\}, \quad \widetilde q:= \sup\{q>0,~\mathbb E_{\mathbb Q}[S^{-q}(T)] < \infty\}.$$
    However, we emphasize that it is very hard to use this relationship as an argument either in favor of or against models with moment explosions. In practice, the observed range of strikes is always finite and hence it is not possible to determine the ``\textit{correct}'' asymptotics of empirical implied volatility $\widehat \sigma_{\text{emp}}(T,\kappa)$ as $\kappa \to \pm \infty$.}.
    \begin{itemize}
        \item[--]  As noted in \cite[Section 8]{Andersen_Piterbarg_2005}, ``\textit{several actively traded fixed-income derivatives require at least $L^2$ solutions to avoid infinite model prices}''. As examples of such assets, \cite[Appendix A]{Andersen_Piterbarg_2005} mentions CMS swaps and Eurodollar futures, but, in principle, problems may occur even for the standard European call options. Indeed, $S(T) \notin L^2(\mathbb P)$ normally implies that $(S(T) - K)_+ \notin L^2(\mathbb P)$ and hence it is generally not clear whether $\mathbb E_{\mathbb P}\left[\frac{d\mathbb Q}{d\mathbb P}(S(T) - K)_+\right] < \infty$. 
        
        \item[--] Moment explosions pose significant challenges for the related stochastic optimization problems. For instance, the infinite second moment of $(S(T) - K)_+$ entirely precludes the application of quadratic hedging techniques. This limitation significantly impacts stochastic volatility models that usually produce incomplete markets \cite{BGP2000}. In addition, in utility maximization problems, moment explosions can result in infinite expected utility, see e.g. \cite{Kallsen_Muhle-Karbe_2010}.

        \item[--] Finally, moment explosions lead to some issues with numerical procedures: for example, as discussed in \cite[Section 4.2]{Alfonsi_2010}, they may invalidate error estimates in discretization schemes for the corresponding models. The Malliavin integration-by-parts technique for numerical option pricing from \cite{AN2015} also implicitly relies on the existence of the second moment of the payoff.
    \end{itemize}
\end{itemize}

\subsection{The SVV model}

In this paper, we introduce a modeling framework, which we call the \textit{Sandwiched Volterra Volatility} (\textit{SVV}) model, that accounts for the problems (P1)--(P3) mentioned above. Namely, we consider 
\begin{gather}\label{intro: price}
    S_i(t) = S_i(0) + \int_0^t \mu_i(s) S_i(s)ds + \int_0^t Y_i(s) S_i(s) dB^S_i(s),
\end{gather}
\begin{equation}\label{intro: volatility}
    Y_i(t) = Y_i(0) + \int_0^t b_i(s, Y_i(s))ds + \int_0^t \mathcal K_i(t,s) dB_i^Y(s),
\end{equation}
$i=1,...,d$, where $S_i$ and $Y_i$ are price and volatility processes respectively, $\mu_i$ are deterministic continuous functions, $B^S_i$, $B^Y_i$ are correlated standard Brownian motions, $\mathcal K_i$ are arbitrary square integrable Volterra kernels such that each process $Z_i(t) := \int_0^t \mathcal K_i(t,s) dB_i^Y(s)$ is H\"older continuous up to the order $H_i \in (0,1)$. The main distinguishing property of the model is related to the drift functions $b_i$: for each $i=1,...,d$ there exist two deterministic continuous functions $0< \varphi_i < \psi_i$ and constants $c > 0$, $\gamma > \frac{1}{H_i} - 1$ and $y_*>0$ such that
\begin{equation}\label{intro: explosion}
\begin{gathered}
    b_i(t,y) \ge \frac{c}{(y-\varphi_i(t))^\gamma}, \quad y\in(\varphi_i(t), \varphi_i(t) + y_*),
    \\
    b_i(t,y) \le -\frac{c}{(\psi_i(t) - y)^\gamma}, \quad y\in(\psi_i(t) - y_*, \psi_i(t)).
\end{gathered}
\end{equation}
Processes of the type \eqref{intro: volatility} were extensively studied in \cite{DNMYT2020} and will be referred to as \emph{sandwiched processes} due to the fact that the above-mentioned shape of the drift ensures
\begin{equation}\label{intro: sandwich}
    \varphi_i(t) < Y_i(t) < \psi_i(t),
\end{equation}
i.e. the volatility process $Y_i$ is ``sandwiched'' between $\varphi_i$ and $\psi_i$, which, to some extent, links our model to the uncertain volatility approaches \cite{Avellaneda_Levy_Paras_1995, Karoui_Jeanblanc-Picque_Shreve_1998}.

Note that the choice of the drift with \eqref{intro: explosion} is not as exotic as it may seem at first glance. For example, as described in \cite[Section 4]{DNMYT2020}, sandwiched processes can be used to define generalizations of some common for stochastic volatility models such as Cox-Ingersoll-Ross (CIR) or Chan–Karolyi–Longstaff–Sanders (CKLS) processes. Here we also mention \cite{Hu2008} as well as the series of papers \cite{MYuT2018, MYuT2019, MYT2020}, where a process of the sandwiched type with $\varphi_i \equiv 0$, $\psi_i \equiv \infty$ was used in the context of CIR model driven by fractional Brownian motion with Hurst index $H>1/2$.

The SVV model \eqref{intro: price}--\eqref{intro: volatility} has several substantial advantages addressing the problems (P1)--(P3) mentioned above.
\begin{itemize}
    \item In line with the recent literature (see e.g. \cite{Jaber_Illand_Shaun_Li_2022} or \cite{Merino_Pospisil_Sobotka_Sottinen_Vives_2021}), we choose H\"older continuous Gaussian Volterra noises as drivers for our volatility processes. Such a choice is aimed at addressing the problem (P1): as discussed in \cite{Funahashi_Kijima_2017-1} (see also \cite[Section 7.7]{Alos_Garcia_Lorite_2021}), a simple linear combination of two fractional Brownian motions $B^{H_1}$ and $B^{H_2}$ with $H_1 > 1/2$ and $H_2<1/2$ can be sufficient to get both roughness/power law and long memory simultaneously. And indeed, as we prove in a separate paper \cite{DN_YT_power_law_2023}, taking
    \begin{equation}\label{intro: linear combination}
        \mathcal K_i(t,s) := \left(\theta_{1,i}(t-s)^{H_{1,i} - \frac{1}{2}} + \theta_{2,i} (t-s)^{H_{2,i} - \frac{1}{2}}\right)\mathbbm 1_{s<t}
    \end{equation}
    in the SVV model \eqref{intro: price}--\eqref{intro: volatility} with $H_{1,i} > 1/2$ and $H_{2,i} <1/2$ does give the power law of the implied volatility skew, despite the seemingly exotic choice of the drift and presence of the long memory component.
    
    In addition to a fairly straightforward kernel \eqref{intro: linear combination}, the SVV model also covers more involved covariance structures of the volatility noise. For example, one may choose
    \[
        \mathcal K_i(t,s) := (t-s)^{h_i(t)-\frac{1}{2}}\mathbbm 1_{s<t},
    \]
    where $h_i$: $[0,T] \to (0,1)$ is some H\"older continuous function, i.e. the process
    \[
        Z_i(t) := \int_0^t (t-s)^{h_i(t)-\frac{1}{2}} dB^Y_i(s)
    \]
    is a multifractional Brownian motion (see e.g. \cite{Peltier_Vehel_1995}). Such a choice can be supported by econometric evidence: in \cite[Section 2.2]{CLV_2014}, the regularity of SPX volatility is found to vary over time between $0.1$ and $0.9$. For more details on multifractional volatility, see also \cite{Ayache_Peng_2012}.

    \item In this paper, we choose $\varphi_i > 0$, so \eqref{intro: sandwich} implies that each process $Y_i$ is bounded away from zero and the integrals $\int_0^T \frac{1}{Y^2_i(s)}ds$ are well-defined. This allows us to provide a clear description of equivalent local martingale measures on the market completely solving the problem (P2). 

    \item The boundedness of $Y_i$ from above guaranteed by \eqref{intro: sandwich} ensures that, for any $r\in\mathbb R$,
    \begin{equation}\label{intro: existence of moments}
        \mathbb E\left[\sup_{t\in[0,T]} S^r_i(t)\right] < \infty,
    \end{equation}
    which eliminates the moment explosion problem described in (P3). This allows us to avoid any issues with infinite model prices as well as aids in numerical algorithms: in particular, we provide $L^r$-discretization schemes for price processes $S_i$, $r\ge 1$, as well as utilize the methodology of \cite{AN2015} for Monte Carlo pricing of options with discontinuous payoffs. Finally, \eqref{intro: existence of moments} for $r=2$ allows to perform mean-variance hedging within the SVV model which is analyzed in a separate paper \cite{Di_Nunno_Yurchenko-Tytarenko_hedging_2022}.  
\end{itemize}

In addition to the advantages mentioned above, we also note that \eqref{intro: sandwich} is a very convenient technical property that is often present in the literature (see e.g. \cite{Alos_Leon_2017, Fouque_Papanicolaou_Sircar_Solna_2003, OcK1991, Rosenbaum_Zhang_2021}). Moreover, the functions $\varphi_i$, $\psi_i$ can be regarded as another calibration parameters similar to the \textit{minimal instantaneous variance} in recent papers \cite{Gatheral_Jusselin_Rosenbaum_2020, Rosenbaum_Zhang_2021}.

\subsection{Structure of the paper and main results}

The results of this paper can be roughly divided into three parts.
\begin{itemize}
    \item[I.] \textbf{Description of the model}. In the first part contained in  Section \ref{sec: model description}, we provide detailed specifications of the SVV market model as well as characterize its properties from the financial viewpoint, Namely, 
    \begin{itemize}
        \item in Theorem \ref{th: properties of S}, we give the moment bounds for price processes $S_i$, $i=1,...,d$;

        \item in Subsection \ref{ssec: martingale measures}, we give the full description of the set of equivalent local martingale measures and prove that the market generated by \eqref{intro: price}--\eqref{intro: volatility} is arbitrage-free and incomplete;

        \item in Subsection \ref{ssec: empirical performance}, we provide some simulations to illustrate implied volatility surfaces generated by the SVV model.
    \end{itemize}

   \item[II.] \textbf{Malliavin differentiability.} The second part detailed in Section \ref{sec: Malliavin} is fully devoted to the problem of Malliavin differentiability of the SVV model \eqref{intro: price}--\eqref{intro: volatility}. The goal of our analysis is twofold. 
   \begin{itemize}
       \item On the one hand, the results of \cite{Alos_Leon_Vives_2007} allow to connect the Malliavin derivative of volatility with the power law of the corresponding implied volatility skew. In other words, Malliavin differentiability of \eqref{intro: volatility} can be used to analytically prove the power law within the SVV model. We perform this analysis in a separate paper \cite{DN_YT_power_law_2023}; it turns out that, with the right choice of the Volterra kernel, the SVV model indeed reproduces the power law \eqref{intro: emp power law}.

       \item On the other hand, Malliavin techniques are useful for numerical pricing of options with discontinuous payoffs in the spirit of \cite{AN2015, BMdP2018, MYT2020}.
   \end{itemize}
    In Section \ref{sec: Malliavin}, we prove the Malliavin differentiability of both volatility \eqref{intro: volatility} and price \eqref{intro: price} using the method similar to \cite[Theorem 3.3]{Hu2008} based on characterization of Sobolev spaces over an abstract Wiener space from \cite{Sugita1985}. As a consequence, we obtain absolute continuity of the law of the price processes \eqref{intro: price}.  

   \item[III.] \textbf{Malliavin integration-by-parts pricing of options with discontinuous payoffs}. In the last part, contained in Section \ref{sec: Malliavin duality approach}--\ref{sec: simulations}, we suggest an immediate numerical application of the Malliavin differentiability results from Section \ref{sec: Malliavin}. Namely, we consider the computation of $\mathbb E f(S(T))$, where $f$ is a discontinuous function. Usually, it is impossible to compute $\mathbb E f(S(T))$ analytically, so one must apply numerical methods for that; for instance, one can take an approximation $\widehat S(T)$ of $S(T)$ using some numerical scheme and perform some Monte Carlo-type simulation. This works well for Lipschitz payoffs $f$; however, according to \cite{Avikainen_2009}, any discontinuities in $f$ lead to a substantial deterioration in convergence speed. For instance, the rate of convergence of $\mathbb E[|f({S}(T)) - f(\widehat{S}(T))|^2]$ to zero is halved in comparison to the order of the scheme for $\widehat S$. This aggravation is additionally worsened by typically low convergence rates of numerical schemes for models with rough volatility and it significantly limits the application of some advanced Monte Carlo methods such as multi-level Monte Carlo. In order to overcome this problem, \cite{AN2015} suggested using Malliavin integration-by-parts to replace the discontinuous $f$ under the expectation with some Lipschitz functional, and, as a result, the initial convergence rate was preserved. A similar technique was applied in \cite{BMdP2018, MYT2020} for models driven by a fractional Brownian motion. In the present paper, we utilize the same approach for the SVV model \eqref{intro: price}--\eqref{intro: volatility} providing some useful extensions. In particular, unlike \cite{MYT2020}, we utilize the pathwise bounds for sandwiched processes derived in \cite{DNMYT2020} which allows to obtain the quadrature formulas without any limitations on the time horizon or regularity of the noise driving the volatility. In addition, all \cite{AN2015, BMdP2018, MYT2020} take $\mathbb E f(\widehat{S}(T))$ with respect to the physical measure (i.e. it is the \textit{expected payoff} rather than the \textit{price}) whereas we provide a quadrature formula under the change of measure as well. The analysis is performed both in 1-dimensional setting (e.g. for digital options) and for multidimensional basket options with discontinuous payoffs. This third part of our paper is organized as follows:
   \begin{itemize}
       \item Section \ref{sec: Malliavin duality approach} adapts the Malliavin integration-by-parts quadrature method from \cite{AN2015} to the SVV setting \eqref{intro: price}--\eqref{intro: volatility};
       \item in Section \ref{sec: numerics}, we give error estimates the mentioned method provided that the volatility is discretized using the drift-implicit Euler scheme from \cite{DNMYT2022};
       \item Section \ref{sec: simulations} contains simulation results.
   \end{itemize}
\end{itemize}

\section{Model description}\label{sec: model description}

In this Section, we define the Sandwiched Volterra Volatility (SVV) model as well as provide some basic results regarding its properties.

\begin{notation}
    Throughout this paper, $C$ denotes any positive deterministic constant the exact value of which is not relevant. Note that $C$ may change from line to line (or even within one line).
\end{notation}

\subsection{Preliminaries and assumptions}\label{subsec: Market description}

\paragraph{Probability space, filtration and correlation structure of Brownian motions.} 

Let $(\Omega, \mathcal F, \mathbb P)$ be the canonical $2d$-dimensional Wiener space, i.e. $\Omega := C_0([0,T]; \mathbb R^{2d})$, $\mathbb P$ is the classical Wiener measure, $\mathcal F$ is the $\mathbb P$-augmented Borel sigma-algebra. Let also $W$ be the corresponding $2d$-dimensional Brownian motion, i.e. $W$: $[0,T]\times\Omega \to \mathbb R^{2d}$, $W(t,\omega) = \omega(t)$, and $\mathbb F = \{\mathcal F_t,~t\in[0,T]\}$ be the $\mathbb P$-augmented filtration generated by $W$.

Now, let $\Sigma = (\sigma_{i,j})_{i,j=1}^{2d}$ be an arbitrary non-degenerate $2d \times 2d$ real correlation matrix, i.e. $\Sigma$ is symmetric, positive definite, $\sigma_{i,j} \in [-1,1]$ and $\sigma_{i,i} = 1$ for all $i=1,...,2d$. Since $\Sigma$ is symmetric and positive definite, it admits a Cholesky decomposition of the form $\Sigma = \Lc\Lc^T$, where $\Lc = (\ell_{i,j})_{i,j=1}^{2d}$ is a lower triangular matrix with strictly positive diagonal entries. Define $B(t) = ( B_1(t), ...,  B_{2d}(t))^T := \Lc W(t)$ and note that components of the $2d$-dimensional stochastic process $\{ B(t),~t\in[0,T]\}$ are $1$-dimensional Brownian motions with
\[
    \E\left[  B_i(t)  B_j(s) \right] =  (t \wedge s)\sigma_{i,j}, \quad t,s\in[0,T],~i,j=1,...,2d.
\]

In what follows, we will need to distinguish the first $d$ components of $B$ (that will generate the ``outer'' noise for the price) from the other $d$ components (that will generate the noise for the volatility) and thus we introduce the notation
\begin{gather*}
    B^S(t) = \left(B^S_1(t),..., B^S_d(t)\right)^T := (B_1(t), ..., B_{d}(t))^T,
    \\
    B^Y(t) = \left(B^Y_1(t),..., B^Y_d(t)\right)^T := (B_{d+1}(t), ..., B_{2d}(t))^T.
\end{gather*}
Clearly, since the matrix $\Lc$ is not degenerate, the filtration generated by the $2d$-dimensional process $B$ coincides with $\mathbb F$.

\paragraph{Stochastic volatility process.}
Consider now $d$ measurable functions $\mathcal K_i$: $[0,T]^2 \to \mathbb R$, $i=1,...,d$, that satisfy the following assumption.
\begin{assumption}\label{assum: assumptions on kernels} For any $i = 1,..., d$
    \begin{itemize}
        \item[(i)] $\mathcal K_i$ is a square integrable Volterra kernel, i.e. $\mathcal K_i(t,s) = 0$ whenever $t < s$,
        \[
            \int_0^T \int_0^T \mathcal K^2_i(t,s)dsdt = \int_0^T \int_0^t \mathcal K^2_i(t,s)dsdt < \infty
        \] 
        and
        \[
            \int_0^T \mathcal K^2_i(t,s)ds < \infty, \quad \int_0^T \mathcal K^2_i(t,s)dt < \infty, \quad \forall  s, t\in[0,T];
        \]
        \item[(ii)] there exists a constant $H_i \in (0,1)$ such that for any $\lambda \in (0,H_i)$
        \begin{equation}\label{eq: condition on kernel for Holder continuity}
            \int_0^t (\mathcal K_i(t,u) - \mathcal K_i(s, u))^2 du \le C_\lambda|t-s|^{2\lambda}, \quad 0\le s\le t\le T,
        \end{equation}
        where $C_\lambda > 0$ is a constant, possibly dependent on $\lambda$.
    \end{itemize}
\end{assumption}

Assumption \ref{assum: assumptions on kernels}(i) allows to define a $d$-dimensional Gaussian Volterra process $Z = \left(Z_1,..., Z_d\right)$ such that
\begin{equation}\label{eq: definition of the Volterra noise}
    Z_i(t) := \int_0^t \mathcal K_i (t, s) dB^Y_i(s), \quad t\in[0,T].
\end{equation}
Moreover, Assumption \ref{assum: assumptions on kernels}(ii) ensures that the paths of each $Z_i$ are regular a.s., which is summarized in the following proposition.

\begin{proposition}[\cite{ASVY2014}, Theorem 1 and Corollary 4]\label{prop: Holder continuity of noise}
    The Gaussian process $\{Z_i(t),~t\in[0,T]\}$ defined by \eqref{eq: definition of the Volterra noise} has a modification that is a.s. H\"older continuous up to the order $H_i \in (0,1)$ if and only if \eqref{eq: condition on kernel for Holder continuity} holds for any $\lambda \in (0,H_i)$. Moreover, for any $\lambda \in (0,H_i)$, the random variable
    \begin{equation}\label{eq: definition of Lambda}
        \Lambda_{\lambda, i} := \sup_{0\le s < t \le T} \frac{|Z_i(t) - Z_i(s)|}{|t-s|^\lambda}
    \end{equation}
    has moments of all orders.
\end{proposition}

\begin{remark}\label{rem: choice of modification of Z}
    By Proposition \ref{prop: Holder continuity of noise}, it is evident that one can choose the modification of $Z = (Z_1,...,Z_d)$ such that for \emph{all} $\omega \in \Omega$ 
    \[
        |Z_i(\omega, t) - Z_i(\omega, s)| \le \Lambda_{\lambda, i}(\omega)|t-s|^\lambda,
    \]
    where $\lambda \in (0,H_i)$. In what follows, this H\"older continuous modification will be used.   
\end{remark}

Now, let $\varphi_i$, $\psi_i$: $[0,T] \to \mathbb R$, $i=1,...,d$, be $d$ pairs of functions such that
\begin{enumerate}
    \item[1)] $\varphi_i$, $\psi_i$ are H\"older continuous up to the order $H_i$, where $H_i$ is from Assumption \ref{assum: assumptions on kernels}(ii);
    \item[2)] for each $i=1,...,d$ and $t\in[0,T]$:
    \begin{equation}\label{eq: psi greater phi, phi greater zero}
        0 < \varphi_i(t) < \psi_i(t).
    \end{equation}
\end{enumerate}
For any $a_1, a_2 \in \left[0, \min_{i=1,...,d}\frac{1}{2}\lVert \psi_i - \varphi_i \rVert_\infty\right)$, where $\lVert\cdot\rVert_\infty$ is the standard $\sup$-norm, denote 
\begin{equation}\label{eq: definition of the set D}
    \mathcal D^i_{a_1, a_2} := \{(t,y)\in[0,T]\times\mathbb R_+,~y\in(\varphi_i(t) + a_1, \psi_i(t) - a_2)\}
\end{equation}
and consider $d$ functions $b_i$: $\mathcal D^i_{0,0} \to \mathbb R$, $i=1,...,d$, that satisfy the following assumptions.

\begin{assumption}\label{assum: assumption on sandwiched drift} \hspace{10cm}
    \begin{itemize}
        \item[(i)]  The functions $b_i\in C(\mathcal D^i_{0,0})$, $i=1,...,d$;
    \end{itemize}
    and there exist constants $c>0$, $p > 1$, $y_* \in \left(0, \frac{1}{2} \min_{i=1,...,d} \lVert \psi_i - \varphi_i\rVert_\infty\right)$ such that for any $i=1,...,d$ the following conditions hold: 
    \begin{itemize}
        \item[(ii)] for any $\varepsilon \in \left(0, \min\left\{1, \min_{i=1,...,d}\frac{1}{2}\lVert \psi_i - \varphi_i \rVert_\infty\right\}\right)$
        \[
            |b_i(t_1,y_1) - b_i(t_2, y_2)| \le \frac{c}{\varepsilon^{p}} \left(|y_1 - y_2| + |t_1 - t_2|^\lambda \right), \quad (t_1, y_1), (t_2,y_2) \in \mathcal D^i_{\varepsilon, \varepsilon};
        \]
        \item[(iii)] for some constant $\gamma_i > \frac{1}{H_i} - 1$, where $H_i$ is from Assumption \ref{assum: assumptions on kernels}(ii),
        \[
            b_i(t, y) \ge \frac{c}{(y - \varphi_i(t))^{\gamma_i}}, \quad (t,y) \in \mathcal D^i_{0,0}\setminus \mathcal D^i_{y_*, 0},
        \]
        \[
            b_i(t, y) \le -\frac{c}{(\psi_i(t) - y)^{\gamma_i}}, \quad (t,y) \in \mathcal D^i_{0,0}\setminus \mathcal D^i_{0, y_*};
        \]
        \item[(iv)] there exists a continuous partial derivative $\frac{\partial b_i}{\partial y}$ with respect to the spatial variable and
        \[
            \frac{\partial b_i}{\partial y}(t, y) < c, \quad (t,y) \in \mathcal D^i_{0,0}.
        \]
    \end{itemize}
\end{assumption}

Next, consider the $d$-dimensional \emph{sandwiched process} $(Y_1,..., Y_d)$ defined by
\begin{equation}\label{eq: volatility process equations}
     Y_i(t) = Y_i(0) + \int_0^t b_i(s, Y_i(s)) ds + Z_i(t), \quad t\in[0,T], \quad i=1,...,d,
\end{equation}
where 
\begin{itemize}
    \item[--] $Y_i(0) \in \left(\varphi_i(0), \psi_i(0)\right)$ are deterministic constants;
    \item[--] $Z_i$, $i=1,...,d$, are defined by \eqref{eq: definition of the Volterra noise} with $\K_i$ satisfying Assumption \ref{assum: assumptions on kernels};
    \item[--] $b_i$, $i=1,...,d$, satisfy Assumption \ref{assum: assumption on sandwiched drift}.
\end{itemize}

\begin{remark}\label{rem: properties of Y}
    Each equation in \eqref{eq: volatility process equations} is treated pathwisely and, according to \cite{DNMYT2020}, has a unique solution for the given $\omega \in \Omega$ provided that the corresponding path $Z_i(\cdot, \omega)$ of the noise is H\"older continuous up to the order $H_i$ (and this happens a.s. by Remark \ref{rem: choice of modification of Z}). Moreover, in this case, according to \cite[Theorem 3.2]{DNMYT2020}, for each $i=1,...,d$ and $\lambda \in \left(\frac{1}{\gamma_i +1},H_i\right)$, where $\gamma_i$ is from Assumption \ref{assum: assumption on sandwiched drift}(iii), one can find deterministic constants $L_{1,i}$, $L_{2,i} > 0$ and $\alpha_i > 0$ that depend only on the shape of $b_i$ and $\lambda$ such that for all $t\in [0,T]$:
    \begin{equation}\label{eq: upper and lower bounds for sandwiched volatility}
        \varphi(t) + \frac{L_{1,i}}{ ( L_{2,i} + \Lambda_{\lambda, i} )^{\alpha_i} } \le Y_i(t) \le \psi(t) - \frac{L^1_i}{ ( L^{2,i} + \Lambda_{\lambda, i} )^{\alpha_i} },
    \end{equation}
    where $\Lambda_{\lambda, i}$ is a random variable defined by \eqref{eq: definition of Lambda}. In particular, this means that $0 < \varphi_i(t) < Y_i(t) < \psi_i(t)$ a.s., $t\in[0,T]$, $i=1,...,d$ (which justifies the name ``\emph{sandwiched}''). Note that $Z_i$ does not have to be Gaussian for \eqref{eq: upper and lower bounds for sandwiched volatility} to hold: it is sufficient to assume $\lambda$-H\"older continuity with $\lambda > \frac{1}{\gamma_i + 1}$.
    
    Since $\varphi_i$, $i=1,...,d$, are continuous positive functions, it is clear that for any $r>0$:
    \[
        \E \left[ \sup_{t\in[0,T]} \frac{1}{Y^r_i(t)} \right] \le \frac{1}{\min_{t\in[0,T]} \varphi_i^r(t)} < \infty.    
    \]
    Furthermore, since $\Lambda_{\lambda,i}$ has moments of all orders, it is easy to see that for any $r > 0$:
    \[
        \E \left[ \sup_{t\in[0,T]} \frac{1}{\left(Y_i(t) - \varphi_i(t)\right)^r} \right] < \infty, \quad \E \left[ \sup_{t\in[0,T]} \frac{1}{\left(\psi_i(t) - Y_i(t) \right)^r} \right] < \infty.
    \]
\end{remark}

\begin{remark}
    Note that item (iv) of Assumption \ref{assum: assumption on sandwiched drift} are not required for existence and uniqueness of the solution to \eqref{eq: volatility process equations}. It will be exploited later for the numerical scheme.
\end{remark}

\paragraph{Price process.} The components of the $d$-dimensional price process $S = (S_1,...,S_d)$ will be defined as solutions to the SDEs of the form
\begin{equation}\label{eq: definition of price process}
    S_i(t) = S_i(0) + \int_0^t \mu_i(s) S_i(s) ds + \int_0^t Y_i(s) S_i(s) d B^S_i(s), \quad t \in[0,T], \quad i=1,...,d,
\end{equation}
where 
\begin{itemize}
    \item[--] $S_i(0) > 0$, $i=1,...,d$, are deterministic constants;
    \item[--] $\mu_i$: $[0,T] \to \mathbb R$, $i=1,..., d$, are $H_i$-H\"older continuous functions, where each $H_i$ is from Assumption \ref{assum: assumptions on kernels}(ii);
    \item[--] $Y_i$, $i=1,..., d$ are sandwiched volatility processes defined above.
\end{itemize}  

We now list some simple properties of processes $S_i$, $i=1,...,d$.

\begin{theorem}\label{th: properties of S}
    For any $i=1,...,d$, equation \eqref{eq: definition of price process} has a unique solution of the form
    \begin{equation}\label{eq: solution to the price equation}
        S_i(t) = S_i(0) \exp\left\{ \int_0^t \left( \mu_i(s) - \frac{ Y_i^2(s)}{2}\right)ds + \int_0^t Y_i(s) d B^S_i(s)  \right\}.
    \end{equation}
    Furthermore, for any $r\in\mathbb R$:
    \[
        \mathbb E \left[\sup_{t\in[0,T]}S_i^r(t)\right] < \infty.
    \]
\end{theorem}

\begin{proof}
    Each $Y_i$, $i=1,...,d$, is a bounded process and thus existence and uniqueness of solution to the corresponding equation follows from e.g. \cite[Theorem 16.1.2]{CE2015} while the explicit form of the solution can be checked straightforwardly via It\^o's formula.
    
    Next, using the explicit form of $S_i$ as well as boundedness of $Y_i$ one can see that there exists a constant $C>0$ such that
    \begin{equation}\label{proofeq: bounds for moments of prices 1}
    \begin{aligned}
        S_i^r(t) &= S_i^r(0) \exp\left\{ r\int_0^t \mu_i(s) ds - \frac{r}{2} \int_0^t Y_i^2(s) ds + r\int_0^t Y_i(s) dB^S_i(s) \right\}
        \\
        &\le S_i^r(0)\exp\left\{|r|T\max_{s\in[0,T]}|\mu_i(s)| + \frac{|r|T}{2}\max_{s\in[0,T]}\psi_i^2(s)\right\} \exp\left\{ r\int_0^t Y_i(s) dB^S_i(s) \right\}
        \\
        &= S_i^r(0)\exp\left\{|r|T\max_{s\in[0,T]}|\mu_i(s)| + \frac{|r|T}{2}\max_{s\in[0,T]}\psi_i^2(s)\right\} \exp\left\{\frac{r^2}{2}\int_0^t Y_i^2(s) ds\right\}\times
        \\
        &\qquad \times \exp\left\{ \int_0^t rY(s) dB_i^S(s) - \frac{1}{2}\int_0^t (rY(s))^2ds \right\}
        \\
        & \le c_i \exp\left\{ \int_0^t rY_i(s) dB_i^S(s) - \frac{1}{2}\int_0^t (rY_i(s))^2ds \right\}
        \\
        &:= c_i \mathcal E_t\{rY_i\cdot B^S_i\},
    \end{aligned}
    \end{equation}
    where
    \[
        c_i := S_i^r(0)\exp\left\{|r|T\max_{s\in[0,T]}|\mu_i(s)| + \frac{|r|(|r|+1)T}{2}\max_{s\in[0,T]}\psi_i^2(s)\right\}
    \]
    and
    $$
        \mathcal E_t\{rY_i\cdot B^S_i\} := \exp\left\{ \int_0^t rY_i(s) dB^S_i(s) - \frac{1}{2}\int_0^t (rY_i(s))^2ds \right\}, \quad t\in[0,T].
    $$
    Note that the Novikov's criterion immediately yields that the process $\mathcal E_t\{rY_i\cdot B^S_i\}$, $t\in[0,T]$, is a uniformly integrable martingale such that
    \[
        \mathbb E \left[ \mathcal E_t\{rY_i\cdot B^S_i\} \right] = 1, \quad t\in [0,T],
    \]
    and, moreover,
    \[
        \mathcal E_t\{rY_i\cdot B^S_i\} = 1 + \int_0^t \mathcal E_s\{rY_i\cdot B^S_i\} r Y_i(s) dB^S_i(s), \quad t\in[0,T].
    \]
    By the Burkholder--Davis--Gundy inequality, there exists a constant $C_1 > 0$ such that
    \begin{equation}\label{proofeq: bounds for moaments of prices 3}
    \begin{aligned}
        \mathbb E \left[\sup_{t\in[0,T]}\mathcal E_t\{rY_i\cdot B^S_i\}\right] &= 1 + \mathbb E \left[\sup_{t\in[0,T]} \int_0^t \mathcal E_s\{rY_i\cdot B^S_i\} r Y_i(s) dB^S_i(s)\right]
        \\
        &\le 1 + C_1 \mathbb E\left[\left( \int_0^T \mathcal E^2_s\{rY_i\cdot B^S_i\} Y_i^2(s) ds \right)^{\frac{1}{2}}\right]
        \\
        &\le 1 + C_1 \sup_{s\in[0,T]} \psi_i(s) \left( \int_0^T \mathbb E\left[\mathcal E^2_s\{rY_i\cdot B^S_i\}\right] ds \right)^{\frac{1}{2}}.
    \end{aligned}   
    \end{equation}
    By Novikov's criterion, $\mathbb E\left[\mathcal E_s(2r Y_i\cdot B^S_i) \right] = 1$, so
    \begin{align*}
        \mathbb E\left[\mathcal E^2_s\{rY_i\cdot B^S_i\}\right] &= \mathbb E\left[\mathcal E_s\{2r Y_i\cdot B^S_i\} \exp\left\{ \int_0^s (rY_i(u))^2 du \right\}\right]
        \\
        &\le \exp\left\{ r^2T\max_{s\in[0,T]}\psi_i^2(s) \right\} \mathbb E\left[\mathcal E_s\{2r Y_i\cdot B^S_i\} \right] 
        \\
        & = \exp\left\{ r^2T\max_{s\in[0,T]}\psi_i^2(s) \right\} < \infty
    \end{align*}
    and hence, taking into account \eqref{proofeq: bounds for moments of prices 1}, we obtain that
    \begin{align*}
        \mathbb E\left[\sup_{t\in[0,T]}S_i^r(t)\right] &\le c_i\left( 1+C_1 \sqrt{T}\sup_{s\in[0,T]}\psi_i(s) \exp\left\{ \frac{r^2T}{2}\max_{s\in[0,T]}\psi_i^2(s) \right\}\right),
    \end{align*}
    which yields the required result.
\end{proof}

\paragraph{Num\'eraire and discounted price process.} As a num\'eraire, we will use the function $e^{\int_0^t \nu(s) ds}$, $t\in[0,T]$, where $\nu$: $[0,T]\to\mathbb R_+$ denotes a H\"older continuous function of order $\min_{i=1,...,d} H_i$ representing an instantaneous interest rate. The discounted price process denoted as
\[
    \widetilde S_i(t) := e^{-\int_0^t \nu(s) ds} S_i(t), \quad t\in[0,T], \quad i=1,...,d,
\]
has thus dynamics of the form
\[
   \widetilde S_i(t) = S_i(0) + \int_0^t \widetilde\mu_i(s) \widetilde S_i(s) ds + \int_0^t Y_i(s) \widetilde S_i(s) d B^S_i(s), \quad t \in[0,T], \quad i=1,...,d, 
\]
where $\widetilde\mu_i := \mu_i - \nu$.

\subsection{Incompleteness of the market and martingale measures}\label{ssec: martingale measures}

Recall that our model is driven by the $2d$-dimensional Gaussian process 
\[
    \left( \begin{matrix} B^S(t) \\ B^Y(t) \end{matrix} \right) = \Lc W(t), \quad t\in[0,T],
\]
where $W = \{W(t),~t\in[0,T]\}$ is a standard $2d$-dimensional Brownian motion and $\Lc = (\ell_{i,j})_{i,j=1}^{2d}$ is nondegenerate lower triangular $2d \times 2d$ real matrix with positive entries on a diagonal such that $\Lc \Lc^T = \Sigma$. Let also
\[
    \Lc = \left(\begin{matrix} 
                    \Lc_{11} &\Lc_{12}\\
                    \Lc_{21} &\Lc_{22}
                \end{matrix}\right),
\]
where $\Lc_{ij}$ we denote the $d \times d$ blocks of matrix $\Lc$. Note that $\Lc_{12}$ is a zero matrix and $\Lc_{11}$ is a lower triangular matrix with positive elements on the diagonal. In particular, $\Lc_{11}$ is invertible, so we can define its inverse $\Lc_{11}^{-1} = \left(\ell^{(-1)}_{j,k}\right)_{j,k=1}^d$. Finally, for a given semimartingale $X$, we denote the corresponding Dol\'eans-Dade exponential by
\[
    \mathcal E_t\left\{ X \right\} := \exp\left\{ X_t - X_0 - \frac{1}{2} [X]_t \right\}.
\]

The next lemma is in the spirit of Lemma~2.10 in \cite{BGP2000}.

\begin{lemma}\label{lemma: local martingality conditions}
    Let $M$ be a positive local $\mathbb P$-martingale w.r.t. $\mathbb F$ with $M(0) = 1$ a.s. Then $M\widetilde S_i$, $i=1,...,d$, are all local martingales if and only if
    \begin{equation}\label{eq: representation of M as a stochastic exponential}
    \begin{aligned}
        M(t) &= \mathcal E_t \left\{ -\sum_{j=1}^d \int_0^\cdot \left(\sum_{k=1}^d \ell^{(-1)}_{j,k} \frac{\widetilde \mu_k(s)}{Y_k(s)} \right)dW_j(s) \right\} \mathcal E_t \left\{ \sum_{l=d+1}^{2d}  \int_0^\cdot \xi_l(s) dW_l(s)\right\}
        \\
        &= \mathcal E_t \left\{ -\sum_{j=1}^d \int_0^\cdot \left(\sum_{k=1}^d \ell^{(-1)}_{j,k} \frac{\widetilde \mu_k(s)}{Y_k(s)}  \right)dW_j(s) + \sum_{l=d+1}^{2d}  \int_0^\cdot \xi_l(s) dW_l(s)\right\},
    \end{aligned}
    \end{equation}
    where $\left(\ell^{(-1)}_{j,k}\right)_{j,k=1}^d$ are the elements of $\Lc_{11}^{-1}$ and $\xi_l$, $l=d+1,...,2d$, are {predictable} processes such that $\int_0^\cdot \xi_l(s) dW_l(s)$ are well-defined local martingales.
\end{lemma}
   
\begin{proof} 
    By the virtue of the local martingale representation theorem, $M$ is a local $\mathbb P$-martingale w.r.t. $\mathbb F$ if and only if there exist {predictable} processes $\zeta_j = \{\zeta_j(t),~t\in[0,T]\}$, $j=1,...,2d$, such that
    \[
        \mathbb P \left(\int_0^T \zeta_j^2(s) ds < \infty\right) = 1, \quad j = 1,...,2d,
    \]
    and
    \begin{equation}\label{eq: Martingale representation of M}
        M(t) = 1 + \sum_{j=1}^{2d} \int_0^t \zeta_j (s) d W_j(s).
    \end{equation}
    By It\^o's formula, 
    \begin{align*}
        d(M(t) \widetilde S_i(t)) &= \left( \tilde\mu_i(t) M(t) + Y_i(t) \sum_{j=1}^{2d} \ell_{i,j} \zeta_j(t) \right) \widetilde S_i(t) dt 
        \\
        &\quad+ \sum_{j=1}^{2d} \left( \ell_{i,j} M(t) Y_i(t) + \zeta_j(t) \right)\widetilde S_i(t) dW_j(t),
    \end{align*}
    and thus $M\widetilde S_i$ is a local martingale if and only if
    \[
        \int_0^t \left( \tilde\mu_i(s) M(s) + Y_i(s) \sum_{j=1}^{2d} \ell_{i,j} \zeta_j(s) \right) \widetilde S_i(s)ds = 0, \quad t\in[0,T]. 
    \]
    This, in turn, implies that
    \[
          \sum_{j=1}^{2d} \ell_{i,j} \zeta_j(t) = - M(t)\frac{\tilde\mu_i(t)}{Y_i(t)}, \quad i=1,...,d,
    \]
    almost everywhere w.r.t. $dt\otimes d\mathbb P$. Since $\Lc$ is lower-triangular and each $Y_i(t)$ is positive a.s., we have that
    \begin{equation}\label{eq: system of equations to find martingale measure}
        \Lc_{11} \left(\begin{matrix} \zeta_1(t)\\ \vdots \\ \zeta_d(t) \end{matrix}\right) = -M(t) \left(\begin{matrix} \frac{\tilde\mu_1(t)}{Y_1(t)} \\ \vdots \\ \frac{\tilde\mu_d(t)}{Y_d(t)} \end{matrix}\right)
    \end{equation}
    almost everywhere w.r.t. $dt\otimes d\mathbb P$ and hence
    \begin{equation}\label{eq: system of equations to find martingale measure solution}
        \left(\begin{matrix} \zeta_1(t)\\ \vdots \\ \zeta_d(t) \end{matrix}\right) = \Lc^{-1}_{11}  \left(\begin{matrix} -M(t)\frac{\tilde\mu_1(t)}{Y_1(t)} \\ \vdots \\ -M(t)\frac{\tilde\mu_d(t)}{Y_d(t)} \end{matrix}\right)
    \end{equation}
    almost everywhere w.r.t. $dt\otimes d\mathbb P$. Using notation $\Lc_{11}^{-1} = \left(\ell^{(-1)}_{j,k}\right)_{j,k=1}^d$ and substituting \eqref{eq: system of equations to find martingale measure solution} to \eqref{eq: Martingale representation of M}, we see that the process $M$ should satisfy the SDE of the form
    \[
        dM(t) = -M(t) \sum_{j=1}^d \sum_{k=1}^d \ell^{(-1)}_{j,k} \frac{\widetilde \mu_k(t)}{Y_k(t)} dW_j(t) + \sum_{l=d+1}^{2d} \zeta_l(t) dW_l(t).
    \]
    This SDE has a unique solution (see e.g. \cite{CE2015}) of the form
    \begin{equation}\label{eq: representation of M as stochastic exponential first piece}
    \begin{aligned}
        M(t) &= \mathcal E_t \left\{ -\sum_{j=1}^d \int_0^\cdot \left(\sum_{k=1}^d \ell^{(-1)}_{j,k} \frac{\widetilde \mu_k(s)}{Y_k(s)} \right)dW_j(s) \right\} \times
        \\
        &\qquad \times \left( 1 + \sum_{l=d+1}^{2d}\int_0^t  \mathcal E^{-1}_u \left\{ -\sum_{j=1}^d \int_0^\cdot \left(\sum_{k=1}^d \ell^{(-1)}_{j,k} \frac{\widetilde \mu_k(s)}{Y_k(s)} \right)dW_j(s) \right\}   \zeta_l(u) dW_l(u) \right)
        \\
        &=: \mathcal E_t \left\{ -\sum_{j=1}^d \int_0^\cdot \left(\sum_{k=1}^d \ell^{(-1)}_{j,k} \frac{\widetilde \mu_k(s)}{Y_k(s)} \right)dW_j(s) \right\} \left( 1 + \sum_{l=d+1}^{2d}\int_0^t \widetilde \zeta_l(u)  dW_l(u) \right).
    \end{aligned}
    \end{equation}
    Furthermore, since the process $M$ is positive, the continuous process $\widetilde M$ defined by
    \[
        \widetilde M(t) := 1 + \sum_{l=d+1}^{2d} \int_0^t \widetilde \zeta_l(s) dW_l(s)
    \]
    must also be positive, thus it can be represented as
    \begin{equation}\label{eq: representation of M as stochastic exponential second piece}
    \begin{aligned}
        \widetilde M(t) = \mathcal E_t\left( \sum_{l=d+1}^{2d}\int_0^\cdot \xi_l(s) dW_l(s) \right)
    \end{aligned}
    \end{equation}
    with $\xi_l$, $l=d+1,...,2d$, being some {predictable} processes such that the corresponding stochastic integrals are well-defined. It remains to notice that $\widetilde \zeta_l$, $l=d+1,...,2d$, can be arbitrary {predictable} stochastic processes such that the corresponding stochastic integrals are well-defined, thus $\xi_l$, $l=d+1,...,2d$, can also be chosen arbitrarily. Taking into account \eqref{eq: representation of M as stochastic exponential first piece} and \eqref{eq: representation of M as stochastic exponential second piece}, we get the representation \eqref{eq: representation of M as a stochastic exponential}.
\end{proof}

In the literature, positive densities $M(T)$ associated with local $\mathbb P$-martingales $M$ such that $M\widetilde S_i$ all become local martingales are commonly called \emph{strict local martingale densities} (see \cite{Schweizer1992} for more detail) and Lemma \ref{lemma: local martingality conditions} thus gives the description of the set of all strict local martingale densities on the market. However, note that the measure defined by $d \mathbb Q := M(T)d\mathbb P$ is not necessarily a probability measure and one has to check that $\mathbb Q(\Omega) = \mathbb E [M(T)] = 1$ (or, equivalently, that $M$ is a uniformly integrable martingale) separately. Moreover, exploiting the Bayes' theorem for conditional expectations, it is easy to see that if $\mathbb Q$ defined above is a probability measure, then $\widetilde S_i$, $i=1,...,d$, are local $\mathbb Q$-martingales, which would imply the no-arbitrage property of the market. It turns out that the market model in consideration indeed has this property which is formulated in the next theorem.

\begin{theorem}
    The market defined in Subsection \ref{subsec: Market description} is arbitrage-free and incomplete.
\end{theorem}
\begin{proof}
    Each $Y_k$, $k=1,...,d$, is bounded away from zero by the corresponding $\min_{t\in[0,T]} \varphi_k(t) > 0$, thus, by Novikov's condition,
    \[
        M(t) = \mathcal E_t \left\{ -\sum_{j=1}^d \int_0^\cdot \left(\sum_{k=1}^d \ell^{(-1)}_{j,k} \frac{\widetilde \mu_k(s)}{Y_k(s)} \right)dW_j(s) \right\}
    \]
    is a uniformly integrable martingale (this case corresponds to $\xi_l \equiv 0$ in \eqref{eq: representation of M as a stochastic exponential}, $l=d+1,...,2d$). Therefore the associated $\mathbb Q$ is a probability measure, i.e. the market is arbitrage free. 
    
    Furthermore, again by Novikov's condition, any $\xi_l$ such that
    \[
        \E \left[ \exp\left\{ \frac{1}{2} \sum_{l=d+1}^{2d} \int_0^T \xi_l^2(t) dt \right\} \right] < \infty
    \]
    leads to $M$ being a uniformly integrable martingale, thus the market is incomplete.
\end{proof}

\subsection{Implied volatility surface generated by the SVV model}\label{ssec: empirical performance}

One of the accepted benchmarks for the quality of a financial model is the shape of the implied volatility generated by it (see e.g. the discussion in \cite{Di_Nunno_Kubilius_Mishura_Yurchenko-Tytarenko_2023}). In this subsection, we provide some results regarding the shape of the implied volatility surfaces generated by the SVV model as well as illustrate these shapes with simulations.

Assume, for simplicity, that the instantaneous interest rate $\nu$ is constant and consider the SVV model of the form
\begin{equation}\label{eq: SVV for simulation of IV theory}
\begin{aligned}
     S(t) &= S(0) + \nu\int_0^t S(s)ds + \int_0^t Y(s) S(s) \left(\sqrt{1 - \rho^2} dB_1(s) + \rho dB_2(s)\right),
     \\
     Y(t) &= Y(0) + \int_0^t b(s,Y(s))ds + \int_0^t \mathcal K(t,s) dB_2(s),
\end{aligned}
\end{equation}
where $B_1, B_2$ are independent Brownian motions, $\rho \in (-1,1)$, $\mathcal K$ satisfies Assumption \ref{assum: assumptions on kernels} and $b$ satisfies Assumption \ref{assum: assumption on sandwiched drift} so that
\begin{equation}\label{eq: SVV IV bounds}
    0 < \min_{t\in[0,T]} \varphi(t) < \varphi(t) < Y(t) < \psi(t) < \max_{t\in[0,T]} \psi(t) < \infty.
\end{equation}

\begin{remark}
    It is easy to check that, in the particular case of the SVV model \eqref{eq: SVV for simulation of IV theory}, the discounted price process $\widetilde S = \{\widetilde S(t),~t\in[0,T]\} = \{e^{-\nu t} S(t),~t\in[0,T]\}$ satisfies the SDE
    \begin{equation*}
        \widetilde S(t) = S(0) + \int_0^t Y(s) \widetilde S(s) \left(\sqrt{1 - \rho^2} dB_1(s) + \rho dB_2(s)\right),
    \end{equation*}
    and hence is a martingale. This means that, under the model \eqref{eq: SVV for simulation of IV theory}, the non-arbitrage price $\Pi$ of a European call-option with maturity $T$ and payoff $K$ is equal to
    \[
        \Pi_{SVV}(T,K) = e^{-\nu T} \mathbb E\left[ (S(T) - K)_+ \right].
    \]
    Such a choice allows to simplify the arguments in this section and increase the numerical stability of the simulations.
\end{remark}

Consider also the standard Black-Scholes price of a European call-option with payoff $K$ and maturity $T$ at a given volatility level $\sigma>0$:
\[
    \Pi_{BS}(T,K; \sigma) := S(0) \Phi\left( \frac{\log\frac{e^{\nu T}S(0)}{K} + \frac{\sigma^2}{2}T}{\sigma\sqrt{T}} \right) - Ke^{-\nu T} \Phi\left( \frac{\log\frac{e^{\nu T}S(0)}{K} - \frac{\sigma^2}{2}T}{\sigma\sqrt{T}} \right),
\]
where $\Phi$ is the CDF of the standard normal distribution. We are interested in the behavior of the implied volatility surface $(T,K) \mapsto \widehat \sigma(T,K)$ generated by the SVV model defined as
\[
    \Pi_{BS}(T,K; \widehat \sigma(T,K)) = \Pi_{SVV}(T,K).
\]
\begin{remark}
    It is easy to check that $\widehat \sigma(T,K)$ exists and is uniquely defined for all $T$, $K>0$ since the function $\sigma \mapsto \Pi_{BS}(T,K; \sigma)$ is strictly increasing for all $T$, $K>0$, and, moreover,
    \begin{align*}
        \lim_{\sigma \to 0+} \Pi_{BS}(T,K; \sigma) &= (S(0) - e^{-\nu T} K)_+ 
        \\
        &=  (\mathbb E\left[e^{-\nu T} S(T)\right] - e^{-\nu T} K)_+ 
        \\
        &\le \mathbb E\left[(e^{-\nu T} S(T) - e^{-\nu T} K)_+\right] = \Pi_{SVV}(T,K).
    \end{align*}
\end{remark}

\paragraph{Boundedness of the implied volatility.} First of all, let us prove that the SVV implied volatility $(T,K) \mapsto \widehat \sigma(T,K)$ is bounded.

\begin{theorem}\label{eq: IV is bounded}
    For all $T$, $K>0$,
    \[
        \min_{t\in[0,T]} \varphi(t) \le \widehat \sigma(T,K) \le \max_{t\in[0,T]} \psi(t).
    \]
\end{theorem}
\begin{proof}
    By the seminal result \cite{Avellaneda_Levy_Paras_1995} of Avellaneda, Levy and Parás, the bounds \eqref{eq: SVV IV bounds} imply that 
    \[
        \Pi_{BS}(T,K; \min_{t\in[0,T]} \varphi(t)) \le \Pi_{SVV}(T,K) \le \Pi_{BS}(T,K; \max_{t\in[0,T]} \psi(t)),
    \]
    which immediately yields the claim.
\end{proof}

\begin{remark}
    Theorem \ref{eq: IV is bounded} indicates that one should choose the upper bound function $\psi$ such that $\max_{t\in[0,T]} \psi(t)$ exceeds the maximum value of the realized implied volatility observed on the market. Under the normal market conditions, it is usually sufficient to take $\max_{t\in[0,T]} \psi(t) = 1$.
\end{remark}

\paragraph{At-the-money skew.} For the notational convenience, denote $\kappa := \log\frac{K}{e^{\nu T}S(0)}$ and consider a reparametrization $\widehat \sigma_{\text{log-m}} = \widehat \sigma_{\text{log-m}}(T,\kappa)$ of $\widehat \sigma = \widehat \sigma(T,K)$ defined as 
\[
    \widehat \sigma_{\text{log-m}}(T,\kappa) := \widehat \sigma(T, S(0)e^{\kappa + \nu T}).
\]
It is well-known (see e.g. a detailed discussion in \cite[Section 2.2]{Di_Nunno_Kubilius_Mishura_Yurchenko-Tytarenko_2023} or \cite{Delemotte_Marco_Segonne_2023}), that smiles at-the-money of empirically observed implied volatilities $\widehat \sigma_{\text{emp}}(T,\kappa)$ become progressively steeper as $T\to 0$ with a rule-of-thumb approximation
\begin{equation}\label{eq: empirical power law}
    \left|\frac{\widehat \sigma_{\text{emp}} (T, \kappa) - \widehat \sigma_{\text{emp}} (T, \kappa')}{\kappa - \kappa'}\right| \propto T^{-\frac{1}{2} + H}, \quad \kappa,\kappa' \approx 0, \quad H \in \left(0, \frac 1 2\right).
\end{equation}
In order to replicate the empirical behavior \eqref{eq: empirical power law}, one may want a model generating the implied volatility $\widehat \sigma_{\text{log-m}} = \widehat \sigma_{\text{log-m}}(T,\kappa)$ with
\begin{equation}\label{eq: theoretical power law}
    \left|\frac{\partial \widehat \sigma_{\text{log-m}}}{\partial \kappa} (T, \kappa)\right|_{\kappa = 0} = O(T^{-\frac{1}{2} + H}), \quad T \to 0.
\end{equation}
Luckily, the SVV model \eqref{eq: SVV for simulation of IV theory} does possess the property \eqref{eq: theoretical power law} with the right choice of the kernel $\mathcal K$.

\begin{theorem}\label{th: SVV power law}
    Let $H\in\left(\frac{1}{6}, \frac{1}{2}\right)$ and $\rho < 0$ in \eqref{eq: SVV for simulation of IV theory}. Take 
    \begin{equation}\label{eq: explosive drift for power law}
        b(t,y) = \frac{\theta_1(t)}{(y - \varphi(t))^{\gamma_1}} - \frac{\theta_2(t)}{(\psi(t) - y)^{\gamma_2}} + a(t,y),
    \end{equation}
    where
    \begin{itemize}
        \item $\gamma_1 > \frac{1}{H} - 1$, $\gamma_2 > \frac{1}{H} - 1$;

        \item the functions $\theta_1$, $\theta_2$: $[0,T]\to\mathbb R$ are strictly positive and continuous;

        \item the function $a$: $[0,T]\times\mathbb R\to\mathbb R$ is locally Lipschitz in $y$ uniformly in $t$, i.e. for any $N > 0$ there exists a constant $C_N>0$ that does not depend on $t$ such that
        \[
            |a(t,y_2) - a(t, y_1)| \le C_N|y_2 - y_1|, \quad t\in [0,T], \quad y_1, y_2 \in [-N,N];
        \]
        
        \item $a$: $[0,T]\times\mathbb R\to\mathbb R$ is two times differentiable w.r.t. the spatial variable $y$ with $a$, $a'_y$, $a''_{yy}$ all being continuous on $[0,T]\times\mathbb R$
    \end{itemize}
    and assume that the kernel $\mathcal K$ is such that
    \begin{itemize}
        \item for any $0 \le s < t \le T$, there is a constant $C>0$ such that
        \begin{equation*}
            |\mathcal K(t,s)| \le C|t-s|^{-\frac{1}{2}+H},
        \end{equation*}
        \item there exists a constant $K_Y > 0$ such that
        \begin{equation}\label{eq: Kernel explosion}
            \frac{1}{\tau^{\frac{3}{2} + H}} \int_{0}^\tau \int_s^\tau \mathcal K(t,s) dtds \to K_Y, \quad \tau \to 0+.
        \end{equation}
    \end{itemize}
    Then the power law \eqref{eq: theoretical power law} holds for the SVV model \eqref{eq: SVV for simulation of IV theory}.
\end{theorem}

\begin{proof}
    The proof of this claim can be found in the dedicated paper \cite{DN_YT_power_law_2023}.
\end{proof}

\begin{remark}
    The claim of Theorem \ref{th: SVV power law} holds for drifts that are more general than \eqref{eq: explosive drift for power law}. For more details, see Remark \cite[Remark 2.6]{DN_YT_power_law_2023}.
\end{remark}

\begin{remark}
    The condition $H>\frac{1}{6}$ in Theorem \ref{th: SVV power law} is consistent with the recent empirical estimate $H \approx 0.19$ for the SPX implied volatility obtained in \cite{Delemotte_Marco_Segonne_2023}.
\end{remark}

\paragraph{Simulations.} For simulations, we take $\nu = 0$ and the SVV model of the form
\begin{equation}\label{eq: SVV for simulation of IV}
\begin{aligned}
     S(t) &= 1 + \int_0^t Y(s) S(s) \left(\sqrt{0.75}dB_1(s)-0.5 dB_2(s)\right),
     \\
     Y(t) &= 0.5 + \int_0^t \left(\frac{0.005}{(Y(s) - 0.05)^5} - \frac{0.005}{(1 - Y(s))^5} + 0.05(0.5 - Y(s))\right)ds + 0.3\int_0^t \mathcal K(t,s) dB_2(s),
\end{aligned}
\end{equation}
where the choice of $\mathcal K$ is varied. In order to plot the implied volatility surface, we 
\begin{itemize}
    \item estimate the call option prices $\mathbb E\left[(S(T) - K)_+\right]$ under the model \eqref{eq: SVV for simulation of IV} for $T=\frac{n}{200}$, $n=1,...,200$, and $K=0.5 + m/100$, $m=0,...,150$, using the standard Monte Carlo method, i.e. average over 1500000 realizations of the payoff $(S(T) - K)_+$ (the algorithm for the simulation of the random variable $S(T)$ is described in full detail further in Subsection \ref{subsec: numerical scheme});
    \item calculate the corresponding Black-Scholes implied volatility $(T,\kappa) \mapsto \widehat \sigma(T,\kappa)$, where $\kappa$ denotes the log-moneyness, using the standard procedure (see e.g. \cite[Section 2.2]{Di_Nunno_Kubilius_Mishura_Yurchenko-Tytarenko_2023} for more details).
\end{itemize}

\begin{remark}
    Note that the parameters in \eqref{eq: SVV for simulation of IV} are illustrational and are not calibrated to the real market data.
\end{remark}


\begin{example}\label{ex: purely rough IV}
    Let us consider the rough fractional kernel
    \begin{equation}\label{eq: IV rough kernel}
        \mathcal K(t,s) := \frac{1}{\Gamma(0.7)}(t-s)^{-0.3} \mathbbm 1_{s<t}
    \end{equation}
    in \eqref{eq: SVV for simulation of IV}. In this case, the implied volatility surface produced by the SVV model \eqref{eq: SVV for simulation of IV} is depicted on Figure~\ref{fig: rough IV surface}.

    Note that Figure~\ref{fig: rough IV surface}(b) visually demonstrates how that the absolute at-the-money implied volatility skew 
    \begin{equation}\label{eq: SV skew}
        \left|\frac{\partial \widehat \sigma_{\text{log-m}}}{\partial \kappa}(T,0)\right|
    \end{equation}    
    increases as $T\to 0$. Figure \ref{fig: rough IV skew} contains the variation of \eqref{eq: SV skew} with respect to $T$ on the standard (Fig.~\ref{fig: rough IV skew}(a)) and logarithmic (Fig.~\ref{fig: rough IV skew}(b)) scales and shows that \eqref{eq: SV skew} indeed follows the power law with
    \[
        \left|\frac{\partial \widehat \sigma_{\text{log-m}}}{\partial \kappa}(T,0)\right| = O(T^{-0.3}), \quad T\to 0.
    \]

    \begin{figure}[h!]
     \centering
     \begin{subfigure}[b]{0.48\textwidth}
         \centering
         \includegraphics[width=\textwidth]{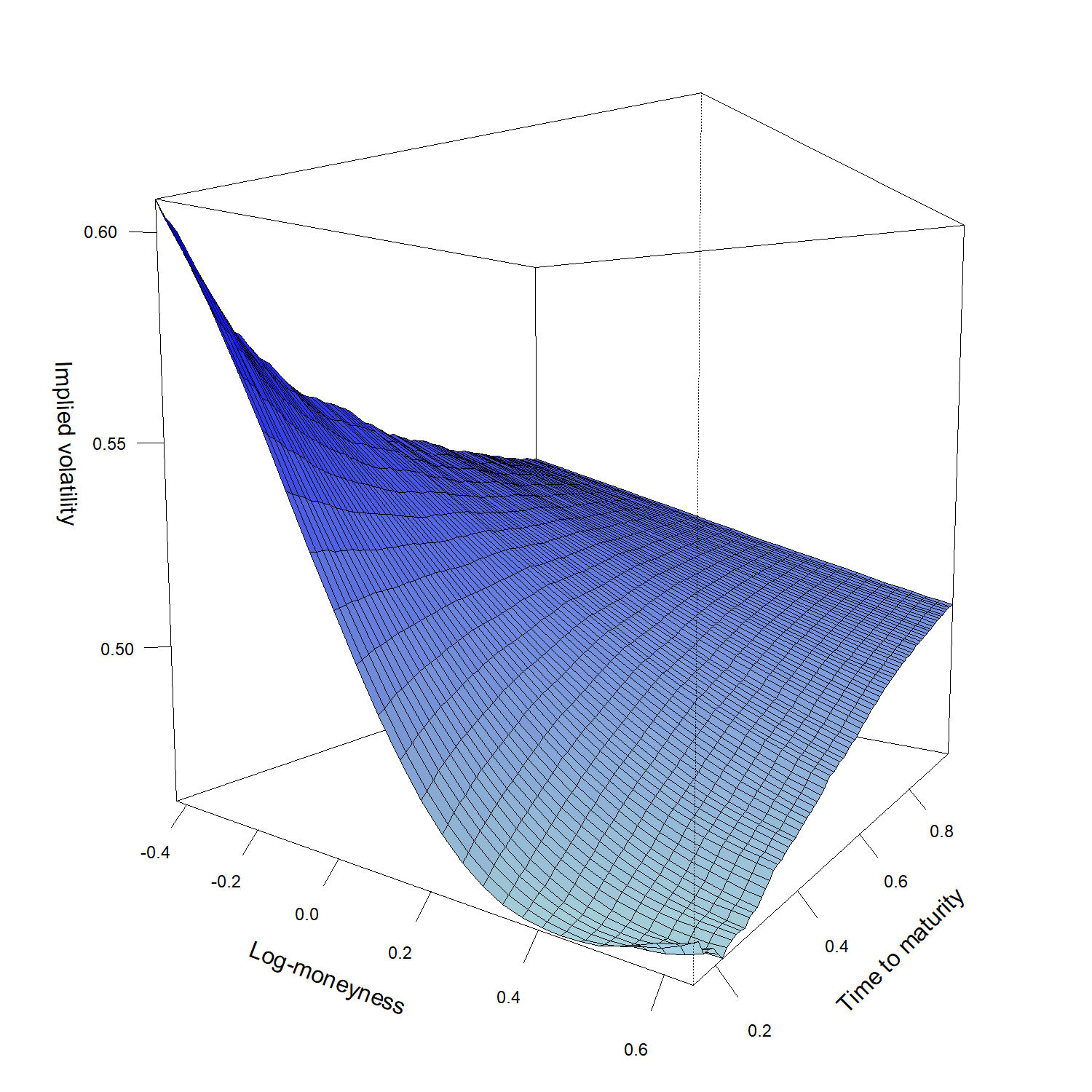}
         \caption{Implied volatility surface}
     \end{subfigure}
     \hfill
     \begin{subfigure}[b]{0.48\textwidth}
         \centering
         \includegraphics[width=\textwidth]{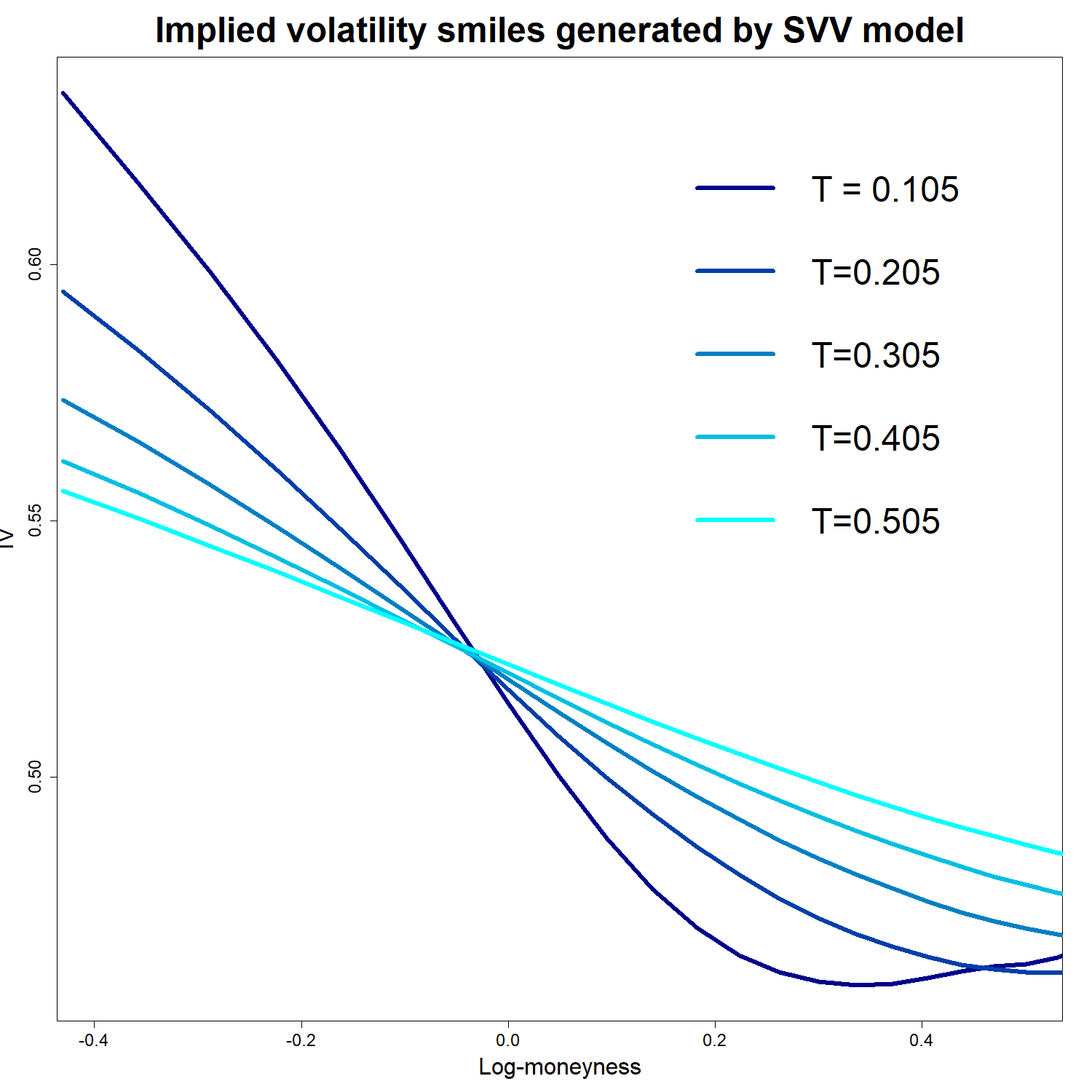}
         \caption{Implied volatility smiles for different maturities}
     \end{subfigure}
    \caption{Implied volatility surface (a) and implied volatility smiles (b) generated by the SVV model \eqref{eq: SVV for simulation of IV} with the rough kernel \eqref{eq: IV rough kernel}. Note that the smile becomes steeper as the time to maturity $T \to 0$, which reproduces a similar effect happening on real markets (for more details, see e.g. \cite{Delemotte_Marco_Segonne_2023, Fouque_Papanicolaou_Sircar_Solna_2004} or \cite[Section 2.2]{Di_Nunno_Kubilius_Mishura_Yurchenko-Tytarenko_2023}).}
    \label{fig: rough IV surface}
    \end{figure}

    \begin{figure}[h!]
     \centering
     \begin{subfigure}[b]{0.48\textwidth}
         \centering
         \includegraphics[width=\textwidth]{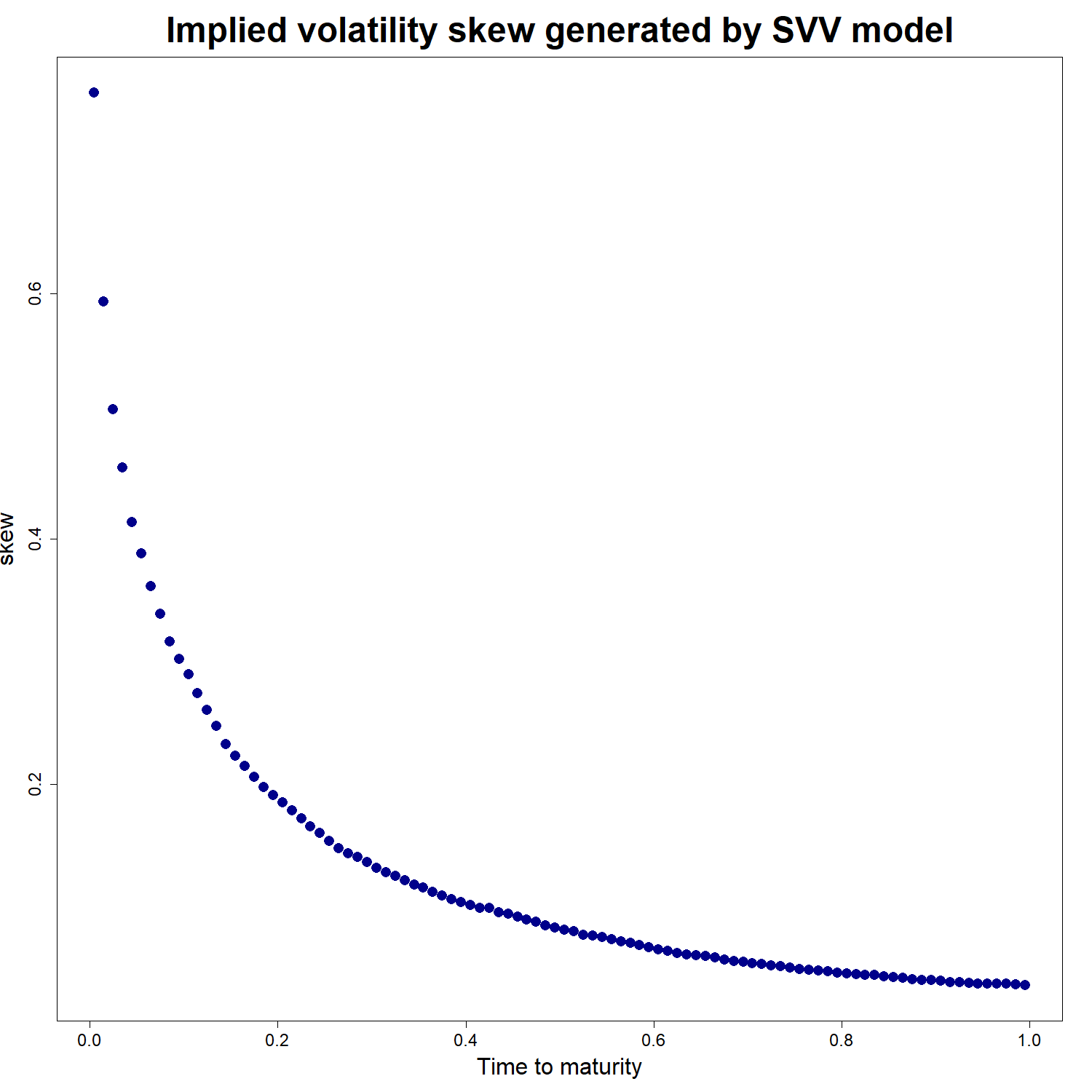}
         \caption{Standard scale}
     \end{subfigure}
     \hfill
     \begin{subfigure}[b]{0.48\textwidth}
         \centering
         \includegraphics[width=\textwidth]{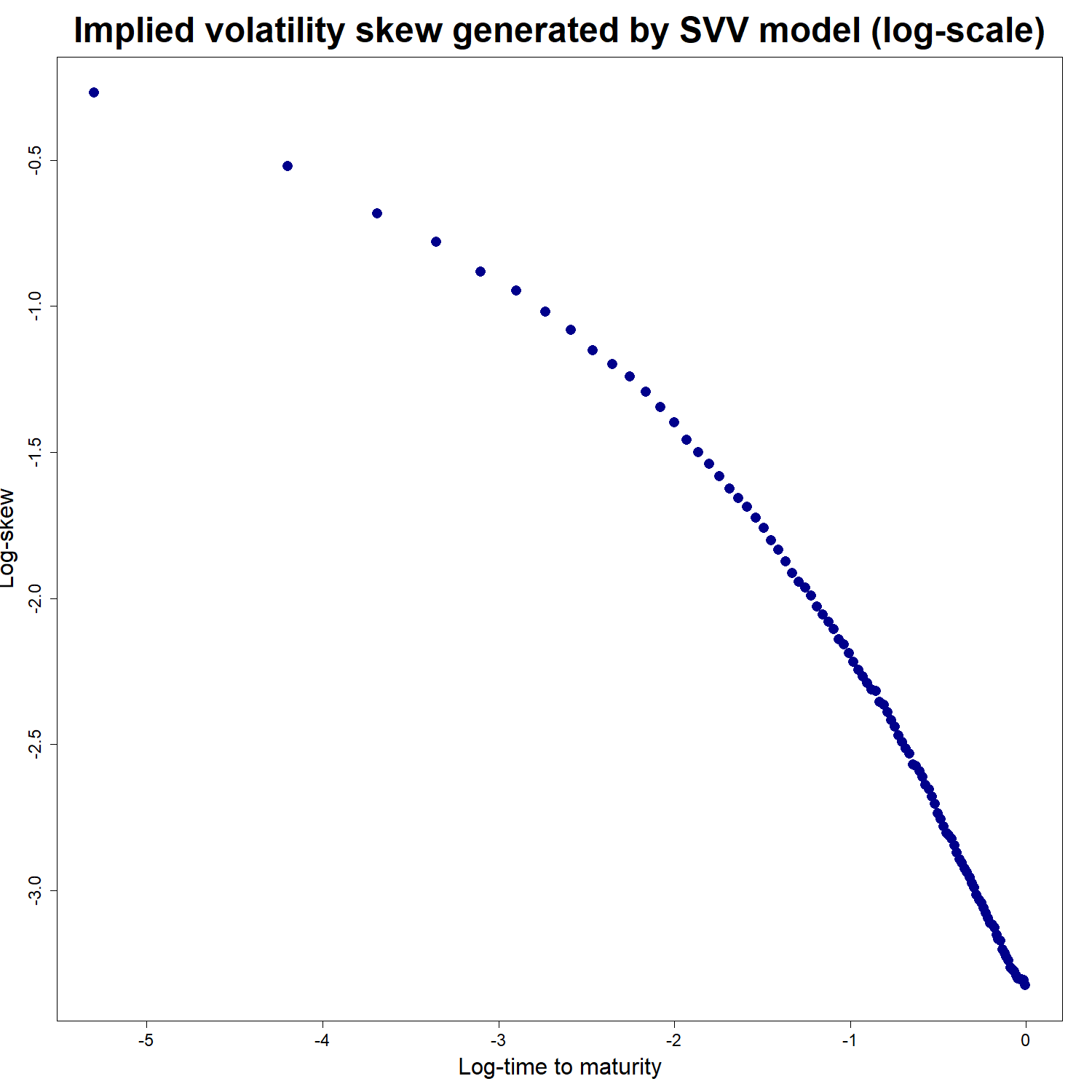}
         \caption{Logarithmic scale}
     \end{subfigure}
    \caption{Absolute at-the-money implied volatility skew \eqref{eq: SV skew} generated by the SVV model \eqref{eq: SVV for simulation of IV} with the rough kernel \eqref{eq: IV rough kernel} on standard (a) and logarithmic (b) scales.}
    \label{fig: rough IV skew}
    \end{figure}
    
    As a final remark, observe that the points on Fig.~\ref{fig: rough IV skew}(b) are not located on a single line; in fact, the linear fit performed over the short maturities substantially differs from the fit over long maturities (see Fig.~\ref{fig: rough change of power law}). Such behavior agrees with the recent empirical study \cite{Delemotte_Marco_Segonne_2023} reporting the same phenomenon for SPX implied volatility skews (c.f. \cite[Figure 2]{Delemotte_Marco_Segonne_2023}).  
    \begin{figure}[h!]
        \centering
        \includegraphics[width = \textwidth]{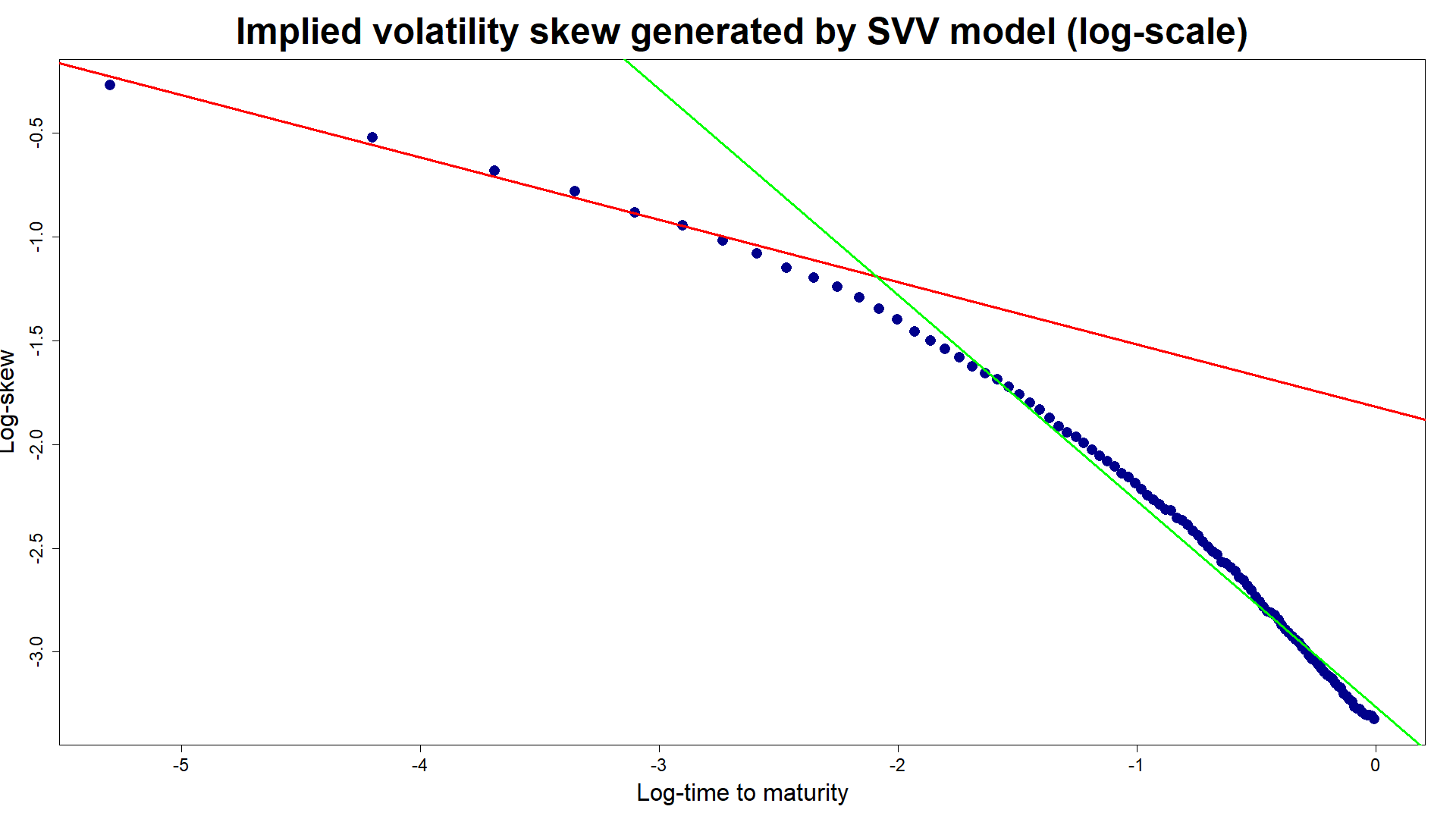}
        \caption{At-the-money implied volatility skew \eqref{eq: SV skew} generated by the SVV model \eqref{eq: SVV for simulation of IV} with the rough kernel \eqref{eq: IV rough kernel} on a logarithmic scale. The lines depict two linear fits: over the short maturities (red line, the slope is -0.3006433 which is consistent with \cite[Theorem 4.8]{DN_YT_power_law_2023}) and over the long maturities (green line, the slope is -0.9920939). Such a behavior is consistent with the recent paper \cite{Delemotte_Marco_Segonne_2023}}
        \label{fig: rough change of power law}
    \end{figure}
\end{example}

\begin{example}
    Let us consider the mixed fractional kernel
    \begin{equation}\label{eq: IV mixed kernel}
        \mathcal K(t,s) = \left( \frac{\sqrt{0.4}}{\Gamma(0.7)} (t-s)^{-0.3} + \frac{\sqrt{1.8}}{\Gamma(1.4)} (t-s)^{0.4}\right)\mathbbm 1_{s<t}
    \end{equation}
    in \eqref{eq: SVV for simulation of IV}. The corresponding implied volatility surface and implied volatility smiles are presented on Figure~\ref{fig: mixed IV surface}. Comparing the results with the purely rough case from Example \ref{ex: purely rough IV}, we note that the the presence of a ``long memory'' component $(t-s)^{0.4}$ in \eqref{eq: IV mixed kernel} keeps the smiles ``less flat'' and more pronounced for larger maturities, which is consistent with the results of \cite{Funahashi_Kijima_2017, Funahashi_Kijima_2017-1}. 
    \begin{figure}
     \centering
     \begin{subfigure}[b]{0.48\textwidth}
         \centering
         \includegraphics[width=\textwidth]{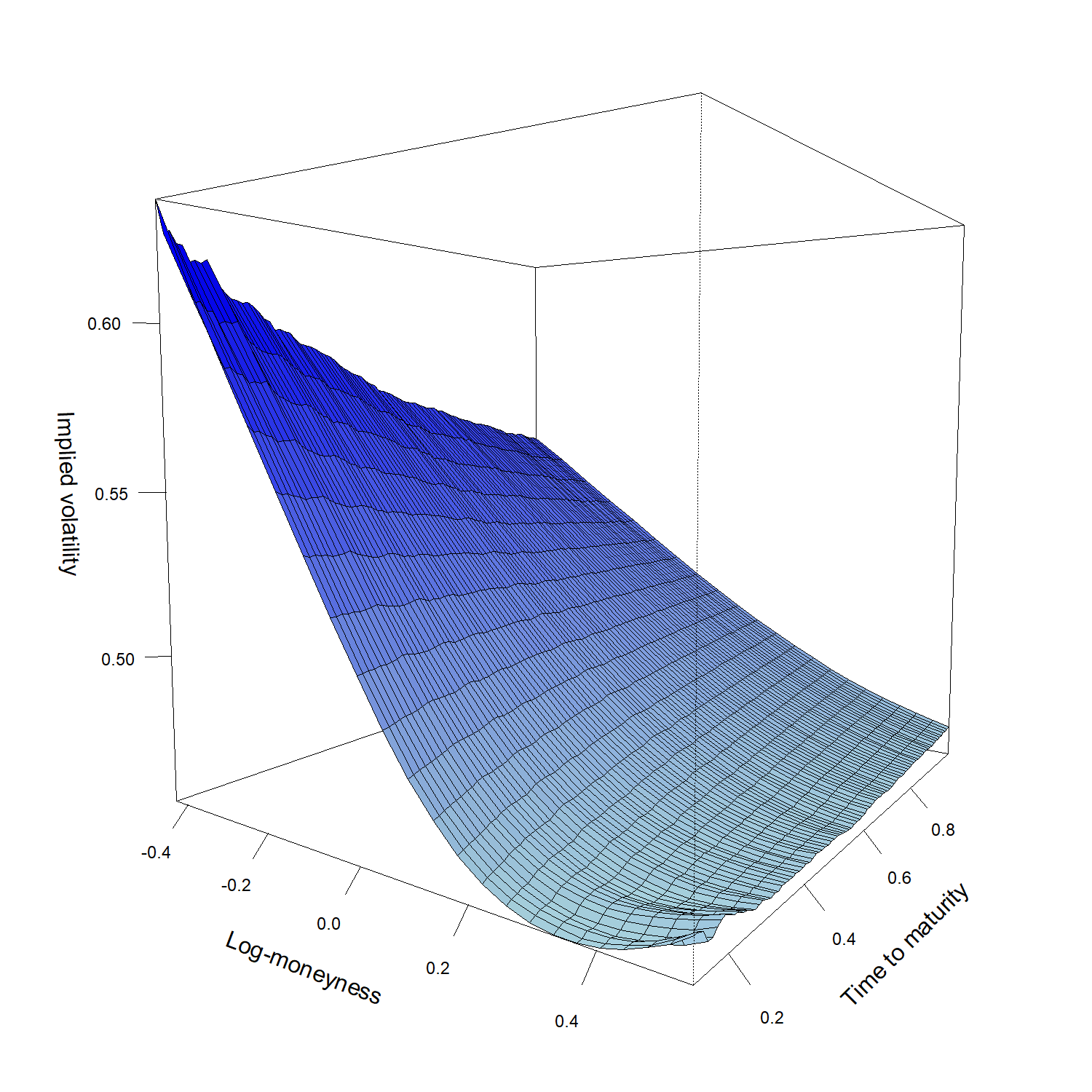}
         \caption{Implied volatility surface}
     \end{subfigure}
     \hfill
     \begin{subfigure}[b]{0.48\textwidth}
         \centering
         \includegraphics[width=\textwidth]{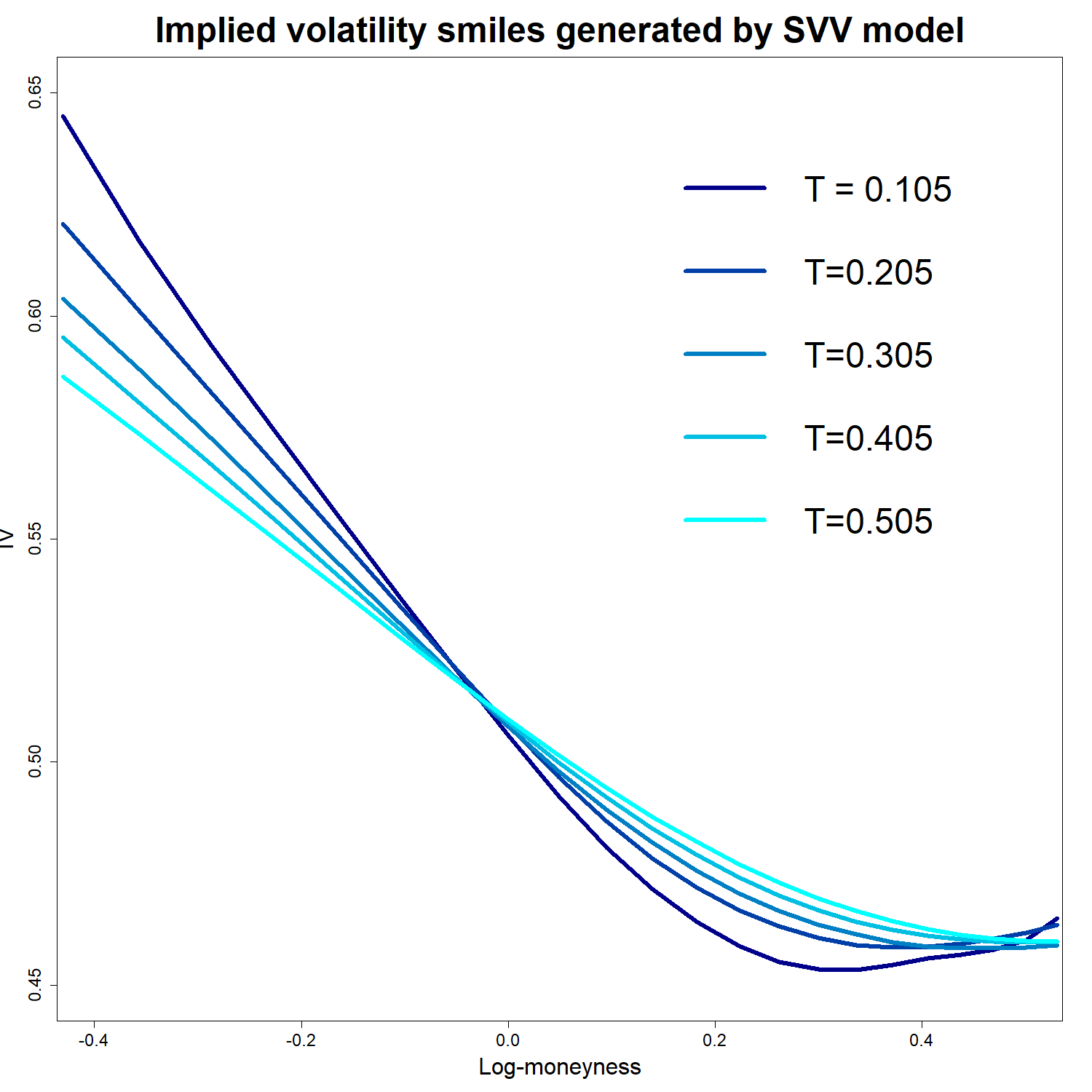}
         \caption{Absolute skews, standard scale}
     \end{subfigure}
    \caption{Implied volatility surface (a) and implied volatility smiles (b) generated by the SVV model \eqref{eq: SVV for simulation of IV} with the mixed fractional kernel \eqref{eq: IV mixed kernel}. Just as in the purely rough case, the absolute at-the-money skew increases for small maturities, which is consistent with the real-life market behavior. Note that the presence of the term $(t-s)^{0.4}$ in the kernel ensures that the smile ``flattens out'' for larger maturities much slower than in the purely rough case, which agrees with the results of \cite{Funahashi_Kijima_2017, Funahashi_Kijima_2017-1}}\label{fig: mixed IV surface}
    \end{figure}
    In addition, just like in the purely rough case from Example \ref{ex: purely rough IV}, the explosive behavior of the skew as $T\to 0$ is guaranteed by Theorem \ref{th: SVV power law} which implies
    \[
        \left|\frac{\partial \widehat \sigma_{\text{log-m}}}{\partial \kappa}(T,0)\right| = O(T^{-0.3}), \quad T\to 0,
    \] 
    see Figure~\ref{fig: mixed IV skew} below. In other words, the mixed kernel indeed allows to keep the power law of the at-the-money skew as $T\to 0$, and increases the curvature of the smile for larger values of $T$.
    
    \begin{figure}[h!]
     \centering
     \begin{subfigure}[b]{0.48\textwidth}
         \centering
         \includegraphics[width=\textwidth]{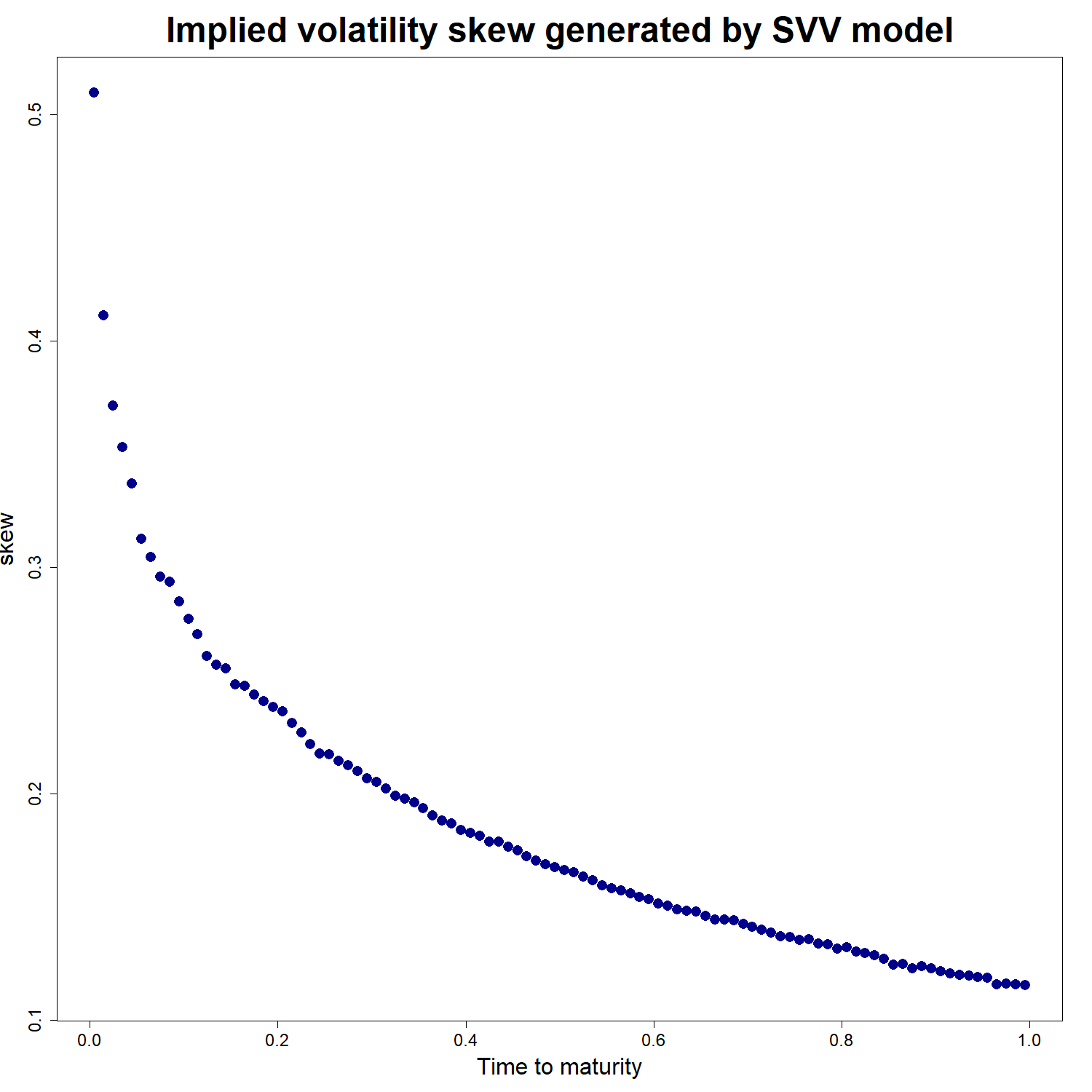}
         \caption{Standard scale}
     \end{subfigure}
     \hfill
     \begin{subfigure}[b]{0.48\textwidth}
         \centering
         \includegraphics[width=\textwidth]{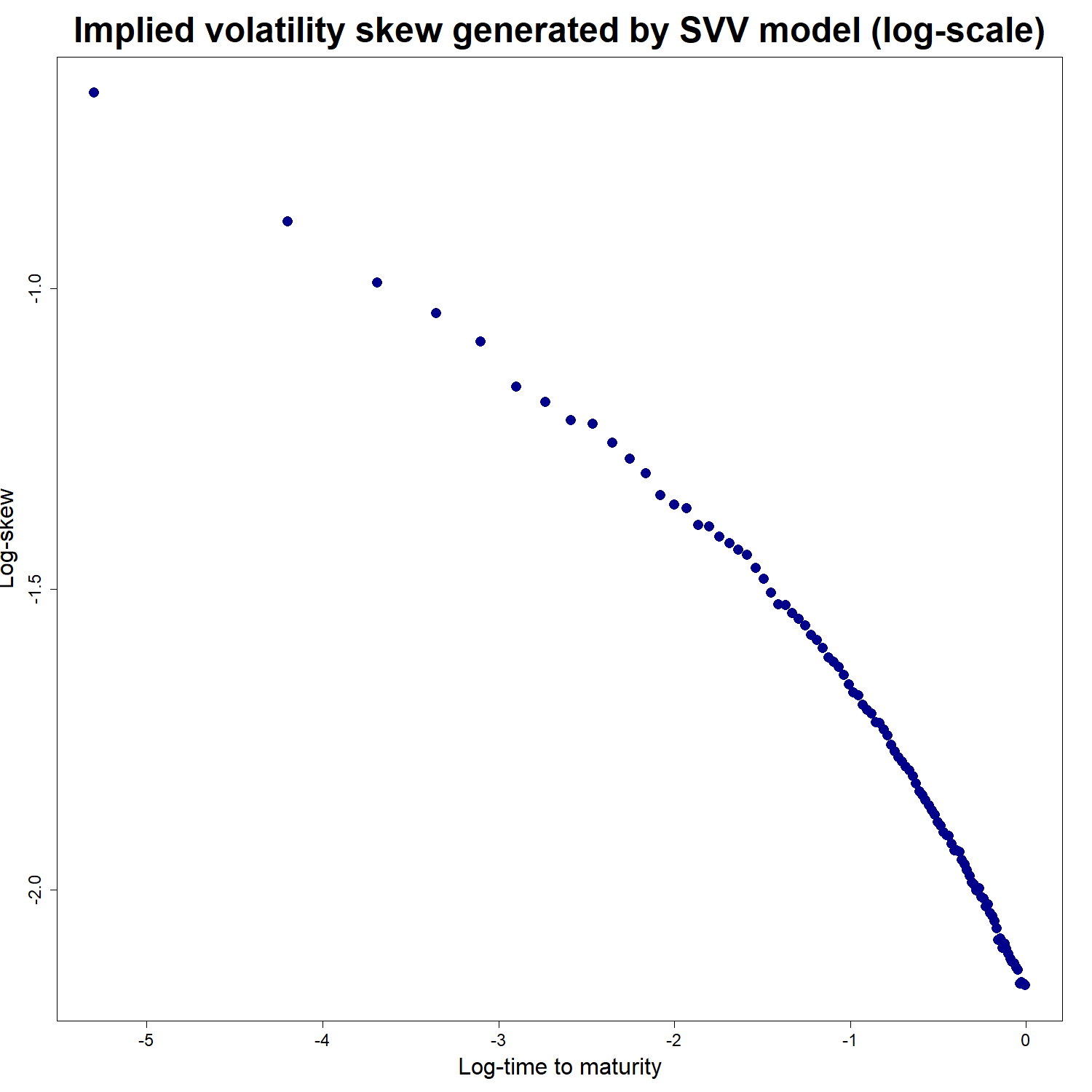}
         \caption{Logarithmic scale}
     \end{subfigure}
    \caption{Absolute at-the-money implied volatility skew \eqref{eq: SV skew} generated by the SVV model \eqref{eq: SVV for simulation of IV} with the mixed kernel \eqref{eq: IV mixed kernel} on standard (a) and logarithmic (b) scales.}
    \label{fig: mixed IV skew}
    \end{figure}
\end{example}

\section{Malliavin differentiability of volatility and price processes and absolute continuity of their laws}\label{sec: Malliavin}

\subsection{Malliavin differentiability of volatility and price w.r.t. the initial Brownian motion $W$}\label{ss: Malliavin w.r.t. original BM}

This Subsection is dedicated to the Malliavin differentiability of the volatility and price processes within the SVV model. Our strategy is as follows: first of all, we use the classical characterization of the Malliavin derivative in terms of the stochastic Gateaux derivative and prove the differentiability of sandwiched processes $Y_i$, $i=1,...,d$. Then, we prove an appropriate Malliavin chain rule and apply it to the representation \eqref{eq: solution to the price equation} to obtain the differentiability of $S_i$, $i=1,...,d$.

\paragraph{Malliavin differentiability of volatility.} In order to introduce the concept of the Malliavin differentiability, one must first fix the underlying isonormal Gaussian process and an obvious choice is the initial Brownian motion $W$. Recall that $(\Omega, \mathcal F, \mathbb P)$ is the canonical $2d$-dimensional Wiener space, $W = (W_1,...,W_{2d})$ is the associated $2d$-dimensional Brownian motion and stochastic volatility processes $Y_i = \{Y_i(t),~t\in[0,T]\}$, $i=1,...,d$, are given by
\[
    Y_i(t) = Y_i(0) + \int_0^t b_i(s, Y_i(s)) ds + Z_i(t), \quad t\in[0,T],    
\]
where $Y_i(0) \in (\varphi_i(0), \psi_i(0))$, each $b_i$ satisfies Assumption \ref{assum: assumption on sandwiched drift}, $Z_i(t) = \int_0^t \mathcal K_i(t,s) dB^Y_i$ is a Gaussian Volterra process with the kernel $\mathcal K_i$ satisfying Assumption \ref{assum: assumptions on kernels} and
\[
    B^Y_i := \sum_{j=1}^{2d} \ell_{d+i, j} W_j.
\]

Denote also by $\mathcal H$ the associated Cameron-Martin space, i.e. the space of all continuous functions $F=(F_1,...,F_{2d})\in C_0([0,T];\mathbb R^{2d})$ such that each $F_i$, $i=1,...,2d$, can be represented as
\[
    F_i(\cdot) = \int_0^\cdot f_i(s)ds,
\]
where $f=(f_1,...,f_{2d})\in L^2([0,T];\mathbb R^{2d})$. For the relation between $F$ and $f$ specified above, we will use the notation $F = \int_0^\cdot f(s)ds$. By ${\mathbb D}^{1,2}$, we denote the space of Malliavin differentiable functions (w.r.t. the $L^2$-closure).

Before proceeding to the main results of this Subsection, let us give an important auxiliary theorem.

\begin{theorem}[\cite{Sugita1985}, Theorem 3.1]\label{th: Sugita}
    A random variable $\eta \in L^2(\Omega, \mathcal F_T, \mathbb P)$ belongs to ${\mathbb D}^{1,2}$ w.r.t. $W$ if and only if the following two conditions are satisfied:
    \begin{enumerate}
        \item[(i)] $\eta$ is \emph{ray absolutely continuous}, i.e. for any $F \in \mathcal H$ there exists a version of the process $\{\eta(\omega + \varepsilon F),~\varepsilon \ge 0\}$ that is absolutely continuous;
        \item[(ii)] $\eta$ is \emph{stochastically Gateaux differentiable}, i.e. there exists a random vector 
        \[
            D\eta \in L^2\left(\Omega; L^2([0,T];\mathbb R^{2d})\right)
        \]
        such that for any $F = \int_0^\cdot f(s)ds \in \mathcal H$, $f\in L^2([0,T];\mathbb R^{2d})$,
        \[
            \frac{1}{\varepsilon} \left( \eta\left(\omega + \varepsilon F \right) - \eta(\omega) \right) \xrightarrow{\mathbb P} \langle D\eta, f \rangle_{L^2\left([0,T];\mathbb R^{2d}\right)}, \quad \varepsilon \to 0.
        \]
    \end{enumerate}
    In this case the random vector $D\eta$ from (ii) is the Malliavin derivative of $\eta$.
\end{theorem}

In order to apply Theorem \ref{th: Sugita} to obtain the Malliavin differentiability of $Y_i(t)$, $i=1,...,d$, $t\in[0,T]$, we first need to formally identify 
\[
    Y_i\left(\omega + \varepsilon F , t\right), \quad i=1,...,d, \quad t\in[0,T],
\]
for any $F= \int_0^\cdot f(s)ds \in \mathcal H$. For the given $\omega$, it is reasonable to associate these expressions with random variables $Y^{\varepsilon, F}_i(\omega, t)$ given by
\[
    Y^{\varepsilon, F}_i(\omega, t) = Y(0) + \int_0^t b_i\left(s, Y^{\varepsilon, F}_i(\omega, s)\right)ds + Z_i^{\varepsilon, F}\left(\omega, t\right), \quad t\in [0,T],
\]
where $Z_i^{\varepsilon, F}\left(\omega, t\right) := Z_i\left(\omega + \varepsilon F, t\right)$. Indeed, equations \eqref{eq: volatility process equations} are treated pathwisely and each $Z_i^{\varepsilon, F}\left(\omega, t\right)$ is H\"older continuous w.r.t. $t$ for all $\omega\in\Omega$. This means that the equation above has a unique solution for all $\omega\in\Omega$ and
\[
    Y^{\varepsilon, F}_i(\omega , t) = Y_i\left(\omega + \varepsilon F , t\right), \quad \omega \in \Omega, \quad t\in [0,T].
\]

However, in order to apply Theorem \ref{th: Sugita}, one needs to get a more convenient representation of $Z_i^{\varepsilon, F}\left(\omega, t \right)$ that ``splits'' the pieces corresponding to $\omega$ and $F$. Intuitively, it is clear that a natural candidate for such a representation is
\[
    Z_i^{\varepsilon, F}\left(\cdot, t\right) = Z_i(\cdot, t) + \varepsilon \sum_{j=1}^{2d} \ell_{d+i, j} \int_0^t \mathcal K_i (t,s) f_j(s) ds, \quad t\in[0,T],
\]
and this guess is formalized in the next simple proposition.

\begin{proposition}\label{prop: modification of shifted noise}
    Let $\varepsilon \ge 0$ and $F = \int_0^\cdot f(s)ds \in \mathcal H$. Then there exists $\Omega_{\varepsilon, F}$, $\mathbb P(\Omega_{\varepsilon, F}) = 1$, such that for any $i=1,...,d$ and $\omega\in\Omega_{\varepsilon, F}$
    \[
        Z_i^{\varepsilon, F}\left(\omega, t\right) = Z_i(\omega, t) + \varepsilon \sum_{j=1}^{2d} \ell_{d+i, j} \int_0^t \mathcal K_i (t,s) f_j(s) ds, \quad t\in[0,T].
    \]
\end{proposition}

\begin{proof}
    By definition, for any $t\in[0,T]$
    \begin{align*}
        Z_i(\cdot, t) & = L^2\text{-}\lim_{n\to\infty} \sum_{l=0}^{k_n -1} K_i^{n, l}(t) \left(B^Y_i(\cdot, t_{n, l+1}) - B^Y_i(\cdot, t_{n,l})\right) 
        \\
        &= L^2\text{-}\lim_{n\to\infty}  \sum_{j=1}^{2d} \ell_{d+i,j} \sum_{l=0}^{k_n -1} K_i^{n, l}(t) \left(W_j(\cdot, t_{n, l+1}) - W_j(\cdot, t_{n,l})\right) 
        \\
        &=: L^2\text{-}\lim_{n\to\infty} Z^n_i(\cdot, t),
    \end{align*}
    where $\pi^n = \{0 = t_{n,0} < t_{n, 1}<...<t_{n, k_n} = t\}$, $n\ge 1$, is a system of partitions such that $\sup_{l=0,...,k_n-1} (t_{n,l+1} - t_{n,l}) \to 0$ as $n\to\infty$ and $\mathcal K^n_i(t, s) :=  \sum_{l=0}^{k_n -1} K^{n, l}_i(t) \mathbbm 1_{[t_{n,l}, t_{n,l+1})}(s)$ is such that $\int_0^t \left(\mathcal K_i(t, s) - \mathcal K^n_i(t, s) \right)^2 ds \to 0$, $n\to\infty$. 
    
    By the Riesz theorem, there exist $\Omega'_{t} \subset \Omega$, $\mathbb P(\Omega'_{t}) = 1$, and a subsequence of partitions that will again be denoted by $\{\pi^n,~n\ge 1\}$ such that for all $i=1,...,d$ and $\omega\in\Omega'_t$
    \begin{equation}\label{proofeq: as. conv of approx noise}
        Z_i(\omega, t) = \lim_{n\to\infty} Z^n_i(\omega, t).
    \end{equation}
    Moreover, note that, by the definition of the probability space and Brownian motion $W$, for any $\omega = (\omega_1,...,\omega_{2d}) \in \Omega = C_0([0,T];\mathbb R^{2d})$
    \begin{align*}
        Z^n_i(\omega, t) &= \sum_{j=1}^{2d} \ell_{d+i,j} \sum_{l=0}^{k_n -1} K_i^{n, l}(t) \left(W_j(\omega, t_{n, l+1}) - W_j(\omega, t_{n,l})\right) 
        \\
        &=\sum_{j=1}^{2d} \ell_{d+i,j} \sum_{l=0}^{k_n -1} K_i^{n, l}(t) \left(\omega_j(t_{n, l+1}) - \omega_j(t_{n,l})\right),
    \end{align*}
    whence, for all $\omega \in \Omega$,
    \[
        Z^n_i(\omega + \varepsilon F, t) = Z^n_i(\omega) + \varepsilon \sum_{j=1}^{2d} \ell_{d+i,j} \sum_{l=0}^{k_n -1} K_i^{n, l}(t) \left(F_j(t_{n, l+1}) - F_j(t_{n,l})\right)
    \]
    and thus, for all $\omega \in \Omega'_t$,
    \begin{equation}\label{proofeq: representation of the sh noise approx}
    \begin{aligned}
        \lim_{n\to\infty} Z^n_i(\omega + \varepsilon F, t) &= \lim_{n\to\infty} Z^n_i(\omega) + \lim_{n\to\infty} \varepsilon \sum_{j=1}^{2d} \ell_{d+i,j} \sum_{l=0}^{k_n -1} K_i^{n, l}(t) \left(F_j(t_{n, l+1}) - F_j(t_{n,l})\right)
        \\
        &=Z_i(\omega, t) + \varepsilon \sum_{j=1}^{2d} \ell_{d+i, j} \int_0^t \mathcal K_i (t,s) f_j(s) ds.
    \end{aligned}
    \end{equation}
    Now, denote $\Omega_{\varepsilon, F, t} := \{\omega\in\Omega'_t:~\omega + \varepsilon F \in \Omega'_t\}$. Then $\mathbb P(\Omega_{\varepsilon, F, t}) = 1$ since, by the Cameron-Martin theorem,
    \begin{align*}
        \mathbb P(\Omega_{\varepsilon, F, t}) & = \int_{\Omega} \mathbbm 1_{\Omega'_t}( \omega + \varepsilon F) d\mathbb P(\omega) = \E\left[ \mathbbm 1_{\Omega'_t} \exp\left\{\varepsilon \mathcal I_f - \frac{\varepsilon^2}{2} \lVert \mathcal I_f \rVert^2_{L^2(\Omega)}\right\} \right] = 1,
    \end{align*}
    where $\mathcal I_f := \sum_{j=1}^{2d} \int_0^T f_j(s) dW_j(s)$, and therefore, taking into account \eqref{proofeq: as. conv of approx noise} and \eqref{proofeq: representation of the sh noise approx}, for any $\omega \in \Omega_{\varepsilon, F, t}$
    \begin{equation}\label{proofeq: pointwise representation of Z}
        Z_i(\omega + \varepsilon F, t) = Z_i(\omega, t) + \varepsilon \sum_{j=1}^{2d} \ell_{d+i, j} \int_0^t \mathcal K_i (t,s) f_j(s) ds.
    \end{equation}
    Finally, it remains to notice that, for the fixed $\varepsilon \ge 0$ and $F \in \mathcal H$, both left- and right-hand sides of \eqref{proofeq: pointwise representation of Z} are continuous w.r.t. $t$, whence we can put
    \[
        \Omega_{\varepsilon, F} := \bigcap_{t\in \mathbb Q\cap [0,T]} \Omega_{\varepsilon, F, t}.
    \]
\end{proof}

Now we are ready to prove the key result of the Subsection. We will use the same approach as in \cite[Theorem 3.3]{Hu2008} and check the conditions of Theorem \ref{th: Sugita}.

\begin{theorem}\label{th: Malliavin differentiability}
    Let Assumptions \ref{assum: assumptions on kernels} and \ref{assum: assumption on sandwiched drift} hold. Then, for all $i=1,...,d$ and $t\in[0,T]$, $Y_i(t) \in {\mathbb D}^{1,2}$ and
    \begin{equation}\label{eq: Malliavin derivative of Y}
    \begin{gathered}
        D Y_i(t) = ( D^1 Y_i(t), ..., D^{2d} Y_i(t) ) \in L^2\left(\Omega\times [0,T]; \mathbb R^{2d}\right),
        \\
        D^j_u Y_{i}(t) := \ell_{d+i, j}\left( \mathcal K_i (t,u) + \int_u^t\mathcal K_i (s,u) \frac{\partial b_i}{\partial y}(s, Y_i(s))e^{ \int_s^t \frac{\partial b_i}{\partial y}(v, Y_i(v)) dv }   ds     \right) \mathbbm 1_{[0,t]}(u),
    \end{gathered}
    \end{equation}
    with $\ell_{d+i, j}$, $i=1,...,d$, $j=1,...,2d$ being the corresponding elements of the matrix $\Lc$. 
\end{theorem}

\begin{proof} 
    We will split the proof into three steps. First, we verify that $D Y_i(t)$ given by \eqref{eq: Malliavin derivative of Y} indeed satisfies
    \[
        \sum_{j=1}^{2d} \mathbb E \int_0^T \left(D^j_u Y_i(t)\right)^2 du < \infty
    \]
    and then check the conditions of Theorem \ref{th: Sugita}.
    
    \noindent\textbf{Step 1: correctness of the form of derivative.} 
    By Remark \ref{rem: properties of Y}, for all $s\in[0,T]$: $(s, Y_i(s)) \in \mathcal D^i_{\frac{1}{\xi_i(0)},\frac{1}{\xi_i(0)}}$, where 
    \[
        \xi_i(0) := \frac{( L_{2,i} + \Lambda_{\lambda, i} )^{\alpha_i}}{L_{1,i}}
    \]
    with $\Lambda_{\lambda, i}$ being defined by \eqref{eq: definition of Lambda} and $L_{1,i}$, $L_{2,i}$, $\alpha_i$ being some deterministic positive constants that depend only on $\lambda \in \left(\frac{1}{\gamma_i +1 }, H_i\right)$ and the shape of the drift $b_i$. Whence, due to Assumption \ref{assum: assumption on sandwiched drift}(ii),
    \[
        \left|\frac{\partial b_i}{\partial y}(s, Y_i(s))\right| \le c\xi_i^p(0), \quad s\in[0,T].
    \]
    Moreover, by Assumption \ref{assum: assumption on sandwiched drift}(iv),
    \[
        e^{ \int_s^t \frac{\partial b_i}{\partial y}(v, Y_i(v)) dv } \le e^{cT}.
    \]
    Therefore
    \begin{align*}
        \mathbb E \int_0^T \left(D^j_u Y_i(t)\right)^2 du & \le 2\ell_{d+i, j}^2\int_0^t \mathcal K^2_i(t, u) du 
        \\
        &\quad + 2T c^2 e^{2cT}\ell_{d+i, j}^2 \mathbb E\left[\xi_i^{2p}(0)\right]\int_0^t  \int_u^t \mathcal K_i^2(s, u)dsdu 
        \\
        &<\infty,
    \end{align*}
    because $\int_0^t \mathcal K_i^2 (t,s)ds < \infty$ by Assumption \ref{assum: assumptions on kernels}(i) and the random variable $\xi_i(0)$ has moments of all orders.
    
    \noindent\textbf{Step 2: ray absolute continuity.} Let $F\in\mathcal H$ be fixed. Then, as explained above, the random process $\{Y^{\varepsilon, F}_i(t),~\varepsilon \ge 0\}$ defined for every $\varepsilon \ge 0$ by the equation
    \begin{equation}\label{proofeq: Malliavin differentiability, modification of Y}
        Y^{\varepsilon, F}_i(t) = Y_0 + \int_0^t b_i(s, Y^{\varepsilon, F}_i(s))ds + Z_i(t) + \varepsilon \sum_{j=1}^{2d} \ell_{d+i, j} \int_0^t \mathcal K_i (t,s) f_j(s) ds
    \end{equation}
    is a version of the process $\{Y_i(\omega + \varepsilon F, t),~\varepsilon \ge 0\}$. Let us prove that this version is Lipschitz w.r.t. $\varepsilon$ and hence absolutely continuous. First, fix $\lambda \in \left(\frac{1}{\gamma_i +1}, H_i\right)$, where $H_i$ is from Assumption \ref{assum: assumptions on kernels}(ii) and $\gamma_i$ is from Assumption \ref{assum: assumption on sandwiched drift}(iii), and take an arbitrary $\varepsilon^*>0$. Then, by Remark \ref{rem: choice of modification of Z} and Assumption \ref{assum: assumptions on kernels}(ii), for any $\varepsilon \in [0,\varepsilon^*]$ 
    \begin{align*}
        &|Z_i(t_1) - Z_i(t_2)| + \varepsilon\left|\sum_{j=1}^{2d} \ell_{d+i, j} \int_0^{t_1} \mathcal K_i (t_1,s) f_j(s) ds - \sum_{j=1}^{2d} \ell_{d+i, j} \int_0^{t_2} \mathcal K_i (t_2,s) f_j(s) ds\right|
        \\
        &\qquad \le \Lambda_{\lambda, i} |t_1 - t_2|^\lambda + \varepsilon \left( \int_0^{t_1 \vee t_2} (\mathcal K_i(t_1,s) - \mathcal K_i(t_2,s))^2ds \right)^{\frac{1}{2}} \Big\lVert \sum_{j=1}^{2d} \ell_{d+i,j} f_j \Big\rVert_{L^2([0,T])}
        \\
        &\qquad\le \left( \Lambda_{\lambda, i} + \varepsilon^* C_\lambda \Big\lVert \sum_{j=1}^{2d} \ell_{d+i,j} f_j \Big\rVert_{L^2([0,T])} \right) |t_1 - t_2|^\lambda,
    \end{align*}
    which implies, by Remark \ref{rem: properties of Y}, that $(t, Y^{\varepsilon, F}_i(t)) \in \mathcal D^i_{\frac{1}{\xi_i(\varepsilon^*)}, \frac{1}{\xi_i(\varepsilon^*)}}$ for any $\varepsilon \in [0,\varepsilon^*]$ and $t\in[0,T]$, where $\mathcal D^i_{\cdot,\cdot}$ is defined by \eqref{eq: definition of the set D} and
    \begin{equation}\label{proofeq: Malliavin differentiability, xi}
        \xi_i(\varepsilon^*) := \frac{\left(L_{2,i} + \Lambda_{\lambda, i} + \varepsilon^* C_\lambda \Big\lVert \sum_{j=1}^{2d} \ell_{d+i,j} f_j \Big\rVert_{L^2([0,T])} \right)^{\alpha_i}}{L_{1,i}}
    \end{equation}
    for some deterministic positive constants $L_{1,i}$, $L_{2,i}$ and $\alpha_i$ that depend only on $\lambda$ and the shape of the drift $b_i$. This implies, by Assumption \ref{assum: assumption on sandwiched drift}(ii), that, for all $s\in[0,T]$ and $0 \le \varepsilon_1 \le \varepsilon_2 \le \varepsilon^*$,
    \[
        \left|b_i(s,Y^{\varepsilon_2, F}_i(s) ) - b_i(s,Y^{\varepsilon_1, F}_i(s) )\right| \le c\xi_i^p(\varepsilon^*) \left| Y^{\varepsilon_2, F}_i(s) - Y^{\varepsilon_1, F}_i(s) \right|,
    \]
    whence 
    \begin{align*}
        |Y^{\varepsilon_2, F}_i(t) - Y^{\varepsilon_1, F}_i(t)| & \le \int_0^t \left|b_i(s,Y^{\varepsilon_2, F}_i(s) ) - b_i(s,Y^{\varepsilon_1, F}_i(s) )\right| ds
        \\
        &\quad + |\varepsilon_2-\varepsilon_1| \left|\sum_{j=1}^{2d} \ell_{d+i, j} \int_0^t \mathcal K_i (t,s) f_j(s) ds  \right|
        \\
        &\le c\xi_i^p(\varepsilon^*) \int_0^t \left| Y^{\varepsilon_2, F}_i(s) - Y^{\varepsilon_1, F}_i(s) \right|ds 
        \\
        &\quad + |\varepsilon_2-\varepsilon_1| \left( \int_0^T \mathcal K_i^2 (t,s)ds \right)^{\frac{1}{2}} \Big\lVert \sum_{j=1}^{2d} \ell_{d+i,j} f_j \Big\rVert_{L^2([0,T])},
    \end{align*}
    where $\int_0^T \mathcal K_i^2 (t,s)ds < \infty$ by Assumption \ref{assum: assumptions on kernels}(i). Now, it follows from Gronwall's inequality that, for all $0 \le \varepsilon_1 \le \varepsilon_2 \le \varepsilon^*$,
    \[
        |Y^{\varepsilon_2, F}_i(t) - Y^{\varepsilon_1, F}_i(t)| \le \left( \int_0^T \mathcal K_i^2 (t,s)ds \right)^{\frac{1}{2}}  \Big\lVert \sum_{j=1}^{2d} \ell_{d+i,j} f_j \Big\rVert_{L^2([0,T])} e^{c T\xi_i^p(\varepsilon^*)} |\varepsilon_2-\varepsilon_1|,
    \]
    which implies the desired absolute continuity w.r.t. $\varepsilon$.
    
    \noindent\textbf{Step 3: stochastic Gateaux differentiability.} Again, it is sufficient to compute the limit in probability of the form
    \[
        \mathbb P\text{-}\lim_{\varepsilon\to 0} \frac{Y_i^{\varepsilon, F}(t) - Y_i^{0, F}(t)}{\varepsilon},
    \]
    where $Y_i^{\varepsilon, F}(t)$ is defined by \eqref{proofeq: Malliavin differentiability, modification of Y}.
    
    Using the mean-value theorem, we obtain that
    \begin{equation}\label{proofeq: Malliavin differentiability, equation for Yeps-Y}
    \begin{aligned}
        Y_i^{\varepsilon, F}(t) - Y_i^{0, F}(t) &= \int_0^t (b_i(s,Y_i^{\varepsilon, F}(s)) - b_i(s, Y_i^{0, F}(s))ds 
        \\
        &\quad + \varepsilon \sum_{j=1}^{2d} \ell_{d+i, j} \int_0^t \mathcal K_i (t,s) f_j(s) ds
        \\
        &= \int_0^t \Theta_{i,\varepsilon}(s) (Y_i^{\varepsilon, F}(s) - Y_i^{0, F}(s))ds 
        \\
        &\quad + \varepsilon \sum_{j=1}^{2d} \ell_{d+i, j} \int_0^t \mathcal K_i (t,s) f_j(s) ds,
    \end{aligned}
    \end{equation}
    where $\Theta_{i,\varepsilon}(s) := \frac{\partial b_i}{\partial y} \left( s,  Y_i^{0, F}(s) + \theta_{i, \varepsilon}(s)(Y_i^{\varepsilon, F}(s) - Y_i^{0, F}(s))\right)$ with $\theta_{i, \varepsilon}(s)$ being some value between 0 and 1. It is easy to verify that \eqref{proofeq: Malliavin differentiability, equation for Yeps-Y} implies the representation of the form
    \[
        Y_i^{\varepsilon, F}(t) - Y_i^{0, F}(t) = \varepsilon \sum_{j=1}^{2d} \int_0^t \exp\left\{ \int_s^t \Theta_{i,\varepsilon}(v) dv \right\} d G_{i,j}(s),
    \]
    where
    \[
        G_{i,j}(t) := \ell_{d+i,j} \int_0^t \mathcal K_i (t,s) f_j(s) ds.
    \]
    Note that, just like in Step 1, it follows from Remark \ref{rem: properties of Y} that for all $\varepsilon \in [0, 1]$ and $s\in[0,T]$ $(s, Y^{\varepsilon, F}_i(s)) \in \mathcal D^i_{\frac{1}{\xi_i(1)}, \frac{1}{\xi_i(1)}}$, where $\xi_i(1)$ is defined via \eqref{proofeq: Malliavin differentiability, xi}. This, together with continuity of $\frac{\partial b_i}{\partial y}$ due to Assumption \ref{assum: assumption on sandwiched drift}(iv) and compactness of $\overline{\mathcal D^i_{\frac{1}{\xi_i(1)}, \frac{1}{\xi_i(1)}}}$, implies that there exists a finite positive random variable $\eta_i$ such that $|\Theta_{i,\varepsilon}(s)| \le \eta_i $ for all $\varepsilon \in[0,1]$ and $s\in[0,T]$. Moreover, the mapping $s \mapsto \exp\left\{ \int_s^t \Theta_{i,\varepsilon}(v) dv \right\}$, $s\in[0,t]$, is Lipschitz and whence the integrals $\int_0^t \exp\left\{ \int_s^t \Theta_{i,\varepsilon}(v) dv \right\} d G_{i,j}(s)$, $j=1,...,2d$, are well-defined as pathwise limits of Riemann-Stieltjes integral sums due to the classical results of Young (\cite[Section 10]{Young1936}, see also \cite[Chapter 6]{FV2010} for a more recent overview on the subject).
    
    Now, applying the integration by parts formula, we obtain
    \begin{align*}
        \frac{Y_i^{\varepsilon, F}(t) - Y_i^{0, F}(t)}{\varepsilon} &= \sum_{j=1}^{2d} \int_0^t \exp\left\{ \int_s^t \Theta_{i,\varepsilon}(v) dv \right\} d G_{i,j}(s)
        \\
        & =\sum_{j=1}^{2d} \left(G_{i,j}(t) - \int_0^t G_{i,j}(s) d\left(\exp\left\{ \int_s^t \Theta_{i,\varepsilon}(v) dv \right\}\right)\right)
        \\
        & =\sum_{j=1}^{2d} \left(G_{i,j}(t) + \int_0^t G_{i,j}(s) \Theta_{i,\varepsilon}(s)\exp\left\{ \int_s^t \Theta_{i,\varepsilon}(v) dv \right\} ds\right),
    \end{align*}
    and, taking into account the definition of $G_{i,j}$,
    \begin{align*}
        \frac{Y_i^{\varepsilon, F}(t) - Y_i^{0, F}(t)}{\varepsilon} &= \sum_{j=1}^{2d} \int_0^T D_{i,j}(\varepsilon, t, u) f_j(u) du
    \end{align*}
    with
    \[
        D_{i,j}(\varepsilon, t, u) := \ell_{d+i, j}\left( \mathcal K_i (t,u) + \int_u^t\mathcal K_i (s,u) \Theta_{i,\varepsilon}(s)\exp\left\{ \int_s^t \Theta_{i,\varepsilon}(v) dv \right\}   ds     \right) \mathbbm 1_{[0,t]}(u).
    \]
    Now, note that for all $s\in[0,T]$ and $\varepsilon \in [0,1]$
    \[
        |\Theta_{i,\varepsilon}(s)|\exp\left\{ \int_s^t \Theta_{i,\varepsilon}(v) dv \right\} < \eta_i \exp\{cT\},
    \]
    where $c$ is from Assumption \ref{assum: assumption on sandwiched drift}(iv), and therefore we can apply the dominated convergence theorem to obtain that
    \[
        \frac{Y_i^{\varepsilon, F}(t) - Y_i^{0, F}(t)}{\varepsilon} \to \langle DY_i(t), f \rangle_{L^2([0,T];\mathbb R^{2d})} \quad a.s., \quad \varepsilon \to 0,
    \]
    where $DY_i(t)$ has exactly the form \eqref{eq: Malliavin derivative of Y}, which finishes the proof.
\end{proof}

\begin{remark}\hspace{10cm}
    \begin{enumerate}
        \item It is possible to prove higher order Malliavin differentiability of $Y$ provided that the drift coefficient $b$ is regular enough. For more details, see \cite{DN_YT_power_law_2023}.

        \item Theorem \ref{th: Malliavin differentiability} can be used to prove that the SVV model can reproduce the power law of the implied volatility skew provided an appropriate choice of the Volterra kernels in volatility. We prove this result in \cite[Theorem 4.8]{DN_YT_power_law_2023}.
    \end{enumerate}    
\end{remark}

\paragraph{Malliavin differentiability of price.} An obvious way to obtain the Malliavin differentiability of the price processes $(S_1,...,S_d)^T$ would be to notice that each $S_i$ has the form
\[
    S_i(t) = e^{X_i(t)}, \quad t\in[0,T],
\]
where $X_i$ is the log-price process that, due to Theorem \ref{th: properties of S}, can be written as
\begin{equation}\label{eq: log-price}
    X_i(t) := \log S_i(t) = X_i(0) + \int_0^t \left( \mu_i(s) - \frac{ Y_i^2(s)}{2} \right)ds + \int_0^t Y_i(s) dB_i^S(s).
\end{equation}
Then, in order to prove that $S_i(t) \in {\mathbb D}^{1,2}$, one can first establish that $X_i(t) \in {\mathbb D}^{1,2}$ and then apply an appropriate chain rule. However, the classical chain rule results (see e.g. \cite[Proposition 1.2.3, Proposition 1.2.4]{Nualart2006}) require either all partial derivatives $\frac{\partial f}{\partial x_j}$ to be bounded (which is not the case for such functions as $f(x) = x^2$ or $f(x) = e^x$) or the law of $\eta$ to be absolutely continuous (which we have not established yet for the log-price $X_i(t)$). Therefore, we give a version of the Malliavin chain rule that fits our needs below\footnote{Note that our version of the chain rule is essentially based on \cite[Lemma A.1]{OcK1991}, but the proof of the latter contains a mistake in the argument.}. 

\begin{proposition}\label{prop: chain rule}
    Let $\eta = (\eta_1,...,\eta_k) \in ({\mathbb D}^{1,2})^k$ and $f\in C^{(1)}(\mathbb R^k;\mathbb R)$ be a real-valued function such that
    \begin{itemize}
        \item[1)] if $\frac{\partial f}{\partial x_j}$ is unbounded, then there exist sequences $\{r_n,~n\ge 1\}$, $\{s_{j,n},~n\ge 1\}$, $r_n, s_{j,n} \to \infty$ as $n\to\infty$, such that \[
            \left| \frac{\partial f}{\partial x_j} (x_1,...,x_k) \right| \ge s_{j,n}
        \]
        implies that $|f(x_1,...,x_k)| \ge r_n$;
        \item[2)] $\mathbb E\left[ |f(\eta)|^2 + \left\lVert \sum_{j=1}^k \frac{\partial f}{\partial x_j} (\eta) D \eta_j \right\rVert^2_{L^2([0,T]; \mathbb R^{2d})} \right] < \infty$.
    \end{itemize}
    Then $f(\eta) \in {\mathbb D}^{1,2}$ and $Df(\eta) = \sum_{j=1}^k \frac{\partial f}{\partial x_j} (\eta) D \eta_j $.
\end{proposition}
\begin{proof}
    According to \cite[Proposition 1.2.3]{Nualart2006}, if all partial derivatives $\frac{\partial f}{\partial x_j}$ are bounded, then the claim of our Proposition is true and the chain rule holds. For the general case, let $\phi \in C^\infty (\mathbb R)$ be compactly supported function such that $\phi(x) = x$ whenever $|x|\le 1$ and $|\phi(x)| \le |x|$ for all $x\in\mathbb R$. 
    
    For any $m\ge 1$ define $f_m(x) := m \phi\left(\frac{f(x)}{m}\right)$ and observe that $f_m$ has bounded partial derivatives. Indeed, if $\frac{\partial f}{\partial x_j}$ is bounded, so is $\frac{\partial f_m}{\partial x_j}$. Otherwise, let $n^*$ be such that $\frac{r_{n^*}}{m} > \max\left\{ - \inf\supp\phi, \sup\supp\phi  \right\}$. Then, for all $(x_1,...,x_k)\in \mathbb R^k$ such that $|\frac{\partial f}{\partial x_j} (x_1,...,x_k)| \ge s_{j,n^*}$, it holds that $\frac{f(x_1,...,x_k)}{m} \notin \supp \phi$ and therefore 
    \begin{align*}
        \frac{\partial f_m}{\partial x_j} (x_1,...,x_k) =  \frac{\partial f}{\partial x_j} (x_1,...,x_k) \phi'\left( \frac{f(x_1,...,x_k)}{m} \right) = 0.
    \end{align*}
    Whence
    \begin{align*}
        \left|\frac{\partial f_m}{\partial x_j} (x_1,...,x_k)\right| =  \left|\frac{\partial f}{\partial x_j} (x_1,...,x_k)\right| \left|\phi'\left( \frac{f(x_1,...,x_k)}{m} \right)\right| \le s_{j,n^*} \max_{z\in\mathbb R} |\phi'(z)|.
    \end{align*}
    
    Note that $|f_m(\eta)| \le |f(\eta)|$ for all $m\ge 1$ and $\lim_{m\to\infty} f_m(\eta) = f(\eta)$ a.s. Moreover,
    \begin{align*}
        |Df_m(\eta)| &= \left| \sum_{j=1}^k \frac{\partial f_m}{\partial x_j} (\eta) D\eta_j \right| \le \max_{z\in\mathbb R} |\phi'(z)| \left|\sum_{j=1}^k \frac{\partial f}{\partial x_j} (\eta) D\eta_j \right|
    \end{align*}
    for all $m\ge 1$ and 
    \[
        \lim_{m\to\infty} Df_m(\eta) = \sum_{j=1}^k \frac{\partial f}{\partial x_j} (\eta) D\eta_j \quad a.s.
    \]
    Therefore, by the dominated convergence theorem,
    \[
        \lim_{m\to\infty} \E\left[ |f_m(\eta) - f(\eta)|^2 + \left\lVert Df_m(\eta) - \sum_{j=1}^k \frac{\partial f}{\partial x_j} (\eta) D \eta_j \right\rVert^2_{L^2([0,T]; \mathbb R^{2d})} \right] = 0,
    \]
    and the result follows by the closedness of the operator $D$.
\end{proof}

\noindent Now we can get the following simple corollary for the latter direct application.

\begin{corollary}\label{cor: Malliavin differentiability of log-price}
    For any $i=1,...,d$ and $t\in[0,T]$, $X_i(t) \in {\mathbb D}^{1,2}$ and its Malliavin derivative $D X_i(t) = (D^1 X_i(t),...,D^{2d} X_i(t))$ has the form
    \[
        D^j_u X_i(t) = \left(-  \int_0^t Y_i(s) D^j_u Y_i(s) ds + \ell_{i,j} Y_i(u) + \int_0^t D^j_u Y_i(s) dB_i^S(s)\right) \mathbbm 1_{[0,t]}(u).    
    \]\end{corollary}
\begin{proof}
    First, note that if $u>t$, then the claim is true due to \cite[Corollary 1.2.1]{Nualart2006}. Let now $u\le t$. By the definition of the Malliavin derivative,
    \begin{equation}\label{proofeq: derivative of deterministic part}
        D^j_u\left[ \int_0^t \mu_i(s) ds \right] = 0.
    \end{equation}
    Next, note that for any $s \in [0,T]$ $Y_i^2(s) \in {\mathbb D}^{1,2}$. Indeed, by Proposition \ref{prop: chain rule}, it is sufficient to check that 
    \[
        \E\left[ Y_i^2(s) + \lVert Y_i(s) DY_i(s) \rVert^2_{L^2([0,T];\mathbb R^{2d})} \right] < \infty
    \]
    which is true since $Y_i(s)$ is bounded. Whence $\frac{1}{2}\int_0^t D_u Y_i^2(s) ds = \int_0^t Y_i(s) D^j_u Y_i(s)ds$, $j=1,...,2d$, these integrals are well defined and so, using the chain rule and closedness of the operator $D$, we can write (see e.g. \cite[Theorem 1.2.2]{Vr2003} or \cite[Theorem 1.2.4]{Veraar_Weis_2016})
    \begin{equation}\label{proofeq: interchange of Malliavin derivative and integral}
        D^j_u\left[ \frac{1}{2} \int_0^t Y_i^2(s) ds \right] = \int_0^t Y_i(s) D^j_u Y_i(s)ds, \quad j=1,...,2d.
    \end{equation}
    Finally, by adaptedness and boundedness of $Y_i$ and \cite[Proposition 1.3.2]{Nualart2006},
    \begin{equation}\label{proofeq: derivative of Ito integral}
    \begin{aligned}
        D^j_u\left[ \int_0^t Y_i(s) dB_i^S(s) \right] & = D^j_u\left[ \sum_{k=1}^{2d} \int_0^t \ell_{i,k} Y_i(s) d W_k(s) \right]
        \\
        & = \ell_{i,j} Y_i(u) + \sum_{k=1}^{2d} \ell_{i,k} \int_0^t D^j_u Y_i(s) d W_k(s)
        \\
        & = \ell_{i,j} Y_i(u) + \int_0^t D^j_u Y_i(s) dB_i^S(s),
    \end{aligned}
    \end{equation}
    as required.
\end{proof}

\subsection{Malliavin differentiability w.r.t. the transformed Brownian motion. Absolute continuity of probability laws}\label{ssec: Malliavin w.r.t. changed BM}

In order to proceed, we will require absolute continuity of the law of $X_i(t)$ w.r.t. the Lebesgue measure. It is well known (see e.g. \cite[Chapter 2]{Nualart2006}) that a sufficient condition for that is
\begin{equation}\label{eq: desire to have non-zero derivative}
    \lVert D_\cdot X_i(t)\rVert_{L^2([0,T]; \mathbb R^{2d})} > 0 \quad a.s.,
\end{equation}
thus it would be sufficient to check \eqref{eq: desire to have non-zero derivative}. However, the form of $D X_i(t)$ obtained in Subsection \ref{ss: Malliavin w.r.t. original BM} is not very convenient for this matter. In order to get the desired properties, we will transform the original Brownian motion $W$ and consider the Malliavin derivative w.r.t. this transformation.

Namely, recall that $B = (B_1,...,B_{2d})^T$ is a $2d$-dimensional Gaussian process and
\[
    B(t) = \left(  \begin{matrix}
                B^S(t)
                \\
                B^Y(t)
            \end{matrix}
    \right)
    = \Lc W(t),
\]
where by $B^S = (B_1^S,...,B_d^S)^T$ and $B^Y = (B_1^Y,...,B^Y_d)^T$ we denote the first and the last $d$ components of $B$ correspondingly and $\Lc = (\ell_{i,j})_{i,j=1}^{2d}$ is a lower-triangular matrix such that $\Lc\Lc^T = \Sigma = (\sigma_{i,j})_{i,j=1}^{2d}$. Now, define $\mathcal U = (u_{i,j})_{i,j = 1}^{2d}$ an \textit{upper} $2d \times 2d$-triangular matrix with positive values on the diagonal such that
\[
    \widetilde W := \mathcal U^{-1} B = \mathcal U^{-1} \Lc W
\]
is a $2d$-dimensional Brownian motion. Let now $\widetilde D = (\widetilde D^{1}, ..., \widetilde D^{2d})$ be the Malliavin derivative with respect to the new Brownian motion $\widetilde W$ and $\widetilde {\mathbb D}^{1,2}$ be its domain. Clearly, the matrix $\mathcal U^{-1} \Lc$ is non-degenerate and hence $W$ and $\widetilde W$ generate the same filtration $\mathbb F$ and, moreover, it is easy to see that $\widetilde {\mathbb D}^{1,2}$  coincides with ${\mathbb D}^{1,2}$. Finally, by construction, $\widetilde W_1$ is independent of the processes $B^S_2$, ..., $B^S_d$, $B^Y_1$, ..., $B^Y_{d}$ and therefore the following result holds.

\begin{theorem}\label{th: new Malliavin differentiability} 
    Let Assumptions \ref{assum: assumptions on kernels} and \ref{assum: assumption on sandwiched drift} hold. Then, for any $t \in (0,T]$,
    \[
        \widetilde D^1_v X_i(t) = \begin{cases}
             u_{1,1} Y_1(v) \mathbbm 1_{[0,t]}(v), &\quad \text{if } i = 1,
             \\
             0, &\quad \text{if } i \ne 1,
        \end{cases}
    \]
    and
    \[
        \widetilde D^1_v S_i(t) = \begin{cases}
             u_{1,1} S_1(t) Y_1(v) \mathbbm 1_{[0,t]}(v), &\quad \text{if } i = 1,
             \\
             0, &\quad \text{if } i \ne 1,
        \end{cases}
    \]
    where $u_{1,1}$ is the upper left element of $\mathcal U$. In particular,  
    \[
        \lVert \widetilde D_\cdot X_1(t) \rVert_{L^2([0,T];\mathbb R^{2d})} > 0 \quad a.s.
    \]
    
    \begin{proof}
        The expression for $\widetilde D^1 X_i(t)$ is obtained straightforwardly by differentiating the right-hand side of \eqref{eq: log-price} and taking into account that
        \[
            B^S_1(t) = \sum_{k=1}^{2d} u_{1,k} \widetilde W_{k}(t)
        \]
        and $\widetilde D^1 Y_i(t) = 0$ due to independence of $\widetilde W_1$ and $B^Y_i$. The expression for $\widetilde D^1 S_i(t)$ can be easily obtained using the chain rule from Proposition \ref{prop: chain rule}.
    \end{proof}
\end{theorem}

\begin{corollary}
    For each $i=1,...,d$ and $t\in(0,T]$, the law of $X_i(t)$ (and therefore of $S_i(t) = e^{X_i(t)}$) has continuous and bounded density.
\end{corollary}
\begin{proof}
    By Theorem \ref{th: new Malliavin differentiability},
    \[
        \lVert \widetilde D_\cdot X_1(t) \rVert_{L^2([0,T];\mathbb R^{2d})} > 0 \quad a.s.
    \]
    and, moreover, $\widetilde D X_1(t)$ is adapted to the filtration generated by $\widetilde B^i$. Thus the result for $X_1(t)$ follows from e.g. \cite[Proposition 2.1.1]{Nualart2006}. The existence of the density for $X_i(t)$, $i=2,...,d$, can be obtained in the same manner by interchanging the order of indices.    
\end{proof}

As we can see, each of $X_i(t)$, $i=1,...,d$, has a density with respect to the Lebesgue measure. However, in what follows we will require a stronger result, namely the existence of the joint density of $X (t) = (X_1(t),...,X_d(t))$. 

\begin{theorem}\label{th: absolute continuity of price}
    For all $t\in(0,T]$, the law of the random vector $X(t)$ is absolutely continuous with respect to the Lebesgue measure on $\mathbb R^d$. 
\end{theorem}

\begin{proof}
    We will follow the approach developed in \cite[Theorem 5.3]{AN2015} for the standard Brownian Heston model. Namely, we will check that for any $f\in C^\infty (\mathbb R^d; \mathbb R)$ such that $f$ and all its partial derivatives are bounded the following relation holds:
    \[
        \left| \mathbb E\left[ \frac{\partial f}{\partial x_i} (X(t))\right] \right| \le c_i \sup_{x\in\mathbb R^d}|f(x)|, \quad i=1,...,d.
    \]
    Then the result will follow from \cite[Lemma 2.1.1]{Nualart2006}. 
    
    Note that, by applying sequentially the chain rule for the Malliavin derivative and the duality relationship between Malliavin derivative and the Skorokhod integral, we can obtain
    \begin{align*}
        \left| \mathbb E\left[ \frac{\partial f}{\partial x_1} (X(t))\right] \right| &= \frac{1}{t} \left|\E\left[ \int_0^t \frac{\partial f}{\partial x_1} (X(t)) \widetilde{D}_u^{1} X_1(t) \frac{1}{\widetilde{D}_u^{1} X_1(t)}du \right]\right|
        \\
        &= \frac{1}{t} \left|\E\left[ \int_0^t \widetilde{D}_u^{1}\left[f (X(t)) \right] \frac{1}{\widetilde{D}_u^{1} X_1(t)}du \right]\right|
        \\
        & = \frac{1}{t} \left|\E\left[  f (X(t)) \int_0^t \frac{1}{u_{1,1} Y_1(u)}d \widetilde W_1(u) \right]\right|
        \\
        & \le \frac{1}{t u_{1,1}} \E\left[ \left|\int_0^t \frac{1}{Y_i(u)}d \widetilde W_1(u) \right|\right] \sup_{x\in\mathbb R^d}|f(x)|
        \\
        & =: c_1 \sup_{x\in\mathbb R^d}|f(x)|,
    \end{align*}
    i.e. the required property holds for $\frac{\partial f}{\partial x_1}(X(t))$. As for $\frac{\partial f}{\partial x_i}(X(t))$ with $i=2,...,d$, one can reorder the elements of $B$ so that by swapping $B_i$ is the first one and then apply the same arguments.
\end{proof}

\section{Smoothing payoff functionals by Malliavin integration by parts}\label{sec: Malliavin duality approach}

In the remaining part of the paper, we tackle directly the computation of option prices with a focus on discontinuous payoffs. As explained earlier in the Introduction, discontinuities in the payoff function $f$ lead to deterioration of the convergence rate of $\mathbb E f(\widehat S(T))$ to $\mathbb Ef(S(T))$, where $\widehat S$ denotes some discretization of $S$, which causes an undesirable bias in Monte Carlo computations. In order to overcome this issue, we will apply the results of Section \ref{sec: Malliavin} and utilize the quadrature technique via Malliavin integration by parts proposed initially in \cite{AN2015} and adapted to fBm-driven stochastic volatility in \cite{BMdP2018, MYT2020}. The idea of the method is as follows: using the Malliavin calculus techniques, one can replace the discontinuous functional $f$ under the expectation with a locally Lipschitz one. The latter expression is far more convenient for numerical algorithms and guarantees a better convergence rate. In this Section, we explain this smoothing technique whereas the numerical algorithm and the corresponding convergence results will be provided in Section \ref{sec: numerics}.  Note that, in the papers \cite{BMdP2018, MYT2020} mentioned above, $\mathbb E f(S_1(T), ... , S_d(T))$ is computed w.r.t. the physical measure and not the martingale one, i.e. the numerical algorithms described there concern the \textit{expected payoff} rather than the \textit{option price}. However, thanks to the results of Subsection \ref{ssec: martingale measures}, we have a detailed description of equivalent local martingale measures on the market and hence are able to modify the method to account for a measure change. We explain this modification in Subsection \ref{subsec: Malliavin duality approach, measure change}.

\subsection{One-dimensional case}\label{subsec: Malliavin duality approach, 1-dim, no measure change}

We start from a simpler case when $d=1$, i.e. the SVV-model has the form
\begin{equation}\label{eq: 1-dim sandwich}
\begin{gathered}
    S(t) = S(0) + \int_0^t \mu(s) S(s)ds +  \int_0^t Y(s) S(s) dW_1(s),
    \\
    Y(t) = Y(0) + \int_0^t b(s, Y(s))ds +  \int_0^t \mathcal K(t, s)\left( {\ell_{2,1}}dW_1(s) + {\ell_{2,2}} dW_2(s)\right),
\end{gathered}
\end{equation}
where $\mathcal K = \mathcal K_1$ satisfies Assumption \ref{assum: assumptions on kernels} with $H_i = H$, $b = b_1$ satisfies Assumption \ref{assum: assumption on sandwiched drift}, $\mu = \mu_1$ is an arbitrarty $H$-H\"older continuous function and $\Lc = (\ell_{i,j})_{i,j=1}^2$ is such that $\Lc\Lc^T = \Sigma = (\sigma_{i,j})_{i,j=1}^2$. Using the transformation from Subsection \ref{ssec: Malliavin w.r.t. changed BM}, one can rewrite \eqref{eq: 1-dim sandwich} as
\begin{equation}\label{eq: 1-dim sandwich transformed}
\begin{gathered}
    S(t) = S(0) + \int_0^t \mu(s) S(s)ds +  \int_0^t Y(s) S(s) \left(u_{1,1} d \widetilde W_1(s) + u_{1,2} d\widetilde W_2(s) \right),
    \\
    Y(t) = Y(0) + \int_0^t b(s, Y(s))ds +   \int_0^t \mathcal K(t, s) d\widetilde W_2(s),
\end{gathered}
\end{equation}
where $\widetilde W = (\widetilde W_1, \widetilde W_2)$ is a 2-dimensional Brownian motion such that
\[
    \left(\begin{matrix}
        u_{1,1} & u_{1,2}
        \\
        0 & 1
    \end{matrix}\right) \left(\begin{matrix}
        \widetilde W_1
        \\
        \widetilde W_2
    \end{matrix}\right) = \left(\begin{matrix}
        B^S_1
        \\
        B^Y_1
    \end{matrix}\right) = \left(\begin{matrix}
        1 & 0
        \\
        \ell_{2,1} & \ell_{2,2}
    \end{matrix}\right) \left(\begin{matrix}
        W_1
        \\
        W_2
    \end{matrix}\right).
\]

Let now $f$ be a measurable payoff function that satisfies the following assumption.
\begin{assumption}\label{assum: 1-dim payoff} Function $f$: $\mathbb R \to \mathbb R$ is locally Riemann integrable function (not necessarily continuous)  of polynomial growth, i.e. there exist constants $q>0$ and $c_f > 0$ such that
\[
    |f(s)| \le c_f (1+|s|^q),\quad s > 0.
\]
\end{assumption}

Denote now $F(s) := \int_0^s f(u) du$, $g(x) := f(e^x)$, $G(x):= F(1) + \int_0^x g(v) dv$, i.e. $f(S(T)) = g(X(T))$, where $X(t) := \log S(t)$, $t\in[0,T]$. By Theorem \ref{th: new Malliavin differentiability}, $X(t) \in \widetilde {\mathbb D}^{1,2}$ with
\[
    \widetilde D^1_v X(t) = u_{1,1} Y(v)\mathbbm 1_{[0,t]}(v)
\]
and, by Theorem \ref{th: absolute continuity of price}, $X(t)$ has density. Thus one can use Proposition 1.2.4 from \cite{Nualart2006} and remark after it to claim that
\[
    \widetilde D^1_v G(X(t)) = g(X(t)) \widetilde D^1_v X(t), \quad t\in (0,T],
\]
if $G$ turns out to be globally Lipschitz continuous. However, we will require the Malliavin chain rule for $G(X(T))$ in the setting provided by Assumption \ref{assum: 1-dim payoff}, i.e. when $G$ is only \textit{locally} Lipschitz continuous. To do this, we will extend \cite[Proposition 1.2.4]{Nualart2006} in a manner similar to \cite[Remark 9]{BMdP2018}. 
\begin{proposition}\label{prop: chain rule for locally Lipschitz}
    For $G$ defined above and any $t\in(0,T]$, $G(X(t)) \in \widetilde{\mathbb D}^{1,2}$ and
    \begin{equation}
        \widetilde D G(X(t)) = g(X(t)) \widetilde D X(t), \quad t\in (0,T].
    \end{equation}
\end{proposition}
\begin{proof}
    For any $n\in\mathbb N$, denote 
    \[
        g_n(x) := g(x) \mathbbm 1_{[-n,n]}(x)
    \]
    and
    \begin{align*}
        G_n(x) &:= F(1) + \int_0^x g_n(v) dv = G(x)\mathbbm 1_{[-n,n]}(x) + G(n)\mathbbm 1_{[n,\infty)}(x) + G(-n)\mathbbm 1_{(-\infty, -n]}(x).
    \end{align*}
    It is clear that $G_n$ is globally Lipschitz and thus $G_n(X(t)) \in \widetilde{\mathbb D}^{1,2}$ and
    \[
        \widetilde D G_n(X(t)) = g_n(X(t)) \widetilde D X(t).
    \]
    Thus, taking into account the closedness of the operator $\widetilde D$, it is sufficient to check that 
    \begin{equation}\label{proofeq: Malliavin chain rule, convergence of r.v.}
        g_n(X(t)) \xrightarrow{L^2(\Omega)} g(X(t)), \quad n\to \infty,
    \end{equation}
    and
    \begin{equation}\label{proofeq: Malliavin chain rule, convergence of M.d.}
        g_n(X(t)) \widetilde D X(t) \xrightarrow{L^2\left(\Omega; L^2([0,T];\mathbb R^2)\right)} g (X(t)) \widetilde D X(t), \quad n\to \infty.
    \end{equation}
    Observe that $g_n(X(t)) \to g(X(t))$ a.s., $n\to\infty$. By Assumption \ref{assum: 1-dim payoff}, $f$ is of polynomial growth and thus, by Theorem \ref{th: properties of S}, $\mathbb E f^2(S(t)) = \mathbb Eg^2(X(t)) < \infty$ and \eqref{proofeq: Malliavin chain rule, convergence of r.v.} follows from the dominated convergence theorem.
    
    Finally, using polynomial growth of $f$, explicit shape of $\widetilde D X(t)$ given by Theorem \ref{th: new Malliavin differentiability} as well as estimates similar to the ones from Step 1 of the proof of Theorem \ref{th: Malliavin differentiability}, one can deduce that
    \[
        \mathbb E\left[g^2(X(t)) \int_0^T (\widetilde D^i_u X(t))^2 du \right] < \infty,
    \]
    and therefore \eqref{proofeq: Malliavin chain rule, convergence of M.d.} also holds due to dominated convergence.
\end{proof}

Now, having Proposition \ref{prop: chain rule for locally Lipschitz}, we are ready to proceed to the main result of this Subsection in the spirit of \cite[Theorem 4.2]{AN2015} or \cite[Lemma 11]{BMdP2018}.

\begin{theorem}\label{th: representation}
    Let $S(T)$ be given by \eqref{eq: 1-dim sandwich transformed}, $f$: $\mathbb R \to \mathbb R$ satisfy Assumption \ref{assum: 1-dim payoff}, $X(t) := \log S(t)$, $F(s) := \int_0^s f(u)du$, $g(x) := f(e^x)$, $G(x) := F(1) + \int_0^x g(v)dv$. Then
    \begin{equation}\label{eq: representation 1, 1-dim}
        \mathbb E f(S(T)) = \frac{1}{Tu_{1,1}} \mathbb E \left[ G(X(T)) \int_0^T \frac{1}{Y(u)} d\widetilde W_1(u) \right].
    \end{equation}
    Alternatively,
    \begin{equation}\label{eq: representation 2, 1-dim}
        \mathbb E f(S(T)) =  \mathbb E \left[ \frac{F(S(T))}{S(T)}  \left(1+\frac{1}{Tu_{1,1}}\int_0^T \frac{1}{Y(u)} d\widetilde W_1(u)\right)\right].
    \end{equation}
\end{theorem}
\begin{proof}
    In order to obtain representation \eqref{eq: representation 1, 1-dim}, observe that
    \begin{align*}
        \mathbb E f(S(T)) &= \mathbb E g(X(T)) = \frac{1}{T} \mathbb E\left[ \int_0^T g(X(T)) \widetilde D^1_u X(T) \frac{1}{\widetilde D^1_u X(T)} du\right] 
        \\
        &= \frac{1}{T u_{1,1}} \mathbb E\left[ \int_0^T \widetilde D^1_u G(X(T)) \frac{1}{Y(u)} du\right] 
        \\
        &= \frac{1}{T u_{1,1}} \mathbb E \left[ G(X(T)) \int_0^T \frac{1}{Y(u)} d\widetilde W_1(u) \right].
    \end{align*}
    Here we used the chain rule from Proposition \ref{prop: chain rule for locally Lipschitz} and duality relation between the Malliavin derivative and Skorokhod integral (which coincides in our case with the It\^o integral). 
    
    As for the representation \eqref{eq: representation 2, 1-dim}, note that
    \[
        G(x) = F(1) + \int_0^x f(e^v) e^v \frac{1}{e^v}dv = \frac{F(e^x)}{e^x} + \int_0^x \frac{F(e^v)}{e^v} dv.
    \]
    Moreover, using the same argument as in Proposition \ref{prop: chain rule for locally Lipschitz}, it is easy to prove that
    \[
        \widetilde D^1 \int_0^{X(T)} \frac{F(e^v)}{e^v}dv = \frac{F(S(T))}{S(T)} \widetilde D^1 X(T), 
    \]
    therefore, again by the Malliavin duality, we obtain that
    \begin{align*}
        \mathbb E \left[ \int_0^{X(T)} \frac{F(e^v)}{e^v}dv \frac{1}{T u_{1,1}} \int_0^T \frac{1}{Y(u)} d\widetilde W_1(u) \right] &= \mathbb E \left[ \widetilde D^1_u \int_0^{X(T)} \frac{F(e^v)}{e^v}dv \frac{1}{T u_{1,1}} \frac{1}{Y(u)} du \right]
        \\
        & = \mathbb E \left[ \frac{1}{T}\int_0^T \frac{F(S(T))}{S(T)} \widetilde D^1_u X(T) \frac{1 }{\widetilde D^1_u X(T)}  du\right]
        \\
        & = \mathbb E \left[ \frac{F(S(T))}{S(T)} \right],
    \end{align*}
    which implies \eqref{eq: representation 2, 1-dim}.
\end{proof}

\subsection{Multidimensional case}

The approach described in Subsection \ref{subsec: Malliavin duality approach, 1-dim, no measure change} can also be utilized for basked options in $d$-dimensional SVV models. Namely, consider a $d$-dimensional volatility process $Y = (Y_1,...,Y_d)$, where each $Y_i$ is defined by \eqref{eq: volatility process equations}, a $d$-dimensional price process $S = (S_1,...,S_d)$ given by \eqref{eq: definition of price process} and the corresponding log-price process $X=(X_1,...,X_d)$, $X_i = \log S_i$.

Consider now a basket option with a payoff $f\left(\sum_{i=1}^d \alpha_i S_i(T)\right)$ with $f$: $\mathbb R \to \mathbb R$ satisfying Assumption \ref{assum: 1-dim payoff} and $\alpha_i > 0$, $i=1,...,d$. As before, define $F(s) := \int_0^{s} f(u) du$ and observe that the function 
\[
    G_1(x_1,...,x_d) := \frac{F\left( \sum_{i=1}^d \alpha_i e^{x_i} \right)}{\alpha_1 e^{x_1}}, \quad (x_1,...,x_d)  \in \mathbb R^d,
\]
is locally Lipschitz continuous: if $\sum_{i=1}^d \alpha_i e^{x_i} > \sum_{i=1}^d \alpha_i e^{x_i'}$, then
\begin{align*}
    |G_1(x_1,...,x_d) - G_1(x_1',...,x_d')| & \le \frac{e^{-x_1}}{\alpha_1} \int_{\sum_{i=1}^d \alpha_i e^{x_i'}}^{\sum_{i=1}^d \alpha_i e^{x_i}} |f(s)|ds
    \\
    &\quad + \frac{1}{\alpha_1}\int_0^{\sum_{i=1}^d \alpha_i e^{x_i'}} |f(s)|ds |e^{-x_1} - e^{-x_1'}|
    \\
    &\le c_f \frac{e^{-x_1}}{\alpha_1}\left( 1+ \left| \sum_{i=1}^d \alpha_i e^{x_i} \right|^q \right)\sum_{j=1}^d \alpha_j e^{x_j \vee x_j'}|{x_j} - {x_j'}|
    \\
    &\quad + c_f \frac{1}{\alpha_1} \left(\sum_{i=1}^d \alpha_i e^{x_i'} + \frac{1}{q+1} \left|\sum_{i=1}^d \alpha_i e^{x_i'}\right|^{q+1} \right) e^{-(x_1 \wedge x_1')}|x_1 - x_1'|.
\end{align*}
Note that polynomial growth of $f$ (and hence $F$) together with Theorem \ref{th: properties of S} imply that random variables $F\left(\sum_{i=1}^d \alpha_i S_i(T)\right)$ and $G_1(X_1(T),...,X_d(T))$ have moments of all orders. Moreover, Theorem \ref{th: absolute continuity of price} implies that the law of $X$ (and hence of $S$) is absolutely continuous w.r.t. the Lebesgue measure on $\mathbb R^d$. Therefore, by approximating $G_1(X(t))$ with Lipschitz functions just like in Proposition \ref{prop: chain rule for locally Lipschitz}, it is straightforward to get the following chain rule for $G_1(X(t))$.
\begin{proposition}\label{prop: chain rule for locally Lipschitz, d-dim}
     For $G_1$ defined above and any $t\in(0,T]$, $G_1(X(t)) \in \widetilde{\mathbb D}^{1,2}$ and
    \begin{equation}
        \widetilde D G_1(X(t)) = \sum_{i=1}^d \frac{\partial G_1(X(t))}{\partial x_i} \widetilde D X_i(t), \quad t\in (0,T].
    \end{equation}
\end{proposition}
We are now ready to formulate an analogue of Theorem \ref{th: representation} for the basket option.
\begin{theorem}\label{th: representation d-dim}
    Let $S = (S_1,...,S_d)$ be given by \eqref{eq: definition of price process}, $X=(X_1,...,X_d)$ be the corresponding log-price process, $X_i = \log S_i$, $f$: $\mathbb R \to \mathbb R$ satisfy Assumption \ref{assum: 1-dim payoff}, $F(s) := \int_0^s f(u)du$ and $\alpha_i > 0$, $i=1,...,d$. Then
    \begin{equation}\label{eq: representation 2, d-dim}
        \mathbb E f\left(\sum_{i=1}^d \alpha_i S_i(T)\right) = \mathbb E \left[ \frac{F\left(\sum_{i=1}^d \alpha_i S_i(T)\right)}{\alpha_1 S_1(T)} \left( 1 + \frac{1}{T u_{1,1}} \int_0^T \frac{1}{Y_1(u)}d\widetilde W_1(u) \right) \right].
    \end{equation}
\end{theorem}
\begin{proof}
    Note that 
    \[
        \frac{\partial G_1(x_1,...,x_d)}{\partial x_1} = f\left( \sum_{i=1}^d \alpha_i e^{x_i} \right)  - \frac{1}{\alpha_1 e^{x_1}} F \left(\sum_{i=1}^d \alpha_i e^{x_i}\right)
    \]
    hence, taking into account Proposition \ref{prop: chain rule for locally Lipschitz, d-dim} and Theorem \ref{th: new Malliavin differentiability},
    \begin{align*}
        \mathbb E \left[ f\left(\sum_{i=1}^d \alpha_i S_i(T)\right) \right] & = \mathbb E \left[ \frac{1}{T} \int_0^T \left(f\left(\sum_{i=1}^d \alpha_i e^{X_i(T)}\right) - \frac{1}{\alpha_1 e^{X_1(T)}} F\left(\sum_{i=1}^d \alpha_i e^{X_i(T)}\right)\right) dv\right]
        \\
        &\quad + \mathbb E \left[ \frac{1}{\alpha_1 e^{X_1(T)}} F\left(\sum_{i=1}^d \alpha_i e^{X_i(T)}\right)\right]
        \\
        & = \mathbb E \left[ \frac{1}{T} \int_0^T \frac{\partial G_1(X(T))}{\partial x_1} dv\right] + \mathbb E \left[ G_1(X(T)) \right]
        \\
        &= \mathbb E \left[ \frac{1}{T} \int_0^T \frac{\partial G_1(X(T))}{\partial x_1} \widetilde D^1_vX_1(T) \frac{1}{\widetilde D^1_vX_1(T)} dv\right] + \mathbb E \left[ G_1(X(T)) \right]
        \\
        & = \mathbb E \left[ \frac{1}{T} \int_0^T \widetilde D_v^1 G_1(X(T)) \frac{1}{u_{1,1} Y_1(v)} dv\right] + \mathbb E \left[ G_1(X(T)) \right]
        \\
        & = \mathbb E \left[ \frac{1}{Tu_{1,1}} G_1(X(T))\int_0^T \frac{1}{ Y_1(v)} d\widetilde W_1(v)\right] + \mathbb E \left[ G_1(X(T)) \right]
        \\
        & = \mathbb E \left[ \frac{F\left(\sum_{i=1}^d \alpha_i S_i(T)\right)}{\alpha_1 S_1(T)} \left(1+ \frac{1}{Tu_{1,1}} \int_0^T \frac{1}{ Y_1(v)} d\widetilde W_1(v) \right)  \right],
    \end{align*}
    as required.
\end{proof}

\subsection{Quadrature via Malliavin integration under a change of measure}\label{subsec: Malliavin duality approach, measure change}

Theorems \ref{th: representation} and \ref{th: representation d-dim} concern the expectation w.r.t. the physical measure, i.e. the obtained representation of the expected payoff is not the option price. In this Section, we will modify the quadrature technique to account for a measure change.

Consider a martingale measure with the density
\begin{equation}\label{eq: minimal martingale measure}
    \mathcal E_T := \mathcal E_T \left\{ -\sum_{j=1}^d \int_0^\cdot \left(\sum_{k=1}^d \ell^{(-1)}_{j,k} \frac{\widetilde \mu_k(s)}{Y_k(s)} \right)dW_j(s) \right\},
\end{equation}
where $\ell^{(-1)}_{j,k}$ and $\widetilde \mu_k(s)$ are given in Subsection \ref{ssec: martingale measures}. 
\begin{remark}
    In the literature, the measure with the density \eqref{eq: minimal martingale measure} is commonly referred to as the \textit{minimal martingale measure}, see e.g. \cite{Schweizer_1995}.
\end{remark}

Recall that
\[
    W = \Lc^{-1} \mathcal U \widetilde W,
\]
so denote $\rho_{j,i}$ the values such that
\[
    W_j(t) = \sum_{i=1}^{2d} \rho_{j,i} \widetilde W_i(t),
\]
i.e.
\begin{align*}
    \mathcal E_T &= \mathcal E_T \left\{ -\sum_{j=1}^d \sum_{i=1}^{2d} \int_0^\cdot \rho_{j,i} \left(\sum_{k=1}^d \ell^{(-1)}_{j,k} \frac{\widetilde \mu_k(s)}{Y_k(s)} \right)d \widetilde W_i(s) \right\}
    \\
    & = \mathcal E_T \left\{ -\sum_{i=1}^{2d}  \int_0^\cdot \left(\sum_{j=1}^d\rho_{j,i} \left(\sum_{k=1}^d \ell^{(-1)}_{j,k} \frac{\widetilde \mu_k(s)}{Y_k(s)}\right) \right)d \widetilde W_i(s) \right\}
    \\
    & = \exp\left\{ -\sum_{i=1}^{2d}  \int_0^T \left(\sum_{j=1}^d\rho_{j,i} \left(\sum_{k=1}^d \ell^{(-1)}_{j,k} \frac{\widetilde \mu_k(s)}{Y_k(s)}\right) \right)d \widetilde W_i(s)\right\} \times
    \\
    &\qquad \times \exp\left\{- \frac{1}{2} \sum_{i=1}^{2d}  \int_0^T \left(\sum_{j=1}^d\rho_{j,i} \left(\sum_{k=1}^d \ell^{(-1)}_{j,k} \frac{\widetilde \mu_k(s)}{Y_k(s)}\right) \right)^2 d s \right\}.
\end{align*}

Note that each $Y_i$ is bounded away from zero by $\varphi_i^* := \min_{t\in[0,T]} \varphi_i(t)$ and hence $\frac{1}{Y_i(s)}$ can be represented as follows:
\[
    \frac{1}{Y_i(s)} = h_i(Y_i(s)), \quad h_i(y) := \begin{cases}
        \frac{1}{y}, &\quad y > \varphi_i^*,
        \\
        -\frac{1}{(\varphi_i^*)^2} y + \frac{2}{\varphi_i^*}, &\quad y\le \varphi_i^*.
    \end{cases}
\]
The function $h_i$ is bounded, continuously differentiable and has bounded derivative, hence by the classical chain rule $\frac{1}{Y_i(s)} \in \widetilde{\mathbb D}^{1,2}$ and
\[
    \widetilde D^1 \left[\frac{1}{Y_i(s)}\right] = 0. 
\]
In particular, this implies that
\[
     \widetilde D^1_v \left[ - \frac{1}{2} \sum_{i=1}^{2d}  \int_0^T \left(\sum_{j=1}^d\rho_{j,i} \left(\sum_{k=1}^d \ell^{(-1)}_{j,k} \frac{\widetilde \mu_k(s)}{Y_k(s)}\right) \right)^2 d s \right] = 0
\]
and
\begin{align*}
    \widetilde D_v^1 &\left[ -\sum_{i=1}^{2d}  \int_0^T \left(\sum_{j=1}^d\rho_{j,i} \left(\sum_{k=1}^d \ell^{(-1)}_{j,k} \frac{\widetilde \mu_k(s)}{Y_k(s)}\right) \right)d \widetilde W_i(s) \right]
    \\
    & = - \sum_{j=1}^d  \sum_{k=1}^d \rho_{j,1} \ell^{(-1)}_{j,k} \frac{\widetilde \mu_k(v)}{Y_k(v)} 
    \\
    & = \sum_{k=1}^d \left(- \sum_{j=1}^d  \rho_{j,1} \ell^{(-1)}_{j,k}\right) \frac{\widetilde \mu_k(v)}{Y_k(v)}
    \\
    &=: \sum_{k=1}^d \beta_k \frac{\widetilde \mu_k(v)}{Y_k(v)},
\end{align*}
where $\beta_k := \left(- \sum_{j=1}^d  \rho_{j,1} \ell^{(-1)}_{j,k}\right)$. Furthermore, using the boundedness of $Y$ and the chain rule from Proposition \ref{prop: chain rule}, it is easy to establisth that $\mathcal E_T \in \widetilde {\mathbb D}^{1,2}$ and
\[
    \widetilde D^1_v \mathcal E_T = \mathcal E_T \sum_{k=1}^d \beta_k \frac{\widetilde \mu_k(v)}{Y_k(v)}.
\]
Let now $f$: $\mathbb R \to \mathbb R$ satisfy Assumption \ref{assum: 1-dim payoff}, $\alpha_i >0$, $i=1,...,d$. Then the following modification of Theorem \ref{th: representation d-dim} is true.
\begin{theorem}\label{th: representation d-dim measure change}
    Let $S = (S_1,...,S_d)$ be given by \eqref{eq: definition of price process}, $X=(X_1,...,X_d)$ be the corresponding log-price process, $X_i = \log S_i$, $f$: $\mathbb R \to \mathbb R$ satisfy Assumption \ref{assum: 1-dim payoff}, $F(s) := \int_0^s f(u)du$ and $\alpha_i > 0$, $i=1,...,d$. Then
    \begin{equation}\label{eq: representation 2, d-dim measure change}
    \begin{aligned}
        \mathbb E &\left[ \mathcal E_T f\left(\sum_{i=1}^d \alpha_i S_i(T)\right) \right]
        \\
        &= \mathbb E \left[ \mathcal E_T  \frac{F\left(\sum_{i=1}^d \alpha_i S_i(T)\right)}{\alpha_1 S_1(T)} \left(1 + \frac{1}{Tu_{1,1}}\left( \int_0^T\frac{1}{Y_1(v)} d\widetilde W_1(v) - \sum_{k=1}^d\int_0^T  \beta_k \frac{\widetilde \mu_k(v)}{Y_k(v)Y_1(v)}dv\right)\right)\right].
    \end{aligned}
    \end{equation}
\end{theorem}
\begin{proof}
    Denote 
    \[
        G_1(x_1,...,x_d) := \frac{F\left( \sum_{i=1}^d \alpha_i e^{x_i} \right)}{\alpha_1 e^{x_1}}, \quad (x_1,...,x_d)  \in \mathbb R^d.
    \]
    Using polynomial growth of $f$, boundedness of $Y_i$, existence of all moments of $S_i$ and Proposition \ref{prop: chain rule}, it is easy to prove that
    \[
        \mathcal E_T G_1(X(T)) \in \widetilde{\mathbb D}^{1,2}
    \]
    and
    \begin{align*}
        \widetilde D_v^1 \left[\mathcal E_T G_1(X(T))\right] & =  \mathcal E_T \widetilde D_v^1 G_1(X(T)) + G_1(X(T)) \widetilde D_v^1 \mathcal E_T 
        \\
        &= \mathcal E_T \widetilde D_v^1 G_1(X(T)) + \mathcal E_T G_1(X(T)) \sum_{k=1}^d \beta_k \frac{\widetilde \mu_k(v)}{Y_k(v)}.
    \end{align*}
    Next, observe that
    \begin{align*}
        \mathbb E &\left[ \mathcal E_T f\left(\sum_{i=1}^d \alpha_i S_i(T)\right) \right] = \mathbb E \left[ \mathcal E_T \left( f\left(\sum_{i=1}^d \alpha_i e^{X_i(T)}\right) - G_1(X(T)) \right) \right] + \mathbb E \left[ \mathcal E_T  G_1(X(T)) \right]
    \end{align*}
    and
    \begin{align*}
        \mathbb E &\left[ \mathcal E_T \left( f\left(\sum_{i=1}^d \alpha_i e^{X_i(T)}\right) - G_1(X(T)) \right) \right] = \mathbb E \left[ \mathcal E_T \frac{\partial G_1(X(T))}{\partial x_1} \right]
        \\
        & = \mathbb E \left[ \frac{1}{T} \int_0^T \mathcal E_T \frac{\partial G_1(X(T))}{\partial x_1} \widetilde D^1_v X_1(T) \frac{1}{\widetilde D^1_v X_1(T)} dv\right]
        \\
        &= \mathbb E \left[ \frac{1}{Tu_{1,1}} \int_0^T \mathcal E_T \widetilde D^1_v G_1(X(T)) \frac{1}{Y_1(v)} dv\right]
        \\
        & = \frac{1}{Tu_{1,1}}\left(\mathbb E \left[  \int_0^T \widetilde D^1_v\left[\mathcal E_T  G_1(X(T))\right] \frac{1}{Y_1(v)} dv\right]  - \mathbb E \left[ \int_0^T G_1(X(T))\widetilde D^1_v\mathcal E_T  \frac{1}{Y_1(v)} dv\right]\right)
        \\
        & = \frac{1}{Tu_{1,1}}\left(\mathbb E \left[ \mathcal E_T  G_1(X(T)) \int_0^T\frac{1}{Y_1(v)} d\widetilde W_1(v)\right]  - \mathbb E \left[  G_1(X(T))\mathcal E_T \int_0^T \sum_{k=1}^d \beta_k \frac{\widetilde \mu_k(v)}{Y_k(v)Y_1(v)}dv\right]\right)
        \\
        & = \mathbb E \left[ \mathcal E_T  \frac{F\left(\sum_{i=1}^d \alpha_i S_i(T)\right)}{\alpha_1 S_1(T)} \frac{1}{Tu_{1,1}}\left( \int_0^T\frac{1}{Y_1(v)} d\widetilde W_1(v) - \sum_{k=1}^d\int_0^T  \beta_k \frac{\widetilde \mu_k(v)}{Y_k(v)Y_1(v)}dv\right)\right].
    \end{align*}
    Taking into account all of the above, we can finally write:
    \begin{align*}
        \mathbb E &\left[ \mathcal E_T f\left(\sum_{i=1}^d \alpha_i S_i(T)\right) \right]
        \\
        &= \mathbb E \left[ \mathcal E_T  \frac{F\left(\sum_{i=1}^d \alpha_i S_i(T)\right)}{\alpha_1 S_1(T)} \left(1 + \frac{1}{Tu_{1,1}}\left( \int_0^T\frac{1}{Y_1(v)} d\widetilde W_1(v) - \sum_{k=1}^d\int_0^T  \beta_k \frac{\widetilde \mu_k(v)}{Y_k(v)Y_1(v)}dv\right)\right)\right]
    \end{align*}
    which ends the proof.
\end{proof}

\begin{remark}
    Note that the martingale measure should not necessarily be the minimal one: it is possible to obtain a representation of the type \eqref{eq: representation 2, d-dim measure change} for other martingale measures provided that the corresponding densities are Malliavin differentiable and one can use the chain rules as in the proof of Theorem \ref{th: representation d-dim measure change}.
\end{remark}

\section{Price approximation for options with discontinuous payoffs}\label{sec: numerics}

In this Section, we provide a numerical algorithm for computation of $\mathbb E\left[f\left(\sum_{i=1}^d \alpha_i S_i(T)\right)\right]$, both with respect to physical and minimal martingale measures. We will exploit the representations obtained in Theorems \ref{th: representation}, \ref{th: representation d-dim} and \ref{th: representation d-dim measure change}: absence of discontinuities under the expectation sign allows to avoid deterioration of the convergence speed. 

The Section is structured as follows. In Subsection \ref{subsec: numerical scheme}, we describe the numerical scheme used to simulate $Y$, $X$ and $S$ (the corresponding numerical approximations will be denoted by  $\widehat Y$, $\widehat X$ and $\widehat S$). In Subsection \ref{subsec: double approximation}, we consider a straightforward way of computing $\mathbb E\left[f\left(\sum_{i=1}^d \alpha_i S_i(T)\right)\right]$ by plugging in $\widehat Y$ and $\widehat S$ into \eqref{eq: representation 1, 1-dim}--\eqref{eq: representation 2, 1-dim}, \eqref{eq: representation 2, d-dim} or \eqref{eq: representation 2, d-dim measure change}. An estimate of the approximation error is provided. In Subsection \ref{subsec: Gaussian approximation}, we discuss another algorithm for computation of the expectation for a one-dimensional case: the vector $\left(X(T), \int_0^T \frac{1}{Y(v)}dv \right)^T$ turns out to be conditionally Gaussian w.r.t. an appropriate $\sigma$-field and one can use this fact to approximate \eqref{eq: representation 1, 1-dim} without simulating $\widehat X(T)$ at all.  

\subsection{Numerical schemes}\label{subsec: numerical scheme}

\paragraph{Numerical scheme for the volatility process.} For the stochastic volatility processes $Y_i$, $i=1,...,d$, we will use the \textit{drift-implicit} Euler scheme described in \cite{DNMYT2022}. For reader's convenience, we briefly describe it below as well as give the necessary convergence results.

Let $\{0=t_0 < t_1<...<t_N=T\}$ be a uniform partition of $[0,T]$, $t_k := \frac{Tk}{N}$, $k=0,1,..., N$, with the mesh $\Delta_N:=\frac{T}{N}$ such that
\begin{equation}\label{eq: condition on the mesh}
    \Delta_N \max_{i=1,...,d} \sup_{(t,y) \in \mathcal D^i_{0,0}} \frac{\partial b_i}{\partial y}(t,y)  < 1.
\end{equation}
\begin{remark}
    Note that such $\Delta_N$ exists due to Assumption \ref{assum: assumption on sandwiched drift}(iv).
\end{remark}
Define $\widehat Y_i(t)$ as follows:
\begin{equation}\label{eq: definition of backward Euler scheme}
\begin{aligned}
    \widehat Y_i(0) &= Y_i(0),
    \\
    \widehat Y_i(t_{k+1}) &= \widehat Y_i(t_{k}) + b_i(t_{k+1}, \widehat Y_i(t_{k+1}))\Delta_N + (Z_i(t_{k+1}) - Z_i(t_k)),
    \\
    \widehat Y_i(t) &= \widehat Y_i(t_k), \quad t\in[t_k, t_{k+1}),
\end{aligned}
\end{equation}
where $Z_i(t) := \int_0^t \mathcal K_i(t,s)dB^Y_i(s)$. Note that the second expression in \eqref{eq: definition of backward Euler scheme} is considered as an equation with respect to $\widehat Y_i(t_{k+1})$ and, under Assumption \ref{assum: assumption on sandwiched drift}, it has a unique solution such that $\widehat Y_i(t_{k+1}) \in (\varphi_i(t_{k+1}), \psi_i(t_{k+1}))$, see \cite[Section 2]{DNMYT2022} for more details. The convergence results from \cite{DNMYT2022} for this scheme are summarized in the following theorem.

\begin{theorem}\label{th: drift-implicit Euler scheme}
    Let Assumptions \ref{assum: assumptions on kernels} and \ref{assum: assumption on sandwiched drift} hold and the mesh of the partition $\Delta_N$ satisfy \eqref{eq: condition on the mesh}. Fix $i=1,...,d$, take an arbitrary $\lambda_i \in (0,H_i)$, where $H_i$ is from Assumption \ref{assum: assumptions on kernels}, and fix $r\ge 1$. Then there exists a constant $C>0$ that does not depend on the partition such that
    \[
        \mathbb E \left[ \sup_{t\in[0,T]} |Y_i(t) - \widehat Y_i(t)|^r \right] \le C \Delta_N^{\lambda_i r}    
    \]
    and
    \[
        \mathbb E \left[\sup_{t\in[0,T]}\left| \frac{1}{Y_i(t)} - \frac{1}{\widehat Y_i(t)} \right|^r\right] \le C \Delta_N^{\lambda_i r}.
    \]
\end{theorem}

\paragraph{Numerical scheme for the price process.} Recall that the log-price $X_i$ has the form
\[
    X_i(t) = X_i(0) + \int_0^t \left( \mu_i(s) - \frac{ Y_i^2(s)}{2} \right)ds +  \int_0^t Y_i(s) dB^S_i(s),
\]
where $\mu_i$ is $H_i$-H\"older continuous with $H_i$ being from Assumption \ref{assum: assumptions on kernels}(ii). For a given uniform partition $\{0=t_0 < t_1 < ... < t_N = T\}$ satisfying \eqref{eq: condition on the mesh}, consider the random process $\widehat X_i = \{\widehat X_i(t),~t\in[0,T]\}$ such that for any point $t_n$ of the partition
\[
    \widehat X_i(t_n) = X_i(0) + \frac{T}{N}\sum_{k=0}^{n-1} \left( \mu_i(t_k) - \frac{ \widehat Y_i^2(t_k)}{2} \right) +  \sum_{k=0}^{n-1} \widehat Y_i(t_k) \left(B^S_i(t_{k+1}) - B^S_i(t_k)\right),
\]
where $\widehat Y_i$ is defined by \eqref{eq: definition of backward Euler scheme}, and for any $t\in (t_n, t_{n+1})$
\[
    \widehat X_i(t) = \widehat X_i(t_n).
\]

\begin{theorem}\label{th: approximation of log-price}
    Let Assumptions \ref{assum: assumptions on kernels} and \ref{assum: assumption on sandwiched drift} hold and the mesh of the partition $\Delta_N$ satisfy \eqref{eq: condition on the mesh}. Fix $i=1,...,d$, take an arbitrary $\lambda_i \in (0,H_i)$, where $H_i$ is from Assumption \ref{assum: assumptions on kernels}, and fix $r\ge 1$. Then there exists a constant $C>0$ that does not depend on the partition such that
    \[
        \mathbb E \left[ \sup_{n = 0,1,...,N}|X_i(t_n) - \widehat X_i(t_n)|^r \right] \le C\Delta^{r\lambda_i}_N.
    \]
    In particular,
    \[
        \mathbb E \left[ |X_i(T) - \widehat X_i(T)|^r \right] \le C\Delta^{r\lambda_i}_N.
    \]
\end{theorem}
\begin{proof}
    Observe that for any $n=1,...,N$
    \begin{align*}
        |X(t_{n}) - \widehat X(t_n)|^r & \le C \left( \sum_{k=0}^{n-1} \int_{t_k}^{t_{k+1}} |\mu_i(s) - \mu_i(t_k)| ds  \right)^r + C\left( \int_{0}^{t_n} |Y_i^2(s) - \widehat Y_i^2(s)| ds \right)^r
        \\
        &\quad + C \left| \int_0^{t_n} (Y_i(s) - \widehat Y_i(s))dB^S_i(s) \right|^r
        \\
        & \le C \left( \sum_{k=0}^{N-1} \int_{t_k}^{t_{k+1}} |s - t_k|^{\lambda_i} ds  \right)^r + C\left( \int_{0}^{T} |Y_i(s) + \widehat Y_i(s)| |Y_i(s) - \widehat Y_i(s)| ds \right)^r
        \\
        &\quad + C \left( \sup_{t\in[0,T]}\left|\int_0^{t} (Y_i(s) - \widehat Y_i(s))dB^S_i(s)\right| \right)^r
        \\
        & \le C \Delta^{r\lambda_i}_N + C \sup_{t\in[0,T]}|Y_i(t) - \widehat Y_i(t)|^r +  C \left( \sup_{t\in[0,T]}\left|\int_0^{t} (Y_i(s) - \widehat Y_i(s))dB^S_i(s)\right| \right)^r.
    \end{align*}
    Hence, by the Burkholder-Davis-Gundy inequality and Theorem \ref{th: drift-implicit Euler scheme},
    \begin{align*}
        \mathbb E\left[\sup_n|X(t_{n}) - \widehat X(t_n)|^r\right] &  \le C \Delta^{r\lambda_i}_N + C \mathbb E\left[\sup_{t\in[0,T]}|Y_i(t) - \widehat Y_i(t)|^r\right] 
        \\
        &\quad +  C \mathbb E\left[\left( \sup_{t\in[0,T]}\left|\int_0^{t} (Y_i(s) - \widehat Y_i(s))dB^S_i(s)\right| \right)^r\right]
        \\
        &\le C \Delta^{r\lambda_i}_N + C\mathbb E\left[ \left(\int_0^{T} (Y_i(s) - \widehat Y_i(s))^2 ds\right)^{\frac{r}{2}} \right]
        \\
        &\le C \Delta^{r\lambda_i}_N.
    \end{align*}
\end{proof}

\noindent Next, denote $\widehat S_i(t) := e^{\widehat X_i(t)}$.

\begin{proposition}\label{prop: moments of approximated price}
    Let Assumptions \ref{assum: assumptions on kernels} and \ref{assum: assumption on sandwiched drift} hold and the mesh of the partition $\Delta_N$ satisfy \eqref{eq: condition on the mesh}. Then, for any $r\in \mathbb R$
    \[
        \sup_{N\ge 1} \sup_{n=0,1,...,N}\mathbb E \left[ \widehat S_i^r(t_n) \right] < \infty.
    \]
\end{proposition}
\begin{proof}
    Since both $\mu_i$ and $\widehat Y_i$ are bounded and the bound does not depend on the partition, it is sufficient to prove that
    \[
        \sup_{N\ge 1}\sup_{n=0,1,...,N}\mathbb E\left[\exp\left\{ r \int_0^{t_n} \widehat Y_i(s) dB_i^S(s) \right\}\right] < \infty.
    \]
    Consider a stochastic process
    \[
        \exp\left\{ r \int_0^{t} \widehat Y_i(s) dB_i^S(s)\right\}, \quad t\in[0,T].
    \]
    Then, since $\widehat Y_i$ is bounded, Novikov's condition implies that
    \[
        \mathbb E \left[\exp\left\{  \int_0^{t} r \widehat Y_i(s) dB_i^S(s) - \frac{r^2}{2} \int_0^{t} \widehat Y_i^2(s) ds\right\}\right] = 1
    \]
    and hence there exist a constant $C$ that does not depend on the partition such that
    \begin{align*}
        \mathbb E\left[\exp\left\{ r \int_0^{t} \widehat Y_i(s) dB_i^S(s)\right\}\right] & = \mathbb E\left[\exp\left\{ r \int_0^{t} \widehat Y_i(s) dB_i^S(s) - \frac{r^2}{2} \int_0^{t} \widehat Y_i^2(s) ds\right\} \exp\left\{ \frac{r^2}{2} \int_0^{t} \widehat Y_i^2(s) ds \right\}\right]
        \\
        &\le C \mathbb E \left[\exp\left\{  \int_0^{t} r \widehat Y_i(s) dB_i^S(s) - \frac{r^2}{2} \int_0^{t} \widehat Y_i^2(s) ds\right\}\right] = C < \infty,
    \end{align*}
    which implies the required result.
\end{proof}

\begin{theorem}
    Let Assumptions \ref{assum: assumptions on kernels} and \ref{assum: assumption on sandwiched drift} hold and the mesh of the partition $\Delta_N$ satisfy \eqref{eq: condition on the mesh}. Fix $i=1,...,d$, take an arbitrary $\lambda_i \in (0,H_i)$, where $H_i$ is from Assumption \ref{assum: assumptions on kernels}, and fix $r\ge 1$. Then there exists a constant $C>0$ that does not depend on the partition such that
    \[
        \sup_{N\ge 1} \sup_{n = 0,1,...,N}\mathbb E \left[ |S_i(t_n) - \widehat S_i(t_n)|^r \right] \le C\Delta^{r\lambda_i}_N.
    \]
    In particular,
    \[
        \mathbb E \left[ |S_i(T) - \widehat S_i(T)|^r \right] \le C\Delta^{r\lambda_i}_N.
    \]    
\end{theorem}
\begin{proof}
    For any $n=0,1,...,N$, we use Theorem \ref{th: properties of S}, Proposition \ref{prop: moments of approximated price} and Theorem \ref{th: approximation of log-price} to write:
    \begin{align*}
        \mathbb E \left[|S_i(t_n) - \widehat S_i(t_n)|^r\right] &= \mathbb E \left[|e^{X_i(t_n)} - e^{\widehat X_i(t_n)}|^r\right]
        \\
        &\le \mathbb E \left[\left( e^{X_i(t_n)} + e^{\widehat X_i(t_n)}\right)^r |X_i(t_n) - \widehat X_i(t_n)|^r\right]
        \\
        &\le \left(\mathbb E \left[\left( e^{X_i(t_n)} + e^{\widehat X_i(t_n)}\right)^{2r}\right]\right)^{\frac{1}{2}} \left(\mathbb E \left[|X_i(t_n) - \widehat X_i(t_n)|^{2r}\right]\right)^{\frac{1}{2}}
        \\
        &\le C\left( \mathbb E[ e^{2rX_i(t_n)}] + \mathbb E[ e^{2r\widehat X_i(t_n)}] \right)^{\frac{1}{2}} \left(\mathbb E \left[|X_i(t_n) - \widehat X_i(t_n)|^{2r}\right]\right)^{\frac{1}{2}}
        \\
        &\le C\Delta^{r\lambda_i}_N,
    \end{align*}
    where $C$ does not depend on the partition, which implies the required result.
\end{proof}

\paragraph{Approximation of the martingale density.} In order to get an approximation of option price (i.e. under the change of measure), we need a result concerning approximation of the corresponding density. Let $\mathcal E_T$ be the minimal martingale measure defined by \eqref{eq: minimal martingale measure}. Fix a uniform partition $\{0 = t_0 < t_1 <...< t_N = T\}$, $t_n = \frac{Tn}{N}$, with the mesh satisfying \eqref{eq: condition on the mesh} and put
\begin{equation}\label{eq: approximation of martingale measure}
    \widehat {\mathcal E}_T := \exp\left\{ -\sum_{j=1}^d \sum_{n=0}^{N-1} \left(\sum_{k=1}^d \ell^{(-1)}_{j,k} \frac{\widetilde \mu_k(t_n)}{{\widehat Y}_k(t_n)} \right)(W_j(t_{n+1}) - W_j(t_n)) - \frac{1}{2} \sum_{j=1}^d \frac{T}{N}\sum_{n=0}^{N-1} \left(\sum_{k=1}^d \ell^{(-1)}_{j,k} \frac{\widetilde \mu_k(t_n)}{{\widehat Y}_k(t_n)} \right)^2 \right\}.
\end{equation}

\begin{theorem}\label{th: approximation of martingale measure}
    Let Assumptions \ref{assum: assumptions on kernels} and \ref{assum: assumption on sandwiched drift} hold and the mesh of the partition $\Delta_N$ satisfy \eqref{eq: condition on the mesh}. Let also $\lambda < \min_{i=1,...,d} H_i$ with $H_i$ being from Assumption \ref{assum: assumptions on kernels}. Then, for any $r\ge 1$, there exists a constant $C>0$ that does not depend on the partition such that
    \[
        \mathbb E \left[|\mathcal E_T - \widehat {\mathcal E}_T|^r\right] \le C\Delta_N^{r\lambda}.
    \]
\end{theorem}
\begin{proof}
    The proof is straightforward, so we will provide an outline omitting all the details. 
    
    First, observe that for every $j=1,...,d$, using boundedness of each $\frac{1}{\widehat Y_k}$, $\lambda$-H\"older continuity of each $\widetilde \mu_k$ as well as Theorem \ref{th: drift-implicit Euler scheme}, we can write 
    \begin{align*}
        \mathbb E &\left[ \left|\sum_{n=0}^{N-1} \left(\sum_{k=1}^d \ell^{(-1)}_{j,k} \frac{\widetilde \mu_k(t_n)}{{\widehat Y}_k(t_n)} \right)(W_j(t_{n+1}) - W_j(t_n)) - \int_0^T \left(\sum_{k=1}^d \ell^{(-1)}_{j,k} \frac{\widetilde \mu_k(s)}{Y_k(s)} \right)dW_j(s)\right|^r  \right]
        \\
        & = \mathbb E \left[ \left|\sum_{n=0}^{N-1} \int_{t_n}^{t_{n+1}} \left(\sum_{k=1}^d \ell^{(-1)}_{j,k} \left(\frac{\widetilde \mu_k(t_n)}{{\widehat Y}_k(s)} - \frac{\widetilde \mu_k(s)}{Y_k(s)}\right) \right)dW_j(s) \right|^r  \right]
        \\
        & \le C \mathbb E \left[ \left(\sum_{n=0}^{N-1} \int_{t_n}^{t_{n+1}} \left(\sum_{k=1}^d \ell^{(-1)}_{j,k} \left(\frac{\widetilde \mu_k(t_n)}{{\widehat Y}_k(s)} - \frac{\widetilde \mu_k(s)}{Y_k(s)}\right) \right)^2ds \right)^{\frac{r}{2}}  \right]
        \\
        &\le C \mathbb E \left[ \left(\sum_{n=0}^{N-1} \int_{t_n}^{t_{n+1}} \sum_{k=1}^d \left(\frac{\widetilde \mu_k(t_n)}{{\widehat Y}_k(s)} - \frac{\widetilde \mu_k(s)}{Y_k(s)} \right)^2ds \right)^{\frac{r}{2}}  \right]
        \\
        & \le C \mathbb E \left[ \left(\sum_{n=0}^{N-1} \int_{t_n}^{t_{n+1}} \sum_{k=1}^d \frac{1}{{\widehat Y}^2_k(s)} \left(\widetilde \mu_k(t_n) - \widetilde \mu_k(s) \right)^2ds \right)^{\frac{r}{2}}  \right] 
        \\
        &\quad + C \mathbb E \left[ \left(\sum_{n=0}^{N-1} \int_{t_n}^{t_{n+1}} \sum_{k=1}^d \left( \frac{1}{{\widehat Y}_k(s)} - \frac{1}{Y_k(s)} \right)^2ds \right)^{\frac{r}{2}}  \right]
        \\
        &\le C \mathbb E \left[ \left(\sum_{k=1}^d\sum_{n=0}^{N-1} \int_{t_n}^{t_{n+1}}  
        \left(\widetilde \mu_k(t_n) - \widetilde \mu_k(s) \right)^2ds \right)^{\frac{r}{2}}  \right]  + C \mathbb E \left[ \left(\sum_{k=1}^d \sup_{s\in[0,T]}\left( \frac{1}{{\widehat Y}_k(s)} - \frac{1}{Y_k(s)} \right)^2 \right)^{\frac{r}{2}}  \right]
        \\
        &\le C \mathbb E \left[ \left(\sum_{k=1}^d\sum_{n=0}^{N-1} \int_{t_n}^{t_{n+1}}  
        \left(t_n - s \right)^{2\lambda}ds \right)^{\frac{r}{2}}  \right]  + C \Delta_N^{r\lambda}
        \\
        & \le C \Delta_N^{r\lambda},
    \end{align*}
    where $C$ is, as always, a constant that does not depend on the partition and which may vary from line to line. Similarly, for each $j=1,...,d$,
    \begin{align*}
        \mathbb E &\left[ \left| \sum_{j=1}^d \frac{T}{N}\sum_{n=0}^{N-1} \left(\sum_{k=1}^d \ell^{(-1)}_{j,k} \frac{\widetilde \mu_k(t_n)}{{\widehat Y}_k(t_n)} \right)^2 - \sum_{j=1}^d \int_0^T \left(\sum_{k=1}^d \ell^{(-1)}_{j,k} \frac{\widetilde \mu_k(s)}{{\widehat Y}_k(s)} \right)^2 \right|^r  \right] \le C \Delta_N^{r\lambda},
    \end{align*}
    therefore we conclude that
    \[
        \mathbb E \left[\left| \log\mathcal E_T - \log\widehat{\mathcal E}_T \right|^r\right] \le C \Delta_N^{r\lambda}.
    \]
    Using the arguments similar to Proposition \ref{prop: moments of approximated price}, one can see that for all $r\ge 1$ there exists a constant $C$ that does not depend on the partition such that
    \begin{equation}\label{eq: moments of martingale measure}
        \mathbb E\left[ \mathcal E_T^r \right] + \max_{N\ge 1} \mathbb E\left[ \widehat{\mathcal E}_T^r \right] < C
    \end{equation}
    and hence
    \begin{align*}
        \mathbb E \left[|\mathcal E_T - \widehat {\mathcal E}_T|^r\right] & \le C \mathbb E \left[\left( \mathcal E_T^r + \widehat{\mathcal E}_T^r \right)\left| \log\mathcal E_T - \log\widehat{\mathcal E}_T \right|^r\right] \le C \Delta_N^{r\lambda}.
    \end{align*}
\end{proof}

\subsection{Approximation by discretizing both $Y$ and $S$}\label{subsec: double approximation}

After introducing the numerical schemes for $Y$, $X$ and $S$, we are ready to proceed to main results of this Section. We begin with a theorem that is in the spirit of \cite[Theorem 15]{BMdP2018}.

\begin{theorem}\label{th: approximation using representation 2}
    Let $S = (S_1,...,S_d)$ be given by \eqref{eq: definition of price process}, $X=(X_1,...,X_d)$ be the corresponding log-price process, $X_i = \log S_i$, $f$: $\mathbb R \to \mathbb R$ satisfy Assumption \ref{assum: 1-dim payoff}, $F(s) := \int_0^s f(u)du$ and $\alpha_i > 0$, $i=1,...,d$. Fix a uniform partition $\{0 = t_0 < t_1 <...< t_N = T\}$, $t_n = \frac{Tn}{N}$, with the mesh satisfying \eqref{eq: condition on the mesh} and take $\widehat Y = (\widehat Y_1, ..., \widehat Y_d)$, $\widehat X = (\widehat X_1, ..., \widehat X_d)$ and $\widehat S = (\widehat S_1,...,\widehat S_d)$ as described in Subsection \ref{subsec: numerical scheme}. Finally, let $\lambda < \min_{i=1,...,d} H_i$ with $H_i$ being from Assumption \ref{assum: assumptions on kernels}. Then 
    \begin{itemize}
        \item[1)] for any $r\ge 1$, there exists a constant $C>0$ such that
        \begin{align*}
            \mathbb E & \Bigg[ \Bigg|\frac{F\left(\sum_{i=1}^d \alpha_i S_i(T)\right)}{\alpha_1 S_1(T)} \left( 1 + \frac{ \int_0^T \frac{1}{Y_1(u)}d\widetilde W_1(u) }{T u_{1,1}}  \right) 
            \\
            &\qquad - \frac{F\left(\sum_{i=1}^d \alpha_i \widehat S_i(T)\right)}{\alpha_1 \widehat S_1(T)} \left( 1 + \frac{ \int_0^T \frac{1}{\widehat Y_1(u)}d\widetilde W_1(u) }{T u_{1,1}}  \right) \Bigg|^r\Bigg] \le C\Delta_N^{r\lambda},
        \end{align*}
        
        \item[2)] there exists a constant $C>0$ such that
        \[
            \left|\mathbb E \left[f\left(\sum_{i=1}^d \alpha_i S_i(T)\right)\right] - \mathbb E \left[\frac{F\left(\sum_{i=1}^d \alpha_i \widehat S_i(T)\right)}{\alpha_1 \widehat S_1(T)} \left( 1 + \frac{ \int_0^T \frac{1}{\widehat Y_1(u)}d\widetilde W_1(u) }{T u_{1,1}} \right)\right]\right|  \le C\Delta_N^{\lambda}.
        \]
    \end{itemize}
\end{theorem}
\begin{proof}
    By Theorem \ref{th: representation d-dim}, it is enough to prove item 1. We provide the proof only for $d=1$ (in this case we can put $\alpha_1 = 1$ without any loss in generality). For the general case of $d>1$, the arguments are the same but require a more notation-heavy presentation. For reader's convenience, we split the proof into three steps. 
    
    \noindent\textbf{Step 1.} We begin with noting that, by the Burkholder-Davis-Gundy inequality and Theorem \ref{th: drift-implicit Euler scheme}, 
     \begin{align*}
        \mathbb E\left[\left|\int_0^T \left(\frac{1}{ Y(u)} - \frac{1}{\widehat Y(u)} \right) d\widetilde W_1(u)\right|^r\right] &\le C\mathbb E \left[ \left(\int_0^T \left(\frac{1}{ Y(u)} - \frac{1}{\widehat Y(u)} \right)^2 du \right)^{\frac{r}{2}} \right]
        \\
        &\le C \Delta_N^{\lambda r}.
    \end{align*}
    
    \noindent\textbf{Step 2.} Next, let us show that for any $r\ge 1$
    \begin{equation}\label{proofeq: Step 2 of approximation theorem}
        \mathbb E \left[\left|\frac{F(S(T))}{S(T)} - \frac{F(\widehat S(T))}{\widehat S(T)}\right|^r\right] \le C\Delta_N^{r\lambda}.
    \end{equation}
    It is clear that
    \begin{equation}\label{proofeq: Step 2 main idea}
    \begin{aligned}
        \mathbb E \left[\left|\frac{F(S(T))}{S(T)} - \frac{F(\widehat S(T))}{\widehat S(T)}\right|^r\right] & \le C \mathbb E \left[ \left|\frac{F(S(T))}{S(T)} - \frac{F(S(T))}{\widehat S(T)}\right|^r \right] 
        \\
        &\quad + C \mathbb E \left[ \left|\frac{F(S(T))}{\widehat S(T)} - \frac{F(\widehat S(T))}{ \widehat S(T)}\right|^r \right].
    \end{aligned}    
    \end{equation}
    Now we estimate both terms in the right-hand side separately. Observe that $F$ has polynomial growth and thus, by Theorem \ref{th: properties of S}, $\mathbb E[|F(S(T))|^{2r}] < \infty$. Therefore
    \begin{align*}
        \mathbb E \left[ |F(S(T))|^r \left|\frac{1}{S(T)} - \frac{1}{\widehat S(T)}\right|^r \right] &\le \left( \mathbb E\left[ |F(S(T))|^{2r} \right] \mathbb E \left[ \left|\frac{1}{S(T)} - \frac{1}{\widehat S(T)}\right|^{2r} \right]\right)^{\frac{1}{2}}
        \\
        &\le C \left( \mathbb E \left[ \left|\frac{1}{S(T)} - \frac{1}{\widehat S(T)}\right|^{2r} \right]\right)^{\frac{1}{2}}. 
    \end{align*}
    Moreover, observe that, by Theorem \ref{th: properties of S} and Proposition \ref{prop: moments of approximated price}, there exists constant $C>0$ that does not depend on the partition such that
    \[
        \mathbb E \left[ e^{-4r(X(T) + \widehat X(T))} (e^{X(T)} + e^{\widehat X(T)})^{4r} \right] < C,
    \]
    thus, using Theorem \ref{th: approximation of log-price} the inequality $|e^x - e^y| \le (e^x + e^y)|x-y|$, $x$, $y\in\mathbb R$, we can write
    \begin{align*}
        \mathbb E &\left[ \left|\frac{1}{S(T)} - \frac{1}{\widehat S(T)}\right|^{2r} \right] = \mathbb E \left[ e^{-2r(X(T) + \widehat X(T))} \left| e^{X(T)} - e^{\widehat X(T)} \right|^{2r} \right]
        \\
        &\le \mathbb E \left[ e^{-2r(X(T) + \widehat X(T))} (e^{X(T)} + e^{\widehat X(T)})^{2r}\left| X(T) - \widehat X(T) \right|^{2r} \right]
        \\
        &\le \left(\mathbb E \left[ e^{-4r(X(T) + \widehat X(T))} (e^{X(T)} + e^{\widehat X(T)})^{4r} \right]\right)^{\frac{1}{2}} \left( \mathbb E \left[ \left| X(T) - \widehat X(T) \right|^{4r} \right] \right)^{\frac{1}{2}}
        \\
        &\le C \Delta_N^{2r\lambda},
    \end{align*}
    where $C>0$ is again a constant that does not depend on the partition. Therefore
    \begin{equation}\label{proofeq: estimation of the first term}
        \mathbb E \left[ |F(S(T))|^r \left|\frac{1}{S(T)} - \frac{1}{\widehat S(T)}\right|^r \right] \le C\Delta^{r\lambda}.
    \end{equation}
    Next, observe that for any $s_1, s_2 > 0$
    \begin{align*}
        |F(s_1) - F(s_2)| &= \left|\int_{s_1\wedge s_2}^{s_1\vee s_2} f(u) du\right| \le C(1+ s_1^q + s_2^q)|s_1 - s_2|,
    \end{align*}
    so
    \begin{align*}
        \mathbb E &\left[ \frac{1}{\widehat S^r(T)}\left| F(S(T))- F(\widehat S(T))\right|^r \right] 
        \\
        &\le \left(\mathbb E \left[ \frac{1}{\widehat S^{2r}(T)} \right] \right)^{\frac{1}{2}}\left( \mathbb E \left[ \left| F(S(T))- F(\widehat S(T))\right|^{2r} \right] \right)^{\frac{1}{2}}
        \\
        &\le C \left( \mathbb E \left[ \left| F(S(T))- F(\widehat S(T))\right|^{2r} \right] \right)^{\frac{1}{2}}
        \\
        &\le C \left( \mathbb E \left[ |(1+ S^q(T) + \widehat S^q(T))|^{2r} \left| S(T)- \widehat S(T)\right|^{2r} \right] \right)^{\frac{1}{2}}
        \\
        &\le C \left( \mathbb E \left[ |(1+ S^q(T) + \widehat S^q(T))|^{4r}\right] \right)^{\frac{1}{4}} \left(\mathbb E\left[\left| S(T)- \widehat S(T)\right|^{4r}\right] \right)^{\frac{1}{4}}
        \\
        &\le C\left(\mathbb E\left[\left| S(T)- \widehat S(T)\right|^{4r}\right] \right)^{\frac{1}{4}},
    \end{align*}
    where we used Theorem \ref{th: properties of S} together with Proposition \ref{prop: moments of approximated price} to estimate $\mathbb E \left[ \widehat S^{-2r}(T) \right]$ together with $\mathbb E \left[ |(1+ S^q(T) + \widehat S^q(T))|^{4r}\right]$.
    
    Finally, using the same argument as in the proof of \eqref{proofeq: estimation of the first term}, we can obtain that
    \[
        \left(\mathbb E\left[\left| S(T)- \widehat S(T)\right|^{4r}\right] \right)^{\frac{1}{4}} \le C\Delta_N^{r\lambda},
    \]
    and thus
    \begin{align}\label{proofeq: estimation of the second term}
        \mathbb E &\left[ \frac{1}{\widehat S^r(T)}\left| F(S(T))- F(\widehat S(T))\right|^r \right] \le C\Delta_N^{r\lambda}.
    \end{align}
    
    Now \eqref{proofeq: Step 2 of approximation theorem} follows from \eqref{proofeq: Step 2 main idea}, \eqref{proofeq: estimation of the first term} and \eqref{proofeq: estimation of the second term}.
    
    \noindent\textbf{Step 3.} Now we proceed to the claim of the theorem. Since
    \[
        \mathbb E\left[\left| \frac{F( S(T))}{ S(T)} \right|^{2r}\right] < \infty
    \]
    due to the polynomial growth of the mapping $x\mapsto \frac{F(x)}{x}$, we can deduce by  Step 1, Step 2 and Theorem \ref{th: representation} that
    \begin{align*}
        &\mathbb E \left[ \left|\frac{F( S(T))}{ S(T)}\left( 1 +\frac{\int_0^T \frac{1}{ Y(u)} d\widetilde W_1(u)}{Tu_{1,1}} \right) -     \frac{F(\widehat S(T))}{\widehat S(T)}\left( 1 +\frac{\int_0^T \frac{1}{\widehat Y(u)} d\widetilde W_1(u)}{Tu_{1,1}} \right)\right|^r\right] 
        \\
        &\quad\le C\mathbb E\left[ \left| \frac{F( S(T))}{ S(T)} - \frac{F( \widehat S(T))}{ \widehat S(T)} \right|^r\right]
        \\
        &\qquad + C \mathbb E\left[\left| \frac{F( S(T))}{ S(T)} \int_0^T \left(\frac{1}{ Y(u)} - \frac{1}{ \widehat Y(u)}\right) d\widetilde W_1(u)   \right|^r\right]
        \\
        &\qquad + C\mathbb E\left[\left| \int_0^T \frac{1}{ \widehat Y(u)} d\widetilde W_1(u) \left( \frac{F( S(T))}{ S(T)} - \frac{F( \widehat S(T))}{ \widehat S(T)} \right) \right|^r\right]
        \\
        &\quad \le C\Delta^{r\lambda}_N + C \left( \mathbb E\left[\left| \frac{F( S(T))}{ S(T)} \right|^{2r}\right]    \right)^{\frac{1}{2}} \left( \mathbb E \left[\left|\int_0^T \left(\frac{1}{ Y(u)} - \frac{1}{ \widehat Y(u)}\right) d\widetilde W_1(u)\right|^{2r}\right]\right)^{\frac{1}{2}}
        \\
        &\qquad + C \left( \mathbb E\left[\left| \int_0^T \frac{1}{ \widehat Y(u)} d\widetilde W_1(u) \right|^{2r}\right]    \right)^{\frac{1}{2}} \left( \mathbb E \left[\left|\frac{F( S(T))}{ S(T)} - \frac{F( \widehat S(T))}{ \widehat S(T)}\right|^{2r}\right]\right)^{\frac{1}{2}}
        \\
        &\quad \le C\Delta_N^{r\lambda},
    \end{align*}
    which ends the proof.
\end{proof}

Theorem \ref{th: approximation using representation 2} concerns representations \eqref{eq: representation 2, 1-dim} and \eqref{eq: representation 2, d-dim}. However, for the one-dimensional case, we also have the representation \eqref{eq: representation 1, 1-dim} which can also be used to approximate the expected payoff. We give the corresponding result below.

\begin{theorem}\label{th: approximation using representation 1}
    Let $S$ be given by \eqref{eq: 1-dim sandwich transformed}, $X$ be the corresponding log-price process, $X = \log S$, $f$: $\mathbb R \to \mathbb R$ satisfy Assumption \ref{assum: 1-dim payoff}, $F(s) := \int_0^s f(u)du$, $g(x) := f(e^x)$, $G(x) := F(1) + \int_0^x g(v) dv$. Fix a uniform partition $\{0 = t_0 < t_1 <...< t_N = T\}$, $t_n = \frac{Tn}{N}$, with the mesh satisfying \eqref{eq: condition on the mesh} and take $\widehat Y$, $\widehat X$ and $\widehat S$ as described in Subsection \ref{subsec: numerical scheme}. Finally, let $\lambda < H$ with $H:=H_1$ being from Assumption \ref{assum: assumptions on kernels}. Then
    \begin{itemize}
        \item[1)] for any $r\ge 1$, there exists a constant $C>0$ such that
        \[
            \mathbb E \left[\left|\frac{1}{Tu_{1,1}}  G( X(T)) \int_0^T \frac{1}{ Y(u)} d\widetilde W_1(u) - \frac{1}{Tu_{1,1}}  G(\widehat X(T)) \int_0^T \frac{1}{\widehat Y(u)} d\widetilde W_1(u) \right|^r\right]  \le C\Delta_N^{r\lambda},
        \]
        \item[2)] there exists a constant $C>0$ such that
        \[
            \left|\mathbb E \left[f(S(T))\right] - \mathbb E\left[ \frac{1}{Tu_{1,1}}  G(\widehat X(T)) \int_0^T \frac{1}{\widehat Y(u)} d\widetilde W_1(u) \right]\right|  \le C\Delta_N^{\lambda}.
        \]
    \end{itemize}
\end{theorem}
\begin{proof}
    Once again, item 1 implies item 2 by Theorem \ref{th: representation} and we sketch the proof of the former omitting the details. First, observe that for any $x_1$, $x_2 \in\mathbb R$
    \begin{align*}
        |G(x_1) - G(x_2)| &= \left|\int_{x_1\wedge x_2}^{x_1\vee x_2} g(v)dv\right| \le \sup_{v\in [x_1\wedge x_2, x_1\vee x_2]} |g(v)| |x_1-x_2|
        \\
        &= \sup_{v\in [x_1\wedge x_2, x_1\vee x_2]} |f(e^v)| |x_1-x_2| \le c_f(1+e^{q (x_1\vee x_2)})|x_1-x_2|
        \\
        &\le c_f(1+e^{q x_1} + e^{q x_2})|x_1-x_2|
    \end{align*}
    and thus, using the same machinery as in the proof of Theorem \ref{th: approximation using representation 2}, we can write
    \begin{align*}
        \mathbb E\left[\left|G( X(T)) -   G(\widehat X(T)) \right|^{r} \right] &\le C\mathbb E\left[\left|(1+S^q(T) + \widehat S^q(T) )|X(T)-\widehat X(T)| \right|^{r} \right] 
        \\
        &\le C\Delta_N^{r\lambda}.
    \end{align*}
    Due to existence of exponential moments of $X(T)$ as well as polynomial growth of $f$, it is also clear that for any $r\ge 1$
    \[
        \mathbb E \left[\left|G( X(T))\right|^{r}\right] < \infty.
    \]
    Finally, as noted in Step 1 of the proof of Theorem \ref{th: approximation using representation 2},
    \[
        \mathbb E\left[\left|\int_0^T \left(\frac{1}{ Y(u)} - \frac{1}{\widehat Y(u)}\right) d\widetilde W_1(u)\right|^{r}\right] \le C\Delta_N^{r\lambda}.
    \]
    Therefore
    \begin{align*}
        &E \left[\left|G( X(T)) \int_0^T \frac{1}{ Y(u)} d\widetilde W_1(u) -   G(\widehat X(T)) \int_0^T \frac{1}{\widehat Y(u)} d\widetilde W_1(u) \right|^r\right]
        \\
        &\quad \le C E \left[\left|G( X(T)) \left(\int_0^T \frac{1}{ Y(u)} d\widetilde W_1(u) - \int_0^T \frac{1}{\widehat Y(u)} d\widetilde W_1(u)\right) \right|^r\right]
        \\
        &\qquad + C E \left[\left|\left(G( X(T)) -   G(\widehat X(T))\right) \int_0^T \frac{1}{\widehat Y(u)} d\widetilde W_1(u) \right|^r\right]
        \\
        &\quad \le C\left( \mathbb E \left[\left|G( X(T))\right|^{2r}\right]\right)^{\frac{1}{2}} \left(\mathbb E\left[\left|\int_0^T \left(\frac{1}{ Y(u)} - \frac{1}{\widehat Y(u)}\right) d\widetilde W_1(u)\right|^{2r}\right]\right)^{\frac{1}{2}}
        \\
        &\qquad + C\left( E \left[\left|\int_0^T \frac{1}{\widehat Y(u)} d\widetilde W_1(u)\right|^{2r}\right]\right)^{\frac{1}{2}} \left(\mathbb E\left[\left|G( X(T)) -   G(\widehat X(T))\right|^{2r}\right]\right)^{\frac{1}{2}}
        \\
        &\quad \le C\Delta_N^{r\lambda}.
    \end{align*}
\end{proof}

We conclude the Subsection with an approximation result in line with Theorem \ref{th: approximation using representation 2} involving a change of measure. We omit the proof since it is straightforward and utilizes the same machinery as the one of Theorem \ref{th: approximation using representation 2} (one just has to additionally use Theorem \ref{th: approximation of martingale measure} in order to estimate expectations of the type $\mathbb E \left[|\mathcal E_T - \widehat{\mathcal E}_T|^r\right]$).

\begin{theorem}\label{th: approximation under change of measure}
    Let $S = (S_1,...,S_d)$ be given by \eqref{eq: definition of price process}, $X=(X_1,...,X_d)$ be the corresponding log-price process, $X_i = \log S_i$, $f$: $\mathbb R \to \mathbb R$ satisfy Assumption \ref{assum: 1-dim payoff}, $F(s) := \int_0^s f(u)du$ and $\alpha_i > 0$, $i=1,...,d$, and $\mathcal E_T$ be the minimal martingale measure defined by \eqref{eq: minimal martingale measure}. Fix a uniform partition $\{0 = t_0 < t_1 <...< t_N = T\}$, $t_n = \frac{Tn}{N}$, with the mesh satisfying \eqref{eq: condition on the mesh} and take $\widehat Y = (\widehat Y_1, ..., \widehat Y_d)$, $\widehat X = (\widehat X_1, ..., \widehat X_d)$, $\widehat S = (\widehat S_1,...,\widehat S_d)$ and $\widehat{\mathcal E}_T$ as described in Subsection \ref{subsec: numerical scheme}. Finally, let $\lambda < \min_{i=1,...,d} H_i$ with $H_i$ being from Assumption \ref{assum: assumptions on kernels}. Then 
    \begin{itemize}
        \item[1)] for any $r\ge 1$, there exists a constant $C>0$ such that
        \begin{align*}
            \mathbb E &\Bigg[ \Bigg|\mathcal E_T  \frac{F\left(\sum_{i=1}^d \alpha_i S_i(T)\right)}{\alpha_1 S_1(T)} \left(1 + \frac{ \int_0^T\frac{1}{Y_1(v)} d\widetilde W_1(v) - \sum_{k=1}^d\int_0^T  \beta_k \frac{\widetilde \mu_k(v)}{Y_k(v)Y_1(v)}dv}{Tu_{1,1}}\right)
            \\
            & \quad - \widehat{\mathcal E}_T  \frac{F\left(\sum_{i=1}^d \alpha_i \widehat S_i(T)\right)}{\alpha_1 \widehat S_1(T)} \left(1 + \frac{  \int_0^T\frac{1}{\widehat Y_1(v)} d\widetilde W_1(v) - \sum_{k=1}^d\int_0^T  \beta_k \frac{\widetilde \mu_k(v)}{\widehat Y_k(v)\widehat Y_1(v)}dv   }{Tu_{1,1}}\right)\Bigg|^r\Bigg] \le C\Delta_N^{r\lambda},
        \end{align*}
        
        \item[2)] there exists a constant $C>0$ such that
        \begin{align*}
            \Bigg|\mathbb E &\left[\mathcal E_Tf\left( \sum_{i=1}^d \alpha_i S_i(T)\right)\right] 
            \\
            &\quad - \mathbb E \left[\widehat{\mathcal E}_T  \frac{F\left(\sum_{i=1}^d \alpha_i \widehat S_i(T)\right)}{\alpha_1 \widehat S_1(T)} \left(1 + \frac{  \int_0^T\frac{1}{\widehat Y_1(v)} d\widetilde W_1(v) - \sum_{k=1}^d\int_0^T  \beta_k \frac{\widetilde \mu_k(v)}{\widehat Y_k(v)\widehat Y_1(v)}dv      }{Tu_{1,1}}\right)\right]\Bigg| \le C\Delta_N^{\lambda}.
        \end{align*}
    \end{itemize}
\end{theorem}

\subsection{Approximation via conditional Gaussianity for one-di\-mensional model}\label{subsec: Gaussian approximation}

Assume now that $d=1$ and the price $S$ is given by \eqref{eq: 1-dim sandwich transformed}. The core idea of Subsection \ref{subsec: double approximation} was to plug in the simulated paths $\widehat Y$, $\widehat S$ and $\widehat X$ into representations \eqref{eq: representation 1, 1-dim} or \eqref{eq: representation 2, 1-dim}. However, this requires simulation of \textit{both} volatility and the price which inevitably leads to higher variance in Monte Carlo estimations. However, it is possible to avoid simulation of the price and in this Subsection we describe the corresponding approach.

First note that the random vector $\left(\widehat X(T), \int_0^T \frac{1}{\widehat Y(u)} d\widetilde W_1(u)\right)^T$ is conditionally Gaussian with respect to the $\sigma$-field $\widetilde {\mathcal F}_T^2 := \sigma\{\widetilde W_2(t),~t\in[0,T]\}$. The next lemma presents the parameters of the joint conditional distribution of the random vector $\left(\widehat X(T), \int_0^T \frac{1}{\widehat Y(u)} d\widetilde W_1(u)\right)$.
\begin{lemma}\label{lemma: cond dist 1 dim}
    The random vector $\left(\widehat X(T), \int_0^T \frac{1}{\widehat Y(u)} d\widetilde W_1(u)\right)^T$ is conditionally Gaussian with respect to the $\sigma$-field $\widetilde {\mathcal F}_T^2 := \sigma\{\widetilde W_2(t),~t\in[0,T]\}$ with conditional mean
    \begin{equation}
    \begin{gathered}
        \left(\begin{matrix}
            X(0) + \frac{T}{N} \sum_{k=0}^{N-1} \left( \mu(t_k) - \frac{1}{2} \widehat Y^2(t_k) \right) + u_{1,2} \sum_{k=0}^{N-1} \widehat Y(t_k) (\widetilde W_2(t_{k+1}) - \widetilde W_2(t_{k}))
            \\
            0
        \end{matrix}\right) =: \left(\begin{matrix}
            m(\widehat Y)
            \\
            0
        \end{matrix}\right)
    \end{gathered}
    \end{equation}
    and conditional covariance matrix
    \begin{align*}
        \widehat{\mathcal C} &=\left(\begin{matrix}
            u_{1,1}^2\frac{T}{N} \sum_{k=0}^{N-1} \widehat Y^2(t_k) &  u_{1,1}T
            \\
            u_{1,1}T & \frac{T}{N} \sum_{k=0}^{N-1} \widehat Y^{-2}(t_k)
        \end{matrix}\right)
        =: \left(\begin{matrix}
            u_{1,1}^2 V^2_1 & u_{1,1}T
            \\
            u_{1,1}T & V^2_2
        \end{matrix}\right),
    \end{align*}
    where $V^2_1 := \frac{T}{N} \sum_{k=0}^{N-1} \widehat Y^2(t_k)$ and $V^2_2 := \frac{T}{N} \sum_{k=0}^{N-1} \widehat Y^{-2}(t_k)$. Moreover, if additionally
    \[
        \det\widehat{\mathcal C}  = u_{1,1}^2 \left(V^2_1 V^2_2 - T^2\right) > 0 
    \]
    with probability 1, then the random vector $\left(\widehat X(T), \int_0^T \frac{1}{\widehat Y(u)} d\widetilde W_1(u)\right)^T$ has a conditional density of the form
    \begin{align*}
        \phi(x,y) = \frac{1}{2\pi \sqrt{\det\widehat{\mathcal C} }}\exp\left\{ -\frac{V^2_2 (x - m(\widehat Y))^2 + u_{1,1}^2 V^2_1y^2 - 2Tu_{1,1}(x - m(\widehat Y))y }{2\det\widehat{\mathcal C}} \right\}.
    \end{align*}
\end{lemma}

\begin{proof}
    Recall that the random variable $\widehat X(T)$ can be represented in the form
    \begin{align*}
        \widehat X(T) & = X(0) + \frac{T}{N}\sum_{k=0}^{N-1} \left( \mu(t_k) - \frac{ \widehat Y^2(t_k)}{2} \right) +  u_{1,1} \sum_{k=0}^{N-1} \widehat Y(t_k) \left(\widetilde W_1(t_{k+1}) - \widetilde W_1(t_k)\right) 
        \\
        &\quad+  u_{1,2} \sum_{k=0}^{N-1} \widehat Y(t_k) \left(\widetilde W_2(t_{k+1}) - \widetilde W_2(t_k)\right) .
    \end{align*}
    Now the explicit form of conditional mean, conditional covariance matrix and conditional density (provided that the conditional covariance matrix is non-degenerate) can be computed directly. 
\end{proof}

\begin{theorem}\label{th: expectation in terms of conditional density}
    Let $\det\widehat{\mathcal C} > 0$ with probability 1, where $\widehat {\mathcal C}$ is the conditional covariance matrix from Lemma \ref{lemma: cond dist 1 dim}. Then
    \begin{equation}\label{eq: expectation in terms of conditional density}
    \begin{aligned}
        &\frac{1}{Tu_{1,1}}\mathbb E \left[  G(\widehat X(T)) \int_0^T \frac{1}{\widehat Y(u)} d\widetilde W_1(u)\right]
        \\
        &\quad = \frac{1}{\sqrt{2\pi}u_{1,1}^2} \int_{\mathbb R} G(x) \mathbb E\left[\frac{x - m(\widehat Y)}{ V_1^3 } \exp\left\{ -\frac{ (x - m(\widehat Y))^2}{2 u_{1,1}^2 V_1^2} \right\} \right]dx.
    \end{aligned}
    \end{equation}
\end{theorem}
\begin{proof}
    Using Lemma \ref{lemma: cond dist 1 dim}, we can write
    \begin{equation*}
    \begin{aligned}
        \mathbb E \left[  G(\widehat X(T)) \int_0^T \frac{1}{\widehat Y(u)} d\widetilde W_1(u)\right] &=\mathbb E \left[  \mathbb E\left[G(\widehat X(T)) \int_0^T \frac{1}{\widehat Y(u)} d\widetilde W_1(u)~\bigg|~\widetilde{\mathcal F}^2_T\right]\right]
        \\
        &=\mathbb E\left[ \int_{\mathbb R^2} G(x)y \phi(x,y)dxdy \right]
        \\
        &= \int_{\mathbb R} G(x) \mathbb E\left[\int_{\mathbb R}y \phi(x,y)dy \right]dx. 
    \end{aligned}
    \end{equation*}
    The integral $\int_{\mathbb R}y \phi(x,y)dy$ can be simplified. Namely, denote for simplicity $\widetilde x := x - m(\widehat Y)$. Then
    \begin{align*}
        \int_{\mathbb R}y \phi(x,y)dy &= \frac{1}{2\pi \sqrt{\det\widehat{\mathcal C} }}\int_{\mathbb R} y \exp\left\{ -\frac{V^2_2 \widetilde x^2 + u_{1,1}^2 V^2_1y^2 - 2u_{1,1} T\widetilde x y }{2\det\widehat{\mathcal C}} \right\}dy
        \\
        & = \frac{1}{2\pi \sqrt{\det\widehat{\mathcal C} }}\int_{\mathbb R} y \exp\left\{ -\frac{1 }{2\det\widehat{\mathcal C}} \left( \left(u_{1,1} V_1 y - \frac{T\widetilde x}{V_1}  \right)^2 + V^2_2 \widetilde x^2 - \frac{T^2\widetilde x^2}{V^2_1} \right) \right\}dy
        \\
        & = \frac{1}{2\pi \sqrt{\det\widehat{\mathcal C} }}\exp\left\{ -\frac{ (V_1^2 V_2^2 - T^2) \widetilde x^2}{2  V_1^2\det{\widehat {\mathcal C}}} \right\} \int_{\mathbb R} y \exp\left\{ -\frac{\left(y - \frac{T\widetilde x}{u_{1,1} V_1^2}\right)^2}{2 \frac{\det\widehat C}{u_{1,1}^2 V_1^2}} \right\} dy
        \\
        & = \frac{1}{2\pi \sqrt{\det\widehat{\mathcal C} }}\exp\left\{ -\frac{ \widetilde x^2}{2 u_{1,1}^2 V_1^2} \right\} \sqrt{2\pi} \frac{T\widetilde x\sqrt{\det\widehat {\mathcal C}}}{u_{1,1}^2 V_1^3}
        \\
        & = \frac{T\widetilde x}{\sqrt{2\pi}u_{1,1}^2 V_1^3} \exp\left\{ -\frac{ \widetilde x^2}{2 u_{1,1}^2 V_1^2} \right\}.
    \end{align*}
    This yields the result.
\end{proof}

Equation \eqref{eq: expectation in terms of conditional density} from Theorem \ref{th: expectation in terms of conditional density} essentially relies on the condition $\det\widehat {\mathcal C} > 0$ with probability 1. A simple sufficient condition is provided by the following proposition.

\begin{proposition}
    Let $Z(t) = \int_0^t \mathcal K(t,s)d\widetilde W_2(s)$, $t\in[0,T]$, and $\{0=t_0 < t_1 <...< t_N = T\}$ be a uniform partition of $[0,T]$ with the mesh $\Delta_N := \frac{T}{N}$ having the following property: $c\Delta_N < 1$, where $c$ is an upper bound for $\frac{\partial b}{\partial y}$ from Assumption \ref{assum: assumption on sandwiched drift}(iv).
    
    If the random vector
    \[
        (Z(t_1) - Z(t_0), Z(t_2) - Z(t_1),..., Z(t_N) - Z(t_{N-1}))^T 
    \]
    has absolutely continuous distribution w.r.t. the Lebesgue measure on $\mathbb R^{N}$, then $\det \widehat{\mathcal C} > 0$ with probability 1 and \eqref{eq: expectation in terms of conditional density} holds.
\end{proposition}
\begin{proof}
    Note that $\det\widehat{\mathcal C} \ge 0$ a.s., because $\widehat {\mathcal C}$ is a conditional covariance matrix, i.e. we need to check that $\det\widehat{\mathcal C} \ne 0$ with probability 1. Taking into account that
    \[
        \det\widehat{\mathcal C} \ne 0 \quad \Longleftrightarrow \quad \frac{1}{N^2} \sum_{k=0}^{N-1}\sum_{l=0}^{N-1} \widehat Y^2(t_k) \widehat Y^{-2}(t_l) \ne 1,
    \]
    it is sufficient to check that the random vector $(\widehat Y(t_1),...,\widehat Y(t_N))^T$ has density with respect to the Lebesgue measure. Denote
    \[
        b_n(y) := y - b(t_n, y)\Delta_N, \quad n=1,...,N,
    \]
    and observe that for any $n = 1,...,N$ is a 1-to-1 mapping of the interval $(\varphi(t_n), \psi(t_n))$ onto $\mathbb R$ with $b_n' > 0$. By \eqref{eq: definition of backward Euler scheme}, for any $n = 1,...,N$
    \[
        \widehat Y(t_n) = b_n^{-1}\left(\widehat Y(t_{n-1}) + Z(t_n) - Z(t_{n-1})\right)
    \]
    and thus
    \begin{equation}\label{proofeq: BE scheme as transformation of the noise}
    \begin{aligned}
        \left(\begin{matrix}
            \widehat Y(t_1)
            \\
            \widehat Y(t_2)
            \\
            \vdots
            \\
            \widehat Y(t_N)
        \end{matrix}\right) &=  \left(\begin{matrix}
            b_1^{-1}\left(Y(0) + Z(t_1) - Z(t_{0})\right)
            \\
            b_2^{-1}\left(b_1^{-1}\left(Y(0) + Z(t_1) - Z(t_{0})\right) + Z(t_2) - Z(t_{1})\right)
            \\
            \vdots
            \\
            b_N^{-1}\left(b_{N-1}^{-1}\left(\cdots)\right) + Z(t_{N}) - Z(t_{N-1})\right)
        \end{matrix}\right) 
        \\
        & =:  I(Z(t_1) - Z(t_0), Z(t_2) - Z(t_1),..., Z(t_N) - Z(t_{N-1})).
        \end{aligned}
    \end{equation}
    It is straightforward to check that function $I$: $\mathbb R^N \to \prod_{n=1}^N\left(\varphi(t_n), \psi(t_n)\right)$ from \eqref{proofeq: BE scheme as transformation of the noise} is a bijection and its Jacobian $\det J_{I}(\vec z)$ is strictly positive for any $\vec z \in \mathbb R^N$. This implies that the random vector $(\widehat Y(t_1),...,\widehat Y(t_N))^T$ has density, i.e. is absolutely continuous w.r.t. the Lebesgue measure, which ends the proof.
\end{proof}

\section{Simulations}\label{sec: simulations}

Let us illustrate the algorithms described in Section \ref{sec: numerics} with numerical simulations. For simplicity, we assume that 
\begin{itemize}
    \item $T=1$, $d=1$;
    
    \item an explicit shape of the SVV model is
    \begin{equation}\label{eq: SVV for Malliavin simulations}
    \begin{aligned}
        S(t) &= 1 + \int_0^t Y(s) S(s) \left(\sqrt{0.75}dB_1(s)-0.5 dB_2(s)\right),
        \\
        Y(t) &= 0.5 + \int_0^t \left(\frac{0.005}{(Y(s) - 0.05)^5} - \frac{0.005}{(1 - Y(s))^5} + 0.05(0.5 - Y(s))\right)ds + 0.3\int_0^t \mathcal K(t,s) dB_2(s),
    \end{aligned}
    \end{equation}
    where $B_1$, $B_2$ are independent Brownian motions and $\mathcal K(t,s) := \frac{1}{\Gamma(H+1/2)}(t-s)^{H - \frac{1}{2}} \mathbbm 1_{s<t}$ with $H=0.2$;
    
    \item the instantaneous interest rate $\nu(t) \equiv 0$;
    
    \item the option payoff function is a discontinuous function given by
    \[
        f(s) = \mathbbm 1_{(0.5, \infty)}(s) + \mathbbm 1_{(1, \infty)}(s) + \mathbbm 1_{(1.5, \infty)}(s).
    \]
\end{itemize}

\begin{remark}
    Comparing \eqref{eq: SVV for Malliavin simulations} with the representation \eqref{eq: 1-dim sandwich transformed} from Section \ref{sec: Malliavin duality approach}, we have that $B_1 = \widetilde W_1$, $B_2 = \widetilde W_2$, $u_{1,1} = \sqrt{0.75}$ and $u_{1,2} = -0.5$.
\end{remark}

We want to compute the value $\mathbb E [f(S(T))]$ and compare the following three numerical methods:
\begin{itemize}
    \item[1)] \textbf{Standard Monte Carlo (Std. MC)}: we simply simulate 10000 independent realizations of 
    \begin{equation}\label{eq: sample of std method}
        f(\widehat S(T))
    \end{equation}
    and average over them;

    \item[2)] \textbf{Malliavin Monte Carlo with double discretization (MMCDD)}: we simulate 10000 independent realizations of
    \begin{equation}\label{eq: sample of method 1}
        \frac{F\left(\widehat S(T)\right)}{\widehat S(T)} \left( 1 + \int_0^T \frac{1}{\widehat Y(u)}d W_1(u) \right)
    \end{equation}
    and then average over them;

    \item[3)] \textbf{Malliavin Monte Carlo via conditional Gaussianity (MMCCG)}: we exploit the method described in Subsection \ref{subsec: Gaussian approximation} and average over 10000 independent realizations of
    \begin{equation}\label{eq: sample of method 2}
        \frac{1}{\sqrt{2\pi}} \int_{\mathbb R} G(x) \left(\frac{x - m(\widehat Y)}{ V_1^3 }\right) \exp\left\{ -\frac{ (x - m(\widehat Y))^2}{2 V_1^2} \right\}dx.
    \end{equation}
\end{itemize}

For all three methods, we consider the uniform partitions of sizes $N=10$, $100$, $1000$, and $10000$. Table \ref{Tab1} contains information on means (being also the estimators of $\mathbb E [f(S(T))]$) and standard deviations of \eqref{eq: sample of std method}, \eqref{eq: sample of method 1} and \eqref{eq: sample of method 2} for different partition sizes and Hurst indices.

\begin{remark}
    In this type of simulation, there are two sources of error:
    \begin{itemize}
        \item[(I)] the error between $\mathbb E[f(S(T))]$ and
        \begin{itemize}
            \item[1)] $\mathbb E[f(\widehat S(T))]$ for Std. MC,
            \item[2)] $\mathbb E\left[\frac{F\left(\widehat S(T)\right)}{\widehat S(T)} \left( 1 + \int_0^T \frac{1}{\widehat Y_1(u)}d W_1(u) \right)\right]$ for MMCDD and
            \item[3)] $\mathbb E\left[\frac{1}{\sqrt{2\pi}} \int_{\mathbb R} G(x) \left(\frac{x - m(\widehat Y)}{ V_1^3 }\right) \exp\left\{ -\frac{ (x - m(\widehat Y))^2}{2 V_1^2} \right\}dx\right]$ for MMCCG;
        \end{itemize}

        \item[(II)] the error coming from the Monte Carlo estimation of the corresponding expectations.
    \end{itemize}
    The error of type (I) is generally controlled by the size of the partition $N$, and the Malliavin integration-by-parts technique is used specifically to reduce it. The error of type (II) is reduced by increasing the number of Monte Carlo samples. In this context, standard deviations (Sd) of random variables \eqref{eq: sample of std method}, \eqref{eq: sample of method 1} and \eqref{eq: sample of method 2} evaluate how many Monte Carlo simulations are required to achieve a certain level of accuracy of the mean estimation. 
\end{remark}

\begin{table}[h!]
\centering
\caption{Computation of $\mathbb E [f(S(T))]$: standard Monte Carlo vs Malliavin MC with double discretization vs Malliavin MC via conditional Gaussianity. Columns entitled ``Mean'' contain the corresponding estimates of $\mathbb E [f(S(T))]$ whereas the ``Sd'' columns present the standard deviations of \eqref{eq: sample of std method} for Std. MC, \eqref{eq: sample of method 1} for MMCDD and \eqref{eq: sample of method 2} for MMCCG}
\label{Tab1}
\begin{tabular}{|c|cccccc|}
\hline
\multirow{2}{*}{\textbf{Partition size, $N$}} & \multicolumn{2}{c|}{\textbf{Std. MC}}               & \multicolumn{2}{c|}{\textbf{MMCDD}}                 & \multicolumn{2}{c|}{\textbf{MMCCG}} \\ \cline{2-7} 
                                              & \multicolumn{1}{c|}{Mean} & \multicolumn{1}{c|}{Sd} & \multicolumn{1}{c|}{Mean} & \multicolumn{1}{c|}{Sd} & \multicolumn{1}{c|}{Mean}  & Sd     \\ \hline
10                                            & 1.5194                    & 0.9212                  & 1.3961                    & 1.3863                  & 1.4111                     & 0.3205 \\ \cline{1-1}
100                                           & 1.4559                    & 0.9162                  & 1.4152                    & 1.4150                  & 1.4075                     & 0.3168 \\ \cline{1-1}
1000                                          & 1.3963                    & 0.9405                  & 1.4101                    & 1.3846                  & 1.4082                     & 0.3178 \\ \cline{1-1}
10000                                         & 1.4036                    & 0.9315                  & 1.4061                    & 1.4002                  & 1.4084                     & 0.3197 \\ \hline
\end{tabular}
\end{table}

As we can see, the estimators of $\mathbb E [f(S(T))]$ by the standard Monte Carlo for small $N$ differ substantially from the ones for larger $N$ whereas both MMCDD and MMCCG show better stability with respect to the size of the partitions (thanks to the Malliavin integration-by-parts regularization technique). In addition, the MMCCG method showed by far the lowest standard deviations out of all three algorithms: this happens because there is no additional error from simulating $\widehat S(T)$ which results in a far more stable behavior of \eqref{eq: sample of method 2} (c.f. \cite{BMdP2018}).


\bibliographystyle{acm}
\bibliography{biblio.bib}

\begin{thebibliography}{10}

\bibitem{Rough_volatility_literature}
Rough volatility literature.
\newblock \url{https://sites.google.com/site/roughvol/home/risks-1}.
\newblock Accessed: 2022-09-15.

\bibitem{Abi_Jaber_2022}
{\sc Abi~Jaber, E.}
\newblock The characteristic function of {G}aussian stochastic volatility
  models: an analytic expression.
\newblock {\em Finance and stochastics 26}, 4 (2022), 733–769.

\bibitem{Alfonsi_2010}
{\sc Alfonsi, A.}
\newblock High order discretization schemes for the {CIR} process: Application
  to affine term structure and {H}eston models.
\newblock {\em Mathematics of computation 79}, 269 (2010), 209–209.

\bibitem{Alos_Garcia_Lorite_2021}
{\sc Alos, E., and Garcia~Lorite, D.}
\newblock {\em Malliavin calculus in finance: Theory and practice}.
\newblock CRC Press, London, England, 2021.

\bibitem{AN2015}
{\sc Altmayer, M., and Neuenkirch, A.}
\newblock Multilevel {M}onte {C}arlo quadrature of discontinuous payoffs in the
  generalized {H}eston model using {M}alliavin integration by parts.
\newblock {\em {SIAM} Journal on Financial Mathematics 6}, 1 (Jan. 2015),
  22--52.

\bibitem{Alos_Leon_2017}
{\sc Alòs, E., and León, J.~A.}
\newblock On the curvature of the smile in stochastic volatility models.
\newblock {\em SIAM journal on financial mathematics 8}, 1 (2017), 373–399.

\bibitem{Alos_Leon_Vives_2007}
{\sc Alòs, E., León, J.~A., and Vives, J.}
\newblock On the short-time behavior of the implied volatility for
  jump-diffusion models with stochastic volatility.
\newblock {\em Finance and stochastics 11}, 4 (2007), 571–589.

\bibitem{Andersen_Piterbarg_2005}
{\sc Andersen, L. B.~G., and Piterbarg, V.~V.}
\newblock Moment explosions in stochastic volatility models.
\newblock {\em Finance and stochastics 11}, 1 (2006), 29–50.

\bibitem{AndersenBollerslev1997}
{\sc Andersen, T.~G., and Bollerslev, T.}
\newblock Intraday periodicity and volatility persistence in financial markets.
\newblock {\em Journal of empirical finance 4}, 2–3 (1997), 115–158.

\bibitem{AndersenBollerslevDieboldLabys2001}
{\sc Andersen, T.~G., Bollerslev, T., Diebold, F.~X., and Labys, P.}
\newblock The distribution of realized exchange rate volatility.
\newblock {\em Journal of the American Statistical Association 96}, 453 (2001),
  42–55.

\bibitem{Andersen_Davis_Kreiss_Mikosch_2009}
{\sc Andersen, T.~G., Davis, R.~A., Kreiss, J.-P., and Mikosch, T.~V.}
\newblock {\em Handbook of financial time series}, 2009~ed.
\newblock Springer, 2009.

\bibitem{Avellaneda_Levy_Paras_1995}
{\sc Avellaneda, M., Levy, A., and Parás, A.}
\newblock Pricing and hedging derivative securities in markets with uncertain
  volatilities.
\newblock {\em Applied mathematical finance 2}, 2 (1995), 73–88.

\bibitem{Avikainen_2009}
{\sc Avikainen, R.}
\newblock On irregular functionals of {SDE}s and the {E}uler scheme.
\newblock {\em Finance and stochastics 13}, 3 (2009), 381–401.

\bibitem{Ayache_Peng_2012}
{\sc Ayache, A., and Peng, Q.}
\newblock Stochastic volatility and multifractional {B}rownian motion.
\newblock In {\em Stochastic Differential Equations and Processes\/} (2012),
  Springer Berlin Heidelberg, p.~211–237.

\bibitem{ASVY2014}
{\sc Azmoodeh, E., Sottinen, T., Viitasaari, L., and Yazigi, A.}
\newblock Necessary and sufficient conditions for {H}ölder continuity of
  {G}aussian processes.
\newblock {\em Statistics \& Probability Letters 94\/} (2014), 230 -- 235.

\bibitem{Bayer_Friz_Gatheral_2016}
{\sc Bayer, C., Friz, P., and Gatheral, J.}
\newblock Pricing under rough volatility.
\newblock {\em Quantitative finance 16}, 6 (2016), 887–904.

\bibitem{Benaim_Friz_2009}
{\sc Benaim, S., and Friz, P.}
\newblock Regular variation and smile asymptotics.
\newblock {\em Mathematical Finance 19}, 1 (2009), 1–12.

\bibitem{Bennedsen_Lunde_Pakkanen_2022}
{\sc Bennedsen, M., Lunde, A., and Pakkanen, M.~S.}
\newblock Decoupling the short- and long-term behavior of stochastic
  volatility.
\newblock {\em Journal of financial econometrics 20}, 5 (2022), 961–1006.

\bibitem{BMdP2018}
{\sc Bezborodov, V., Persio, L.~D., and Mishura, Y.}
\newblock Option pricing with fractional stochastic volatility and
  discontinuous payoff function of polynomial growth.
\newblock {\em Methodology and Computing in Applied Probability 21}, 1 (Aug.
  2018), 331--366.

\bibitem{BGP2000}
{\sc Biagini, F., Guasoni, P., and Pratelli, M.}
\newblock Mean-variance hedging for stochastic volatility models.
\newblock {\em Mathematical Finance 10}, 2 (Apr. 2000), 109--123.

\bibitem{Black_1976}
{\sc Black, F.}
\newblock Studies of stock price volatility changes.
\newblock In {\em Proceedings of the Business and Economics Section of the
  American Statistical Association\/} (1976), p.~177–181.

\bibitem{BollerslevMikkelsen1996}
{\sc Bollerslev, T., and Mikkelsen, H.~O.}
\newblock Modeling and pricing long memory in stock market volatility.
\newblock {\em Journal of Econometrics 73}, 1 (1996), 151–184.

\bibitem{Cheung_Ng_1992}
{\sc Cheung, Y.-W., and Ng, L.~K.}
\newblock Stock price dynamics and firm size: {A}n empirical investigation.
\newblock {\em The journal of finance 47}, 5 (1992), 1985--1997.

\bibitem{Christie_1982}
{\sc Christie, A.}
\newblock The stochastic behavior of common stock variances: {V}alue, leverage
  and interest rate effects.
\newblock {\em Journal of financial economics 10}, 4 (1982), 407–432.

\bibitem{ChronopoulouViens2012}
{\sc Chronopoulou, A., and Viens, F.~G.}
\newblock Estimation and pricing under long-memory stochastic volatility.
\newblock {\em Annals of finance 8}, 2–3 (2012), 379–403.

\bibitem{CE2015}
{\sc Cohen, S., and Elliott, R.~J.}
\newblock {\em Stochastic calculus and applications}.
\newblock Birkh\"auser, New York, 2015.

\bibitem{ComteCoutinRenault2012}
{\sc Comte, F., Coutin, L., and Renault, E.}
\newblock Affine fractional stochastic volatility models.
\newblock {\em Annals of finance 8}, 2–3 (2012), 337–378.

\bibitem{ComteRenault1998}
{\sc Comte, F., and Renault, E.}
\newblock Long memory in continuous-time stochastic volatility models.
\newblock {\em Mathematical Finance 8}, 4 (1998), 291–323.

\bibitem{Cont2001}
{\sc Cont, R.}
\newblock Empirical properties of asset returns: stylized facts and statistical
  issues.
\newblock {\em Quantitative finance 1}, 2 (2001), 223–236.

\bibitem{Cont2006}
{\sc Cont, R.}
\newblock Volatility clustering in financial markets: Empirical facts and
  agent-based models.
\newblock In {\em Long Memory in Economics\/} (2006), Springer Berlin
  Heidelberg, p.~289–309.

\bibitem{Cont_2010}
{\sc Cont, R.}
\newblock {\em Encyclopedia of Quantitative Finance}.
\newblock Wiley, 2010.

\bibitem{ContFonseca2002-1}
{\sc Cont, R., and da~Fonseca, J.}
\newblock Deformation of implied volatility surfaces: an empirical analysis.
\newblock In {\em Empirical Science of Financial Fluctuations\/} (2002),
  Springer Japan, p.~230–239.

\bibitem{ContFonseca2002-2}
{\sc Cont, R., and da~Fonseca, J.}
\newblock Dynamics of implied volatility surfaces.
\newblock {\em Quantitative finance 2}, 1 (2002), 45–60.

\bibitem{CLV_2014}
{\sc Corlay, S., Lebovits, J., and Lévy~Véhel, J.}
\newblock Multifractional stochastic volatility models.
\newblock {\em Mathematical Finance. An International Journal of Mathematics,
  Statistics and Financial Economics 24}, 2 (2014), 364–402.

\bibitem{Das_Sundaram_1999}
{\sc Das, S.~R., and Sundaram, R.~K.}
\newblock Of smiles and smirks: {A} term structure perspective.
\newblock {\em Journal of financial and quantitative analysis 34}, 2 (1999),
  211.

\bibitem{Delemotte_Marco_Segonne_2023}
{\sc Delemotte, J., De~Marco, S., and Segonne, F.}
\newblock Yet another analysis of the {SP500} at-the-money skew: {C}rossover of
  different power-law behaviours.
\newblock {\em SSRN Electronic Journal\/} (2023).

\bibitem{Derman_Miller_Park_2016}
{\sc Derman, E., Miller, M.~B., and Park, D.}
\newblock {\em The volatility smile}.
\newblock John Wiley \& Sons, Nashville, TN, 2016.

\bibitem{Di_Nunno_Kubilius_Mishura_Yurchenko-Tytarenko_2023}
{\sc Di~Nunno, G., Kubilius, K., Mishura, Y., and Yurchenko-Tytarenko, A.}
\newblock From constant to rough: {A} survey of continuous volatility modeling.
\newblock {\em Mathematics 11}, 19 (2023), 4201.

\bibitem{DNMYT2022}
{\sc Di~Nunno, G., Mishura, Y., and Yurchenko-Tytarenko, A.}
\newblock Drift-implicit {E}uler scheme for sandwiched processes driven by
  {H}ölder noises.
\newblock {\em Numerical algorithms 93}, 2 (2023), 459–491.

\bibitem{DNMYT2020}
{\sc Di~Nunno, G., Mishura, Y., and Yurchenko-Tytarenko, A.}
\newblock Sandwiched {SDE}s with unbounded drift driven by {H}ölder noises.
\newblock {\em Advances in applied probability 55}, 3 (2023), 927–964.

\bibitem{Di_Nunno_Yurchenko-Tytarenko_hedging_2022}
{\sc Di~Nunno, G., and Yurchenko-Tytarenko, A.}
\newblock Sandwiched {V}olterra {V}olatility model: {M}arkovian approximations
  and hedging.
\newblock {\em ArXiv:2209.13054\/} (2022).

\bibitem{DN_YT_power_law_2023}
{\sc Di~Nunno, G., and Yurchenko-Tytarenko, A.}
\newblock Power law in {S}andwiched {V}olterra {V}olatility model.
\newblock {\em Modern Stochastics Theory and Applications\/} (2024), 169–194.

\bibitem{DingGrangerEngle1993}
{\sc Ding, Z., Granger, C. W.~J., and Engle, R.~F.}
\newblock A long memory property of stock market returns and a new model.
\newblock {\em Journal of empirical finance 1}, 1 (1993), 83–106.

\bibitem{Duffee_1995}
{\sc Duffee, G.}
\newblock Stock returns and volatility a firm-level analysis.
\newblock {\em Journal of financial economics 37}, 3 (1995), 399–420.

\bibitem{Duong_Swanson_2011}
{\sc Duong, D., and Swanson, N.~R.}
\newblock Volatility in discrete and continuous-time models: {A} survey with
  new evidence on large and small jumps.
\newblock In {\em Missing Data Methods: Time-Series Methods and Applications\/}
  (2011), Emerald Group Publishing Limited, p.~179–233.

\bibitem{EuchGatheralRosenbaum2018}
{\sc El~Euch, O., Gatheral, J., and Rosenbaum, M.}
\newblock Roughening {H}eston.
\newblock {\em SSRN Electronic Journal\/} (2018).

\bibitem{Euch_Rosenbaum_2018}
{\sc El~Euch, O., and Rosenbaum, M.}
\newblock Perfect hedging in rough {H}eston models.
\newblock {\em The annals of applied probability: an official journal of the
  Institute of Mathematical Statistics 28}, 6 (2018).

\bibitem{El_Euch_Rosenbaum_2019}
{\sc El~Euch, O., and Rosenbaum, M.}
\newblock The characteristic function of rough {H}eston models.
\newblock {\em Mathematical Finance 29}, 1 (2019), 3–38.

\bibitem{Karoui_Jeanblanc-Picque_Shreve_1998}
{\sc El~Karoui, N., Jeanblanc-Picquè, M., and Shreve, S.~E.}
\newblock Robustness of the {B}lack and {S}choles formula.
\newblock {\em Mathematical Finance 8}, 2 (1998), 93–126.

\bibitem{Fouque2000}
{\sc Fouque, J.-P., Papanicolaou, G., and Sircar, K.~R.}
\newblock {\em Derivatives in financial markets with stochastic volatility}.
\newblock Cambridge University Press, 2000.

\bibitem{Fouque_Papanicolaou_Sircar_Solna_2003}
{\sc Fouque, J.~P., Papanicolaou, G., Sircar, R., and Sølna, K.}
\newblock Singular perturbations in option pricing.
\newblock {\em SIAM journal on applied mathematics 63}, 5 (2003), 1648–1665.

\bibitem{Fouque_Papanicolaou_Sircar_Solna_2004}
{\sc Fouque, J.-P., Papanicolaou, G., Sircar, R., and Sølna, K.}
\newblock Maturity cycles in implied volatility.
\newblock {\em Finance and stochastics 8}, 4 (2004).

\bibitem{FV2010}
{\sc Friz, P.~K., and Victoir, N.~B.}
\newblock {\em Multidimensional {S}tochastic {P}rocesses as {R}ough {P}aths:
  {T}heory and {A}pplications}.
\newblock Cambridge Studies in Advanced Mathematics. Cambridge University
  Press, 2010.

\bibitem{Fukasawa_2021}
{\sc Fukasawa, M.}
\newblock Volatility has to be rough.
\newblock {\em Quantitative finance 21}, 1 (2021), 1–8.

\bibitem{Fukasawa_Gatheral_2022}
{\sc Fukasawa, M., and Gatheral, J.}
\newblock A rough {SABR} formula.
\newblock {\em Frontiers of Mathematical Finance 1}, 1 (2022), 81.

\bibitem{Funahashi_Kijima_2017}
{\sc Funahashi, H., and Kijima, M.}
\newblock Does the {H}urst index matter for option prices under fractional
  volatility?
\newblock {\em Annals of finance 13}, 1 (2017), 55–74.

\bibitem{Funahashi_Kijima_2017-1}
{\sc Funahashi, H., and Kijima, M.}
\newblock A solution to the time-scale fractional puzzle in the implied
  volatility.
\newblock {\em Fractal and fractional 1}, 1 (2017), 14.

\bibitem{Gatheral2006}
{\sc Gatheral, J.}
\newblock {\em The volatility surface: A practitioner`s guide}.
\newblock John Wiley \& Sons, 2006.

\bibitem{GatheralJaissonRosenbaum2018}
{\sc Gatheral, J., Jaisson, T., and Rosenbaum, M.}
\newblock Volatility is rough.
\newblock {\em Quantitative finance 18}, 6 (2018), 933–949.

\bibitem{Gatheral_Jusselin_Rosenbaum_2020}
{\sc Gatheral, J., Jusselin, P., and Rosenbaum, M.}
\newblock The quadratic rough {H}eston model and the joint {S\&P} 500/{VIX}
  smile calibration problem.
\newblock {\em arXiv [q-fin.MF]\/} (2020).

\bibitem{Gerhold_Gerstenecker_Pinter_2019}
{\sc Gerhold, S., Gerstenecker, C., and Pinter, A.}
\newblock Moment explosions in the rough {H}eston model.
\newblock {\em Decisions in Economics and Finance 42}, 2 (2019), 575–608.

\bibitem{Harms_Stefanovits_2019}
{\sc Harms, P., and Stefanovits, D.}
\newblock Affine representations of fractional processes with applications in
  mathematical finance.
\newblock {\em Stochastic processes and their applications 129}, 4 (2019),
  1185–1228.

\bibitem{Heston1993}
{\sc Heston, S.~L.}
\newblock A closed-form solution for options with stochastic volatility with
  applications to bond and currency options.
\newblock {\em The review of financial studies 6}, 2 (1993), 327–343.

\bibitem{Hu2008}
{\sc Hu, Y., Nualart, D., and Song, X.}
\newblock A singular stochastic differential equation driven by fractional
  {B}rownian motion.
\newblock {\em Statistics {\&} Probability Letters 78}, 14 (Oct. 2008),
  2075--2085.

\bibitem{HullWhite1987}
{\sc Hull, J., and White, A.}
\newblock The pricing of options on assets with stochastic volatilities.
\newblock {\em The journal of finance 42}, 2 (1987), 281–300.

\bibitem{Veraar_Weis_2016}
{\sc Hyt\"onen, T., van Neerven, J., Veraar, M., and Weis, L.}
\newblock {\em Analysis in {B}anach spaces. {V}olume {I}: {M}artingales and
  {L}ittlewood-{P}aley theory}, 1~ed.
\newblock Springer International Publishing, Cham, Switzerland, 2016.

\bibitem{Jaber_Illand_Shaun_Li_2022}
{\sc Jaber, E.~A., Illand, C., and Li, S.}
\newblock Joint {SPX}-{VIX} calibration with {G}aussian polynomial volatility
  models: deep pricing with quantization hints.
\newblock {\em arXiv [q-fin.MF]\/} (2022).

\bibitem{Kahl2008}
{\sc Kahl, C.}
\newblock {\em Modelling and simulation of stochastic volatility in finance}.
\newblock Dissertation.com, 2008.

\bibitem{Kallsen_Muhle-Karbe_2010}
{\sc Kallsen, J., and Muhle-Karbe, J.}
\newblock Utility maximization in affine stochastic volatility models.
\newblock {\em International journal of theoretical and applied finance 13}, 03
  (2010), 459–477.

\bibitem{Keller-Ressel_2011}
{\sc Keller-Ressel, M.}
\newblock Moment explosions and long-term behavior of affine stochastic
  volatility models.
\newblock {\em Mathematical Finance 21}, 1 (2011), 73–98.

\bibitem{KnightSatchell2007}
{\sc Knight, J., and Satchell, S.}
\newblock {\em Forecasting volatility in the financial markets}, 3~ed.
\newblock Butterworth-Heinemann, 2007.

\bibitem{Lee_2004}
{\sc Lee, R.~W.}
\newblock The moment formula for implied volatility at extreme strikes.
\newblock {\em Mathematical Finance 14}, 3 (2004), 469–480.

\bibitem{Merino_Pospisil_Sobotka_Sottinen_Vives_2021}
{\sc Merino, R., Pospíšil, J., Sobotka, T., Sottinen, T., and Vives, J.}
\newblock Decomposition formula for rough {V}olterra stochastic volatility
  models.
\newblock {\em International journal of theoretical and applied finance 24}, 02
  (2021), 2150008.

\bibitem{MYuT2018}
{\sc Mishura, Y., and Yurchenko-Tytarenko, A.}
\newblock Fractional {C}ox–{I}ngersoll–{R}oss process with non-zero
  ``mean''.
\newblock {\em Modern Stochastics: Theory and Applications 5}, 1 (2018),
  99--111.

\bibitem{MYuT2019}
{\sc Mishura, Y., and Yurchenko-Tytarenko, A.}
\newblock Fractional {C}ox–{I}ngersoll–{R}oss process with small {H}urst
  indices.
\newblock {\em Modern Stochastics: Theory and Applications 6}, 1 (2018),
  13--39.

\bibitem{MYT2020}
{\sc Mishura, Y., and Yurchenko-Tytarenko, A.}
\newblock Approximating expected value of an option with non-{L}ipschitz payoff
  in fractional {H}eston-type model.
\newblock {\em International Journal of Theoretical and Applied Finance 23}, 05
  (July 2020), 2050031.

\bibitem{Nualart2006}
{\sc Nualart, D.}
\newblock {\em The {M}alliavin calculus and related topics}.
\newblock Springer, 2014.

\bibitem{OcK1991}
{\sc Ocone, D.~L., and Karatzas, I.}
\newblock A generalized {C}lark representation formula, with application to
  optimal portfolios.
\newblock {\em Stochastics and stochastics reports 34}, 3–4 (1991), 187--220.

\bibitem{Peltier_Vehel_1995}
{\sc Peltier, R.-F., and Vehel, J.~L.}
\newblock {\em Multifractional Brownian motion: Definition and preliminary
  results}.
\newblock PhD thesis, INRIA, 1995.

\bibitem{Renault_Touzi_1996}
{\sc Renault, E., and Touzi, N.}
\newblock Option hedging and implied volatilities in a stochastic volatility
  model.
\newblock {\em Mathematical Finance 6}, 3 (1996), 279–302.

\bibitem{Rosenbaum_2008}
{\sc Rosenbaum, M.}
\newblock Estimation of the volatility persistence in a discretely observed
  diffusion model.
\newblock {\em Stochastic processes and their applications 118}, 8 (2008),
  1434–1462.

\bibitem{Rosenbaum_Zhang_2021}
{\sc Rosenbaum, M., and Zhang, J.}
\newblock Deep calibration of the quadratic rough {H}eston model.
\newblock {\em arXiv [q-fin.CP]\/} (2021).

\bibitem{Schweizer1992}
{\sc Schweizer, M.}
\newblock Martingale densities for general asset prices.
\newblock {\em Journal of Mathematical Economics 21}, 4 (1992), 363--378.

\bibitem{Schweizer_1995}
{\sc Schweizer, M.}
\newblock On the minimal martingale measure and the {F}öllmer-{S}chweizer
  decomposition.
\newblock {\em Stochastic analysis and applications 13}, 5 (1995), 573–599.

\bibitem{Shephard_Andersen_2009}
{\sc Shephard, N., and Andersen, T.~G.}
\newblock Stochastic volatility: Origins and overview.
\newblock In {\em Handbook of Financial Time Series\/} (Berlin, Heidelberg,
  2009), Springer Berlin Heidelberg, p.~233–254.

\bibitem{Skiadopoulos_2001}
{\sc Skiadopoulos, G.}
\newblock Volatility smile consistent option models: {A} survey.
\newblock {\em International journal of theoretical and applied finance 04}, 03
  (2001), 403–437.

\bibitem{Sugita1985}
{\sc Sugita, H.}
\newblock On a characterization of the {S}obolev spaces over an abstract
  {W}iener space.
\newblock {\em Journal of Mathematics of Kyoto University 25}, 4 (1985),
  717--725.

\bibitem{Vr2003}
{\sc Vrabie, I.~I.}
\newblock {\em $C_0$-Semigroups and Applications}.
\newblock North-Holland Mathematics Studies. JAI Press, 2003.

\bibitem{Wiggins_1987}
{\sc Wiggins, J.~B.}
\newblock Option values under stochastic volatility: Theory and empirical
  estimates.
\newblock {\em Journal of financial economics 19}, 2 (1987), 351–372.

\bibitem{Young1936}
{\sc Young, L.~C.}
\newblock An inequality of the {H}\"{o}lder type, connected with {S}tieltjes
  integration.
\newblock {\em Acta Mathematica 67\/} (1936), 251--282.

\end{thebibliography}

\end{document}